\documentclass{article}
\newtheorem{Fact}{Fact}

\newtheorem{Lemma}{Lemma}
\newtheorem{Theorem}{Theorem}

\newcommand{\scriptbf}[1]{\mbox{\scriptsize${\bf #1}$}}

\newenvironment{proof}{{\bf Proof.}}{\rule{.3em}{2ex} \hspace*{\fill} \newline}

\newcommand{\type}{ {\bf :} }
\newcommand{\objtype}{ \epsilon }
\newcommand{\morph}{{:}}
\newcommand{\memberof}{{\in}}
\newcommand{\define}{\stackrel{{\rm df}}{=}}
\newcommand{\orthogonal}{{\perp}}
\newcommand{\diff}{\: \dot{-} \:}

\newcommand{\tprod}{\otimes}
\newcommand{\bsum}{\oplus}

\newcommand{\tsum}{\nabla}
\newcommand{\bprod}{\triangle}

\newcommand{\imply}{{\Rightarrow}}

\newcommand{\tensorimplysource}
   {\! \mbox{\hspace{.21em}--\hspace{-.21em}} \backslash }
\newcommand{\tensorimplytarget}
   { / \mbox{\hspace{-.23em}--\hspace{.23em}} \!}

\newcommand{\logequiv}
   { \vdash \mbox{\hspace{-1.1em}} \dashv }

\newcommand{\tneg}[3]{\neg_{#1 #3} #2}

\newcommand{\vecneg}[1]{\vec{\neg}{#1}}
\newcommand{\tnegtneg}[3]{\neg \! \neg_{#1 #3} #2}
\newcommand{\boolneg}[1]{{\sim}{#1}}
\newcommand{\total}[1]{{#1}^{\dagger}}
\newcommand{\abov}[1]{{\uparrow}(#1)}
\newcommand{\below}[1]{{\downarrow}(#1)}

\newcommand{\singleton}[1]{\{#1\}}
\newcommand{\orthog}[1]{{{\perp}_{\!#1}}}

\newcommand{\coorthog}[1]{{\perp}^{\rm co}_{{#1}}}

\newcommand{\principalideal}[1]{{\downarrow}(#1)}

\newcommand{\power}[1]{{\cal P}(#1)}
\newcommand{\closurepower}[1]{{\cal K}(#1)}
\newcommand{\mat}[1]{{\cal M}(#1)}

\newcommand{\matropo}[1]{{\cal M}_{\cal T}(#1)}
\newcommand{\distrib}[1]{{\cal D}(#1)}
\newcommand{\Boolean}[1]{{\cal B}(#1)}
\newcommand{\Heyting}[1]{{\cal H}(#1)}
\newcommand{\Hoare}[1]{{\cal H}(#1)}

\newcommand{\zenter}[1]{{\cal Z}(#1)}
\newcommand{\yoneda}[3]{{\rm y}_{#1}^{#3}{#2}}

\newcommand{\comonoid}[3]{\Omega_{#1}^{#3}(#2)}

\newcommand{\filtersource}[1]{{\cal F}_0(#1)}
\newcommand{\filtertarget}[1]{{\cal F}_1(#1)}

\newcommand{\interior}[1]{{#1}^{\circ}}
\newcommand{\closure}[1]{{#1}^{\bullet}}

\newcommand{\involution}[1]{{#1}^{{\propto}}}

\newcommand{\product}[2]{\mbox{$ #1 {\times} #2 $}}
\newcommand{\triproduct}[3]{\mbox{$ #1 {\times} #2 {\times} #3 $}}

\newcommand{\pair}[2]{\langle #1,#2 \rangle}
\newcommand{\relcopair}[4]{[{#1},{#2}]^{#3}_{#4}}
\newcommand{\relpair}[4]{\langle{#1},{#2}\rangle^{#3}_{#4}}
\newcommand{\domain}[1]{\partial_0{#1}}

\newcommand{\source}[2]{{]{#1}[_{#2}}}
\newcommand{\target}[2]{{\rangle{#1}\langle^{#2}}}
\newcommand{\cotuple}[2]{{[{#1}]_{#2}}}
\newcommand{\tuple}[2]{{\langle{#1}\rangle^{#2}}}
\newcommand{\decomposition}[3]{({#1})_{#2}^{#3}}
\newcommand{\triple}[3]{\mbox{$ \langle #1,#2,#3 \rangle $}}
\newcommand{\quadruple}[4]{\mbox{$ \langle #1,#2,#3,#4 \rangle $}}

\newcommand{\term}[3]{#1 \stackrel{#2}{\rightharpoondown} #3}

\newcommand{\opterm}[3]{#1 \stackrel{#2}{\leftharpoondown} #3}
\newcommand{\rulezero}[2]
   {\begin{array}{cc}
       #1
       & (\mbox{#2})
    \end{array}}
\newcommand{\ruledoublezero}[3]
           {\begin{array}{cc}
               \begin{array}{c}
                  #1 \\
                  #2
               \end{array}
               & (\mbox{#3})
            \end{array}}
\newcommand{\ruleone}[3]
   {\begin{array}{cc}
       \begin{array}{c}
          #1 \\ \hline
          #2
       \end{array}
       & \left( \mbox{#3} \right)
    \end{array}}
\newcommand{\ruletwo}[4]
   {\begin{array}{cc}
       \begin{array}{ccc}
          #1 && #2 \\ \hline
          \multicolumn{3}{c}{#3}
       \end{array}
       & (\mbox{#4})
    \end{array}}

\newcommand{\yiya}[1]
   {\setlength{\unitlength}{#1pt}
    \begin{picture}(1,1)(-.5,-.5)
       \put(0,0){\circle{1}}
       \put(.25,0){\oval(.5,.5)[b]}
       \put(-.25,0){\oval(.5,.5)[t]}
    \end{picture}
    \setlength{\unitlength}{1pt}}
\newcommand{\yinyang}{\yiya{8}}

\title{Dialectical Logic: \\
       the Process Calculus \\
        \mbox{\footnotesize
              $\pi \alpha \lambda \iota \nu \tau o \nu o \zeta$      $\alpha \rho \mu o \nu \iota \eta$
            - $\pi \alpha \lambda \iota \nu \tau \rho o \pi o \zeta$ $\alpha \rho \mu o \nu \iota \eta$
              \normalsize}}
\author{Robert E. Kent}
\date{}

\begin{document}
   \maketitle

\begin{abstract}
{\em Dialectical logic\/} is the logic of dialectical processes. 
The goal of dialectical logic is to reveal the dynamical notions
inherent in logical computational systems.
The fundamental notions of {\em proposition\/} and {\em truth-value\/} in standard logic
are subsumed by the notions of {\em process\/} and {\em flow\/} in dialectical logic.
Standard logic motivates the core sequential aspect of dialectical logic.
Horn-clause logic requires types and nonsymmetry and also motivates the parallel aspect of dialectical logic.
The process logics of Milner and Hoare reveal the internal/external aspects of dialectical logic.
The sequential internal aspect of dialectical logic 
should be viewed as a typed or distributed version of Girard's linear logic with nonsymmetric tensor.
The simplest version of dialectical logic is inherently intuitionistic.
However,
by following Glivenko's approach in standard logic using double negation closure,
we can define a classical version of dialectical logic.
\end{abstract}

   \tableofcontents

\section*{Introduction}
\addcontentsline{toc}{section}{Introduction}
Abstract objective knowledge,
such as general science and philosophy,
originated in the fifth and sixth centuries B.C.
in the thought, teachings and writings of the preSocratic Greek philosophers.
The aim of the preSocratics was
to give a nonmythological account of the origin of the world ({\bf kosmos}),
and to rationally explain its motion.
By far the most common explanation given by the preSocratics
for the origin and motion of the kosmos
was in terms of pairs of opposing tendencies,
such as 
the {\em hot\/} and the {\em cold\/},
the {\em wet\/} and the {\em dry\/}, 
{\em love\/} and {\em strife\/}, etc.
In fact,
the notion of complementary pairs of opposing tendencies
has occurred throughout the history of ideas. 
Ancient examples of opposing tendencies occur
not only in preSocratic Greek philosophy,
but also in naturalistic Chinese philosophy,
as the dualistic concept of {\em yin\/} and {\em yang\/};
and in Indian Hindu philosophy,
as {\em Brahma\/} the creator and {\em Shiva\/} the destroyer with {\em Vishnu\/} the preserver.

For the preSocratics,
who were postmythological but prelogical,
the components of such opposed pairs were neither properties nor objects,
but motive forces.
The dynamics in this world-view is obvious.
Unfortunately,
much of this dynamical world-view was lost to the history of ideas
when logic was conceived as a study of static notions.
A central theme of this paper is that much of this dynamical world-view needs
to be re-revealed, re-developed, and extended,
in order to comprehend modern logical computational systems.
A modern theory of dialectics offers the appropriate conceptual framework for doing this;
it takes the notion of opposing tendencies as its central concept,
and calls it {\em dialectical contradiction\/}.
This modern dialectical theory still retains the motive force interpretation
for the components (aspects) of dialectical contradictions:
dialectical contradictions specify {\em dialectical motion\/},
where motion is not mere physical motion,
but any change whatsoever;
motion is synonymous with transformation.
The distinction between the concepts of dialectical contradiction and dialectical motion,
two fundamental notions of dialectics,
is itself dialectical,
the {\em potential\/} aspect and the {\em actual\/} aspect.
These two concepts occur in ancient and modern interpretations of the fragments of Heraclitus,
the most dialectically oriented preSocratic \cite{Hussey},
and are contained here in the subtitle:
\mbox{\small
      $\pi \alpha \lambda \iota \nu \tau o \nu o \zeta$ $\alpha \rho \mu o \nu \iota \eta$
    - $\pi \alpha \lambda \iota \nu \tau \rho o \pi o \zeta$ $\alpha \rho \mu o \nu \iota \eta$
      \normalsize};
({\bf palintonos harmonie} - {\bf palintropos harmonie});
(crudely) polar tension structure - polar turning structure;
the ``tension'' interpretation - the ``oscillation'' interpretation, of Heraclitus;
or for us, dialectical contradiction - dialectical motion.

The history of dialectics is replete with intuitively suggestive,
but ill-defined and non-rigorous, ideas and examples \cite{Bernow,Piccone}.
If the dialectical point of view is to be useful as a human conceptual structure,
its objective aspect must have a rigorous foundation.
The notion of dialectical contradiction is monistically objectified [Lawvere]
by the mathematical idea of {\em adjunction\/}.
Since adjoint pairs are (one of) the most important concepts of category theory,
this point-of-view is summarized by the statement:
{\sc Category Theory $\equiv$ Objective Dialectics}.
The notion of dialectical contradiction is pluralistically objectified [Kent87]
by the mathematical idea of {\em dialectical base\/}.
In objective dialectics,
since dialectical contradictions are represented by adjunctions,
systems of dialectical contradictions are represented by diagrams in the unbounded category ({\bf to apeiron}) whose morphisms are adjoint pairs.
Such a diagram,
whose component preorders usually have certain completeness properties,
is called a {\em dialectical base} of preorders.
From a static, non-dynamic, non-dialectical point-of-view,
this has also been called an indexed preorder \cite{Hyland}.
Within the notion of dialectical contradiction
the distinction between the concepts of adjunctions and dialectical bases is dialectical,
the {\em one-many\/} dialectic.

The notion of dialectical motion can be specified [Kent87]
by the mathematical idea of {\em dialectical system\/},
or parallel pair of distributed {\em terms\/}.
Dialectical systems have the following essential aspects:
{\bf [ancient]}
they are based upon contradictions or opposing tendencies;
they define motion, flow or development;
{\bf [modern]}
they contain internally interacting and combining objects or entities in dialectical motion; and
they specify the reproduction or renewal of such entities,
where reproduction is equilibrium of dialectical motion.
Dialectical systems are the ``motors of nature'' specifying the dialectical motion of structured entities,
and a dialectical base provides the ``motive power'' for this motion.
The notion of dialectical motion can be realized
by the mathematical idea of {\em dialectical flow\/},
which is the oscillation (alternation-composition) of inverse flow along one term and direct flow along the other term.
Direct and inverse flow are suitably generalized {\em Kan extensions\/} which make use of a dialectical base.
Dialectical systems specify dialectical flow,
and dialectical flow is the realization of dialectical systems;
the {\em specification-realization\/} dialectic.

It has been known for some time now \cite{Lawvere} that logic is dialectical in nature,
but the full force of its dialectical structure has only recently \cite{Girard,Kent88} been discussed.
Dialectical ideas,
not only come chronologically and historically before logical ideas,
but also come conceptually before them as well.
The theory and practice of computer science and dynamic systems contain many dialectical contradictions.
Two of the most important of these,
the {\em flow\/} dialectic and the {\em constraint\/} dialectic,
constitute the proper study of dialectical logic \cite{Kent88};
whereas a third,
the {\em part}-{\em whole\/} dialectic,
is important in its standard aspect \cite{Kent89}.
Dialectical logic is the logic of dialectical processes.
It invests the dynamical view of systems theory with the fundamental ideas of category theory;
but in turn, it gives these categorical notions that dynamical view.
Dialectical logic provides a unified semantics for
both the object paradigm and the process paradigm of programming-in-the-large.
By subsuming process logic \cite{Milner,Hoare78} along with clause logic,
it allows the specification of strongly-typed parallel logic programs.
In dialectical logic
aspects of the process paradigm are modelled as a flow dialectic,
whereas aspects of the object paradigm are modelled as a constraint dialectic orthogonal to flow.
The flow (or {\em product-implication\/}) dialectic is the internal aspect of dialectical logic,
whereas the constraint dialectic is its external aspect.

Dialectical logic is based upon the two interdependent concepts
of structure and dialecticality.
Dialecticality is built out of the aspects
of dialectical tension and dialectical flow,
as mentioned above.
Structure is concentrated in the compositionality of monoids and comonoids
(this includes the grand unification principle \cite{Manes} that ``composition determines semantics''),
and in the type-summability of orthogonal terms (the object calculus, discussed below).
Structure occurs peripherally in the interactions of limits, the combinations of colimits,
and the reproduction of fixpoints.
The structurality of limits and colimits,
being special Kan extensions,
has obvious dialecticality.
This is but one indication of the interdependence of structure and dialectics;
other indications are the simple facts that
monoids have associated adjoint pairs,
and adjoint pairs compose into monoids and comonoids.
Parsimonious use of
(1) abstract monoidal concepts
for modelling ``construction'', ``composition'' and ``interaction'',
along with
(2) adjointness notions for modelling ``dialectical flow'' (such as ``predicate transformation'')
has great potential in the computational and system sciences.

Dialectical logic is an extension of standard logic.
The extension of propositional calculus is called the {\em process calculus\/};
the extension of predicate calculus is called the {\em object calculus\/}.
In this paper we are mainly concerned with the process calculus;
its intuitionistic and classical semantics,
and its classical axiomatics.
In a succeeding paper \cite{Kent88} we will be concerned chiefly with the object calculus.
In order that readers may begin to explore the fascinating possibilities of dialectics,
I have included in the appendix to this paper an introduction to this object aspect of dialectics.

\section{Preliminaries}

\paragraph{Dialectical Laws.}
The ``laws of dialectics'' are laws of logic.
The most fundamental dialectical law,
the law of {\em the interpenetration of opposites\/},
is represented in general by adjoint pairs of functors or monotonic functions,
and in particular by the flow dialectic (tensor product - tensor implication adjointness).
As a special case of this,
the dialectical law of {\em the negation of the negation\/}
is represented in general as a self-adjoint functor or monotonic function,
and in particular by tensor negation.
Here we discuss the general case.
The paper as a whole is a discussion of the particular case.

Two opposed monotonic functions $\pair{B}{\leq_B} \stackrel{f}{\rightarrow} \pair{A}{\leq_A}$
and $\pair{B}{\leq_B} \stackrel{g}{\leftarrow} \pair{A}{\leq_A}$ between preorders form an {\em adjoint pair\/},
denoted $f \dashv g$,
when they satisfy the equivalence
$f(b) \leq_A a$ iff $b \leq_B g(a)$.
This equivalence can be interpreted as the ``polar-tension structure'' of the preSocratic Greek philosopher Heraclitus \cite{Hussey},
and in Greek is rendered 
$\pi \alpha \lambda \iota \nu \tau o \nu o \zeta$      $\alpha \rho \mu o \nu \iota \eta$.
The fact that $f \dashv g$ is an adjoint pair
is equivalently defined by
the ``unit'' inequality $B \leq f \cdot g$ and the ``counit'' inequality $g \cdot f \leq A$.
The composite monotonic functions
$\pair{B}{\leq_B} \stackrel{f \cdot g}{\rightarrow} \pair{B}{\leq_B}$
and
$\pair{A}{\leq_A} \stackrel{g \cdot f}{\rightarrow} \pair{A}{\leq_A}$
are closure and interior operators, respectively.
A {\em closure operator\/}
$\pair{B}{\leq_B} \stackrel{k}{\rightarrow} \pair{B}{\leq_B}$ 
is a monotonic endofunction
which is ``increasing'' $B \leq k$ and ``idempotent'' $k \cdot k = k$.
Dually,
an {\em interior\/} (or {\em kernel\/}) {\em operator\/}
$\pair{A}{\leq_A} \stackrel{j}{\rightarrow} \pair{A}{\leq_A}$ 
is a monotonic endofunction
which is ``decreasing'' $A \geq j$ and ``idempotent'' $j \cdot j = j$.
An adjoint pair $f \dashv g$ is a {\em reflective pair\/}
when the counit is an equality $g \cdot f = A$,
stating that the interior operator $g \!\cdot\! f$ is an identity.
So an adjoint pair $f \dashv g$ is a reflective pair
iff $f$ is a surjective monotonic function
iff $g$ is an injective monotonic function.
An adjoint pair $f \dashv g$ is a {\em coreflective pair\/}
when the unit is an equality $B = f \cdot g$,
stating that the closure operator $f \!\cdot\! g$ is an identity.
So an adjoint pair $f \dashv g$ is a coreflective pair
iff $f$ is an injective monotonic function
iff $g$ is a surjective monotonic function.

The corestriction
$\pair{B}{\leq_B} \stackrel{\closure{(\,)}_k}{\rightarrow} \pair{k(B)}{\leq_{k(B)}}$ 
of a closure operator $k$
to its image $k(B) \define \{ k(b) \mid b \memberof B \}$ of $k$-{\em closed elements\/} of $B$
forms a reflective pair $\closure{(\,)}_k \dashv {\rm Inc}$
with the inclusion
$\pair{k(B)}{\leq_{k(B)}} \stackrel{{\rm Inc}}{\rightarrow} \pair{B}{\leq_B}$.
The corestriction
$\pair{A}{\leq_A} \stackrel{\interior{(\,)}_j}{\rightarrow} \pair{j(A)}{\leq_{j(A)}}$ 
of an interior operator $j$
to its image $j(A) \define \{ j(a) \mid a \memberof A \}$ of $j$-{\em open elements\/} of $A$
forms a coreflective pair ${\rm Inc} \dashv \interior{(\,)}_j$
with the inclusion
$\pair{j(A)}{\leq_{j(A)}} \stackrel{{\rm Inc}}{\rightarrow} \pair{A}{\leq_A}$.
So for any adjoint pair $f \dashv g$,
the subpreorders
 of $f \!\cdot\! g$-closed elements $\closure{B} \subseteq B$
and $g \!\cdot\! f$-open elements $\interior{A} \subseteq A$
participate themselves in the special adjunctions
$\closure{(\,)} \dashv {\rm Inc}$
and ${\rm Inc} \dashv \interior{(\,)}$
of reflective and coreflective pairs,
respectively.
The restriction of the adjoint pair to closed/open elements forms an inverse pair of monotonic functions,
making $B$-closed elements isomorphic to $A$-open elements
$\closure{B} \cong \interior{A}$.
The adjoint pair, the closed element reflection, the open element coreflection, and the inverse pair,
form a commuting square of dialectical contradictions.
For a reflective pair $f \dashv g$,
all elements of $A$ are open $\interior{A} = A$,
and hence $A$ is isomorphic to the $B$-closed elements $\closure{B} \cong A$. 
Any reflective pair $f \dashv g$ is equivalent to the $\closure{(\,)} \dashv {\rm Inc}_{\closure{B}}$
reflective pair which factors the closure operator $f \!\cdot\! g$ through its image $\closure{B}$.
For a coreflective pair $f \dashv g$,
all elements of $B$ are closed $\closure{B} = B$,
and hence $B$ is isomorphic to the $A$-open elements $B \cong \interior{A}$. 
Any coreflective pair $f \dashv g$ is equivalent to the ${\rm Inc}_{\interior{A}} \dashv \interior{(\,)}$
coreflective pair which factors the interior operator $g \!\cdot\! f$ through its image $\interior{A}$.
So any inverse pair is an adjoint pair with the identity orderings,
and any adjoint pair determines an inverse pair.
Adjointness is a kind of generalized inverseness
(another related kind of generalized inverseness is the notion of orthogonality defined below).

The special case of self-adjointness,
where $f = g^{\rm op}$ and $A = B^{\rm op}$,
defines the notion of ``negation''.
When a monotonic function
$\pair{A}{\leq} \stackrel{f}{\rightarrow} \pair{A}{\leq}^{\rm op}$ 
is self-adjoint $f \dashv f^{\rm op}$ it is called a {\em negation\/}.
The polar-tension structure is the equivalence
$a \leq f(a')$ iff $a' \leq f(a)$,
and $A$-closed elements and $A^{\rm op}$-open elements coincide,
with dialecticality expressed as duality
$\closure{A} \cong \interior{(A^{\rm op})} = {(\closure{A})}^{\rm op}$.
So restricting $f$ to the $f^2$-closed elements of $A$ makes $f$ into an {\em involution\/}:
``idempotent'' $f^2(a) = a$,
``monotonic'' if $a \leq b$ then $f(b) \leq f(a)$,
and satisfying $f(a \vee b) = f(a) \wedge f(b)$ (a DeMorgan's law) and $f(\bot) = \top$ when the joins exist.

\paragraph{Biposets.}
A {\em biposet\/} is another name for an ordered category;
that is,
a category ${\bf P} = \quadruple{{\bf P}}{\preceq}{\circ}{{\rm Id}}$
whose homsets are posets and whose composition is monotonic on left and right.
We prefer to view biposets as vertical structures, preorders with a tensor product,
rather than as horizontal structures, ordered categories.

In more detail,
a biposet {\bf P} consists of the following data and axioms.
There is a collection of {\bf P}-objects $x,y,z, \cdots$ called {\em types\/},
     and a collection of {\bf P}-arrows $r,s,t, \cdots$ called {\em terms\/}.
Terms could also be called ``preprocesses'',
since processes (which are discussed in \cite{Kent88})
are terms which satisfy certain constraints or closure conditions.
Each term $r$ has a unique source type $y$ and a unique target type $x$,
denoted by the relational notation $\term{y}{r}{x}$.
The collection of terms from source type $y$ to target type $x$
is ordered by a binary relation $\preceq_{y,x}$ called {\em term entailment\/},
which is
transitive,
if $r \preceq s$ and $s \preceq t$ then $r \preceq t$,
reflexive
$r \preceq r$, and
antisymmetric,
$r \equiv s$ implies $r = s$,
where $r \equiv s$ means $r \preceq s$ and $s \preceq r$.
Dialectical logic entailment $\preceq_{y,x}$ between terms generalizes
standard logic entailment $\vdash$ between propositions.
For any two terms $\term{z}{s}{y}$ and $\term{y}{r}{x}$ with matching types
(target type of $s$ = source type of $r$) 
there is a composite term $\term{z}{s \circ r}{x}$,
where $\circ$ is a binary operation called {\em tensor product\/},
which is associative
$t \circ (s \circ r) = (t \circ s) \circ r$,
and monotonic on left and right,
$s \preceq s'$ and $r \preceq r'$ imply $(s \circ r) \preceq (s' \circ r')$.
Tensor product allows each term $\term{y}{r}{x}$ to specify a {\em right direct flow\/}
${\bf P}[z,y] \stackrel{\circ r}{\rightarrow} {\bf P}[z,x]$
and a {\em left direct flow\/}
${\bf P}[x,z] \stackrel{r \circ}{\rightarrow} {\bf P}[y,z]$
for each type $z$.
Any type $x$ is a term $\term{x}{x}{x}$,
which is an identity,
$r \circ x = r$ and $x \circ s = s$.
A biposet with one object (universal type) is called a {\em monoidal poset\/}.
For each {\bf P}-type $x$,
the collection ${\bf P}[x,x]$ of endoterms at $x$ is a monoidal poset.
If {\bf P} is a biposet,
then the op-dual or opposite biposet ${\bf P}^{\rm op}$ is the opposite category with the same homset order as {\bf P},
and the co-dual biposet ${\bf P}^{\rm co}$ is (the same category) {\bf P} with the opposite homset order.
A {\em morphism of biposets\/} ${\bf P} \stackrel{H}{\rightarrow} {\bf Q}$
is a functor which preserves homset order.
Any Heyting algebra is a biposet,
where tensor product coincides with lattice meet $s \circ r = s \wedge r$.
The category {\bf Rel} of sets and (binary) relations is a biposet,
where tensor product is relational composition
$S \circ R \define \{ (z,x) \mid \exists_{y \in Y} (z,y) \memberof S \mbox{ and } (y,x) \memberof R \}$.
A bimodule $\term{{\cal Y}}{R}{{\cal X}}$
between two preorders
${\cal Y} = \pair{Y}{\leq_Y}$ and ${\cal X} = \pair{X}{\leq_X}$
is a monotonic function
$\product{{\cal Y}^{\rm op}}{{\cal X}} \stackrel{R}{\rightarrow} 2$.
The category {\bf Bim} of preorders and preorder bimodules
(bimodules $\term{{\cal Y}}{R}{{\cal X}}$ are in bijection with
closed-above subsets $R \subseteq \product{{\cal Y}^{\rm op}}{{\cal X}}$)
is a biposet,
where tensor product is again relational composition
$S \circ R \define \{ (z,x) \mid \exists_{y \in Y} (z,y) \memberof S \mbox{ and } (y,x) \memberof R \}$.
Given an alphabet $A$,
the category of formal $A$-languages $\power{A^\ast}$ is a biposet;
whose arrows are formal languages,
whose composition is language concatenation,
and whose identity is singleton empty string $\{\varepsilon\}$.

Given two types $y$ and $x$ in a biposet {\bf P},
two opposed terms (terms oppositely directed)
$\term{x}{s}{y}$ and $\opterm{x}{r}{y}$
are {\em semi-orthogonal\/} at $x$,
denoted $s \orthog{x} r$,
when $s \circ r \prec_{x,x} x$.
Semi-orthogonality is a nonsymmetric notion.
By combining semi-orthogonality at source and target we get a symmetric notion:
two opposed terms $\term{y}{r}{x}$ and $\opterm{y}{s}{x}$
form an orthogonal pair of terms or an {\em orthoterm\/},
denoted by $\term{y}{r \orthogonal s}{x}$,
when they satisfy semi-orthogonality at $y$ and semi-orthogonality at $x$;
that is,
$r \orthogonal s$ iff ($r \circ s \preceq y$ and $s \circ r \preceq x$).
In this case, we say that $r$ is {\em orthogonal\/} to $s$.
Orthoterms axiomatize ``ring-structured {\bf P}-terms''. 
Orthoterms compose in the obvious way:
$(s \orthogonal s') \circ (r \orthogonal r') = (s \circ r) \orthogonal (r' \circ s')$,
and $(x \orthogonal x)$ is the identity orthoterm at $x$.
The homset order on orthoterms is defined by:
$(p \orthogonal q) \preceq (r \orthogonal s)$ when $p \preceq r$ and $q \succeq s$.
So each biposet {\bf P} has an associated {\em orthoterm category\/} ${\bf P}^\perp$,
whose objects are {\bf P}-types and whose arrows are {\bf P}-orthoterms.
There are two projection functors
${\bf P}^{\rm op} \stackrel{\partial_0}{\leftarrow} {\bf P}^\perp \stackrel{\partial_1}{\rightarrow} {\bf P}$,
whose product pairing functor is the inclusion
${\bf P}^\perp \stackrel{{\rm Inc}}{\rightarrow} \product{{\bf P}^{\rm op}}{{\bf P}}$.
Let ${\perp}(r)$ denote the collection of all terms opposed and orthogonal to $r$;
${\perp}(r) \define \{ \term{x}{s}{y} \mid r \orthogonal s \}$.
Then ${\perp}(r)$ is a closed-below subset of ${\bf P}[x,y]$.
In defining the phase semantics for linear logic,
Girard {\em implicitly\/} uses the notion of orthogonality
with respect to a single subset of ``antiphases'' $\perp$. 
Since orthogonality is defined with respect to types (identity endoterms) $x,y,z,\cdots$,
Girard's set of antiphases $\perp$ corresponds to any arbitrary {\bf P}-type.
Orthogonality of terms in biposets for dialectical logic generalizes
disjointness of elements in Heyting algebras for standard logic.

A monoid {\bf M} is {\em symmetric\/} (or {\em commutative\/})
when its tensor product is commutative:
$s \circ r = r \circ s$.
More generally,
a biposet {\bf P} is {\em quasisymmetric\/} or {\em orthogonally balanced\/} when
$s \orthog{x} r$ implies $r \orthog{y} s$
for all {\bf P}-types $y$ and $x$ and all opposed pairs of {\bf P}-terms
$\term{y}{r}{x}$ and $\opterm{y}{s}{x}$.
Obviously,
these implications can be replaced by logical equivalences. 
Quasisymmetry asserts that semi-orthogonality is equivalent to orthogonality:
$r \orthogonal s$ iff $s \orthog{x} r$ iff $r \orthog{y} s$.
A symmetric monoidal poset (ordered commutative monoid) is quasisymmetric as a one object biposet.

\paragraph{Internal Dialectics.}
For any opposed pair of ordinary relations $\term{Y}{R}{X}$ versus $\opterm{Y}{S}{X}$
the ``unit inequality'' $Y \subseteq R \circ S$ and
the ``counit inequality'' $S \circ R \subseteq X$
{\em taken together\/} are equivalent to the facts
that $R$ is the graph 
$R = \yoneda{}{(f)}{1} = \{ (y,f(y)) \mid y \memberof Y \}$
of a function $Y \stackrel{f}{\rightarrow} X$
and that $S$ is the transpose 
$S = R^{\rm op} = {\yoneda{}{(f)}{1}}^{\rm op}
   = \yoneda{}{(f)}{0} = \{ (f(y),y) \mid y \memberof Y \}$.
On the other hand,
the graph $\term{Y}{\yoneda{}{(f)}{1}}{X}$ of any function $Y \stackrel{f}{\rightarrow} X$ 
and its transpose $\yoneda{}{(f)}{0} = (\yoneda{}{(f)}{1})^{\rm op}$
satisfy the unit and counit inequalities.
So these conditions describe functionality in the biposet {\bf Rel}.
For any opposed pair of preorder bimodules 
$\term{{\cal Y}}{R}{{\cal X}}$ versus $\opterm{{\cal Y}}{S}{{\cal X}}$
where ${\cal X}$ is a complete lattice,
the ``unit inequality'' ${\cal Y} \subseteq R \circ S$ and
the ``counit inequality'' $S \circ R \subseteq {\cal X}$
{\em taken together\/} are equivalent to the facts
that $R$ is the graph 
$R = \yoneda{}{(f)}{1} = \{ (y,x) \mid f(y) \leq_X x \}$
of a monotonic function ${\cal Y} \stackrel{f}{\rightarrow} {\cal X}$
where $f$ is given by $f(y) = \bigwedge \{ x \memberof X \mid yRx \}$,
and that $S$ is the transposed graph of $f$'s order-theoretic involution
$S = (\yoneda{}{(\involution{f})}{1})^{\rm op}
   = \yoneda{}{(f)}{0} = \{ (x,y) \mid x \leq_X f(y) \}$
with $f$ given by $f(y) = \bigvee \{ x \memberof X \mid xSy \}$.
On the other hand,
the graph $\term{{\cal Y}}{\yoneda{}{(f)}{1}}{{\cal X}}$
of any monotonic function ${\cal Y} \stackrel{f}{\rightarrow} {\cal X}$ 
and its transpose $\yoneda{}{(f)}{0} = (\yoneda{}{(\involution{f})}{1})^{\rm op}$
satisfy the unit and counit inequalities.
So these conditions describe functionality in a part of the biposet {\bf Bim}.
In the general case,
when ${\cal X}$ is not necessarily complete,
the ``unit inequality'' $Y \subseteq R \circ S$ and
the ``counit inequality'' $S \circ R \subseteq X$
{\em taken together\/} are equivalent to the facts
that $R$ is the tensor implication ({\bf Bim} is a Heyting category)
$R = S \tensorimplysource {\cal X}
   = \{ (y,x) \mid (\forall x') \mbox{ if } {x'}Sy \mbox{ then } x' \leq_X x \}$
and that $S$ is the implication
$S = {\cal X} \tensorimplytarget R
   = \{ (x,y) \mid (\forall x') \mbox{ if } yR{x'} \mbox{ then } x \leq_X x' \}$.
So these conditions describe a potential functionality in the entire biposet {\bf Bim},
and can be used as a way of axiomatizing potential functionality in general biposets.
But they are also the defining conditions for internal adjoint pairs.

Two opposed terms $\term{y}{r}{x}$ and $\opterm{y}{s}{x}$ form an {\em adjoint pair\/} of terms or an {\em adjunction\/},
denoted by $\term{y}{r \dashv s}{x}$,
when they satisfy the ``unit inequality'' $y \preceq r \circ s$ and the ``counit inequality'' $s \circ r \preceq x$.
This axiomatizes ``functionality'' of {\bf P}-terms. 
The term $r$ is called the {\em left adjoint\/} and the term $s$ is called the {\em right adjoint\/} in the adjunction $r \dashv s$.
It is easy to show that right adjoints (and left adjoints) are unique,
when they exist:
if $\term{y}{r \dashv s_1}{x}$ and $\term{y}{r \dashv s_2}{x}$ then $s_1 = s_2$.
Denote the unique right adjoint of $\term{y}{r}{x}$ by $\opterm{y}{r^{\rm op}}{x}$. 
A {\em functional {\bf P}-term\/} is a {\bf P}-term with a right adjoint.
We usually use the notation $\term{y}{f \dashv f^{\rm op}}{x}$ for functional terms.
For any adjoint pair $\term{y}{f \dashv f^{\rm op}}{x}$:
when the unit is equality $y = f \circ f^{\rm op}$ they are a {\em coreflective pair\/};
when the counit is equality $f^{\rm op} \circ f = x$ they are a {\em reflective pair\/}; and
when both unit and counit are equalities they are an {\em inverse pair\/}.
For any functional term $\term{y}{f \dashv f^{\rm op}}{x}$:
the adjunction $f \dashv f^{\rm op}$ is a coreflection
   iff $f$ is an monomorphism (iff $f^{\rm op}$ is an epimorphism);
the adjunction is a reflection
   iff $f$ is an epimorphism (iff $f^{\rm op}$ is an monomorphism); and
the adjunction is a inversion
   iff $f$ is an isomorphism (iff $f^{\rm op}$ is an isomorphism),
   iff $f^{\rm op} = f^{-1}$ is the two-sided inverse of $f$.
Again we see that
(in this case, internal) adjointness is a kind of generalized inverse.
An internal coreflective pair $\term{y}{i \dashv p}{x}$ is also called a {\em subtype\/} of $x$.
Adjoint pairs compose in the obvious way:
$(g \dashv g^{\rm op}) \circ (f \dashv f^{\rm op}) = (g \circ f) \dashv (f^{\rm op} \circ g^{\rm op})$,
and $(x \dashv x)$ is the identity adjoint pair at $x$.
So each biposet {\bf P} has an associated adjoint pair category ${\bf P}^\dashv$,
whose objects are {\bf P}-types and whose arrows are {\bf P}-adjunctions.
Equivalently,
${\bf P}^\dashv$-arrows are just functional {\bf P}-terms.
There is an inclusion functor ${\bf P}^\dashv \stackrel{{\rm Inc}}{\rightarrow} {\bf P}$.
The construction ${(\;)}^\dashv$ can be described as
either ``internal dialecticality'' or ``functionality''.

In objective dialectics,
since dialectical contradictions are represented by adjunctions,
systems of dialectical contradictions are represented by diagrams in (pseudofunctors into) the category {\bf Adj}
whose objects are small categories and whose morphisms are adjoint pairs of functors. 
We call such a (pseudo)functor ${\bf P} \stackrel{E}{\longrightarrow} {\bf Adj}$
a {\em dialectical base} or an {\em indexed adjointness\/},
and use the notation
$E(y \stackrel{r}{\rightarrow} x) = (E^r \dashv E_r) \morph E(y) \rightarrow E(x)$.
A dialectical base can be split
into its {\em direct flow aspect\/} ${\bf P} \stackrel{E^{(\;)}}{\longrightarrow} {\bf Cat}$
and its {\em inverse flow aspect\/} ${\bf P}^{\rm op} \stackrel{E_{(\;)}}{\longrightarrow} {\bf Cat}$.
Objects of ${\bf P}$ are called {\em types} and arrows of ${\bf P}$ are called {\em terms}.
A {\em dialectical system\/} $y \stackrel{\iota,o}{\longrightarrow} x$ is a graph in {\bf P},
with inverse flow specifier $\iota$ and direct flow specifier $o$.
Dialectical systems are the ``motors of nature'' specifying the dialectical motion of structured entities,
and a dialectical base provides the ``motive power'' for this motion
(from a dialectical point-of-view ``motion'' is synonymous with ``transformation'').
In this paper we are chiefly concerned with dialectical bases of preorders.
Here a dialectical base ${\bf P} \stackrel{E}{\longrightarrow} {\bf adj}$
factors through the category {\bf adj} of preorders and adjoint pairs of monotonic functions,
and direct flow ${\bf P} \stackrel{E^{(\;)}}{\longrightarrow} {\bf PO}$
and inverse flow ${\bf P}^{\rm op} \stackrel{E_{(\;)}}{\longrightarrow} {\bf PO}$
map to preorders (and usually semilattices).
Any functional term $\term{y}{f}{x}$ in a biposet {\bf P} defines
a {\em direct image\/} monotonic function
$\scriptbf{P}[y,y] \stackrel{\scriptbf{P}^f}{\longrightarrow} {\bf P}[x,x]$
defined by ${\bf P}^f(q) \define f^{\rm op} \circ q \circ f$ for endoterms $\term{y}{q}{y}$,
and an {\em inverse image\/} monotonic function
${\bf P}[y,y] \stackrel{\scriptbf{P}_f}{\longleftarrow} {\bf P}[x,x]$
defined by ${\bf P}_f(p) \define f \circ p \circ f^{\rm op}$ for endoterms $\term{x}{p}{x}$.
It is easy to check that direct and inverse image form an adjoint pair of monotonic functions
${\bf P}(\term{y}{f}{x}) = {\bf P}[y,y] \stackrel{\scriptbf{P}^f \dashv \scriptbf{P}_f}{\longrightarrow} {\bf P}[x,x]$
for each functional {\bf P}-term $\term{y}{f \dashv f^{\rm op}}{x}$.
The construction ${\bf P}$,
mapping types to their poset of endoterms
${\bf P}(x) = {\bf P}[x,x]$ and
mapping functional {\bf P}-terms to their adjoint pair of direct/inverse image adjunction,
is a dialectical base (indexed adjointness) ${\bf P}^\dashv \stackrel{\scriptbf{P}}{\longrightarrow} {\bf adj}$.

\paragraph{Bisemilattices.}
The structural aspect of both the intuitionistic and classical semantics of dialectical logic is defined in terms of bisemilattices.
A {\em join bisemilattice\/} or {\em semiexact biposet\/} is a biposet whose homsets are finitely complete \mbox{(join-)}semilattices
and whose composition is finitely (join-)continuous.
Horizontally the term ``semilattice-valued category'' might be indicated,
but vertically from a bicategorical viewpoint the term ``bisemilattice'' seems appropriate.
In more detail,
a join bisemilattice
${\bf P} = \triple{\quadruple{{\bf P}}{\preceq}{\circ}{{\rm Id}}}{\vee}{\bot}$
consists of the data and axioms of a biposet
${\bf P} = \quadruple{{\bf P}}{\preceq}{\circ}{{\rm Id}}$,
plus the following.
For any two parallel terms $\term{y}{s,r}{x}$
there is a {\em join\/} term $\term{y}{s \vee r}{x}$
satisfying $s \vee r \preceq_{y,x} t$ iff $s \preceq_{y,x} t$ and $r \preceq_{y,x} t$.
For any pair of types $y$ and $x$ there is an {\em empty\/} (or {\em bottom\/}) term $\term{y}{\bot_{y,x}}{x}$
satisfying $\bot_{y,x} \preceq r$. 
The tensor product is finitely (join-) continuous
(distributive w.r.t. finite joins) on the right and the left,
$s \circ (r_1 \vee \cdots \vee r_n) = (s \circ r_1) \vee \cdots \vee (s \circ r_n)$
and $(s_1 \vee \cdots \vee s_m) \circ r = (s_1 \circ r) \vee \cdots \vee (s_m \circ r)$
for any natural numbers $n$ and $m$, including $0$.
A join bisemilattice with one object (universal type) is called a {\em monoidal join semilattice\/}.
For any {\bf P}-term $\term{y}{r}{x}$
the associated closed-below subset ${\perp}(r)$ of terms orthogonal to $r$
is also closed under finite joins:
$\bot_{x,y} \memberof {\perp}(r)$,
and if $s_1,s_2 \memberof {\perp}(r)$
then $(s_1 \vee s_2) \memberof {\perp}(r)$ also.
So ${\perp}(r)$ is an order ideal called the {\em orthogonality ideal\/} of $r$.
If {\bf P} is a join bisemilattice,
then the opposite biposet ${\bf P}^{\rm op}$ is also a join bisemilattice.
A {\em meet bisemilattice\/} is a biposet whose co-dual biposet is a join bisemilattice;
that is,
whose homsets are finitely complete (meet-)semilattices
and whose composition is finitely (meet-) continuous.
For any two parallel terms $\term{y}{s,r}{x}$
there is a {\em meet\/} term $\term{y}{s \wedge r}{x}$
satisfying $t \preceq_{y,x} s \wedge r$ iff $t \preceq_{y,x} s$ and $t \preceq_{y,x} r$.
For any pair of types $y$ and $x$ there is a {\em full\/} (or {\em top\/}) term $\term{y}{\top_{y,x}}{x}$
satisfying $r \preceq \top_{y,x}$.
A {\em morphism of join bisemilattices\/} ${\bf P} \stackrel{H}{\rightarrow} {\bf Q}$
is a functor which preserves homset order and finite homset joins.
A {\em bilattice\/} or {\em exact biposet\/} is a join bisemilattice whose homsets are lattices.
Note: a bilattice is not necessarily a meet bisemilattice.

To recapitulate,
a join bisemilattice
${\bf P} = \langle {\bf P},\preceq,\circ,{\rm Id},\vee,\bot \rangle$
is the central structural notion in dialectical logic.
It should be viewed as a direct generalization of a distributive lattice
$L = \langle L,\leq,\wedge,\top,\vee,\bot \rangle$.
The generalization occurs in two different senses.
(1) A join bisemilattice is a distributed structure:
the notion of types is included,
and the lattice operations are distributed over and between types.
(2) The lattice meet $s \wedge r$ is replaced by the tensor product $s \circ r$,
and the top (meet unit) $\top$ is replaced by the identities $\term{x}{x}{x}$.
Since a lattice meet is associative, unital, commutative, idempotent, and unit bounded,
whereas a tensor product is only associative and unital,
we see that commutativity, idempotency and unit-boundedness are discarded globally
in the generalization.
However, these three properties are incorporated in dialectical logic in two distinct ways.
On the one hand,
in the object aspect of dialectical logic the laws of idempotency and partiality (unit-boundedness)
are incorporated locally in the idea of comonoid (see appendix).
These local comonoidal contexts are standard contexts.
Comonoidal structures define the generalized topological notions of interior and closure of terms,
which are the modalities of affirmation and consideration from linear logic \cite{Girard}.
In axiomatics and proof theory,
the idempotency and partiality axioms are known as contraction and weakening.
On the other hand,
in the construction of the classical context from the intuitionistic context,
a natural weakened form of commutativity,
called quasisymmetry,
is found to be essential.
Moreover,
in the object aspect of classical dialectical logic,
quasisymmetry is equivalent to internal (topological) dialecticality!

A {\em complete Heyting category\/} or {\em complete bilattice\/},
abbreviated {\em cHc\/}, 
is the same as a complete join bisemilattice;
that is,
an join bisemilattice {\bf H} whose homsets are complete join semilattices
(arbitrary joins exist) and whose tensor product is join continuous
(completely distributive w.r.t. joins) on the right and the left,
$s \circ (\vee_i r_i) = \vee_i (s \circ r_i)$
and $(\vee_j s_j) \circ r = \vee_j (s_j \circ r)$.
Since the homset
${\bf H}[x,z]$ is a complete lattice
and the left tensor product
${\bf H}[x,z] \stackrel{r \circ}{\rightarrow} {\bf H}[y,z]$
is continuous, 
it has (and determines) a right adjoint
${\bf H}[x,z] \stackrel{r \tensorimplysource}{\leftarrow} {\bf H}[y,z]$
called {\em left tensor implication\/},
and defined by
$r \tensorimplysource t \define \bigvee \{ \term{x}{s}{z} \mid r \circ s \preceq_{y,z} t \}$.
Adjointness means that left tensor product and left tensor implication
satisfy the dialectical axiom
$r \circ s \preceq_{y,z} t \mbox{ iff } s \preceq_{x,z} r \tensorimplysource t$.
Similarly,
the right tensor product
${\bf H}[z,y] \stackrel{\circ r}{\rightarrow} {\bf H}[z,x]$
has (and determines) a right adjoint
${\bf H}[z,y] \stackrel{\tensorimplytarget r}{\leftarrow} {\bf H}[z,x]$
called {\em right tensor implication\/},
and defined by
$s \tensorimplytarget r \define \bigvee \{ \term{z}{t}{y} \mid t \circ r \preceq_{z,x} s \}$.
Adjointness means that right tensor product and right tensor implication
satisfy the dialectical axiom
$t \circ r \preceq_{z,x} s \mbox{ iff } t \preceq_{z,y} s \tensorimplytarget r$.
A complete Heyting category with one object (universal type) is called a {\em complete Heyting monoid\/}
\cite{Birkhoff,Henkin}
{\bf M} = \mbox{$\langle M,\preceq,\circ,e,\tensorimplysource,\tensorimplytarget,\vee,\bot,\wedge,\top \rangle$.}
If {\bf M} is symmetric,
then the two tensor implications are one:
$\imply \define \tensorimplysource = \tensorimplytarget$.
A complete symmetric Heyting monoid is known as a {\em closed\/} ({\em monoidal\/}) {\em poset\/}.

\paragraph{Examples.}
Complete Heyting categories are everywhere.
The datatype
{\bf 2} = \mbox{$\langle \{0,1\},\leq,\wedge,1,\Rightarrow,\vee,0 \rangle$}
        = $\power{{\bf 1}}$
of boolean values is a complete Heyting monoid,
whose tensor product is the homset lattice meet $\wedge = {\tt and}$ with unit $1 = {\tt true}$,
and whose homset boolean sum is $\vee = {\tt or}$ with bottom $\bot = 0 = {\tt false}$. 
The powerset datatype
$\power{A} = \mbox{$\langle \power{A},\subseteq,\cap,A,\Rightarrow,\cup,\emptyset \rangle$}$
of subsets of a fixed set $A$ is a complete Heyting monoid.
More generally,
any complete Heyting algebra
{\bf M} = \mbox{$\langle M,\preceq,\wedge,\top,\Rightarrow,\vee,\bot \rangle$}
is the same as a complete {\em cartesian\/} Heyting monoid,
where tensor product coincides with homset lattice meet $s \circ r = s \wedge r$.
The category {\bf Rel} is a complete Heyting category.
Given a monoid {\bf M} = \mbox{$\langle M,\circ,e, \rangle$.},
the category of formal {\bf M}-languages $\power{{\bf M}}$ is a complete Heyting monoid,
where tensor product is language concatenation $L \bullet K$ with unit $\{e\}$,
and the two tensor implications are (left and right) language division or cut
$L \backslash K \define \{ m \memberof M \mid \forall_{n \in M} \mbox{ if } n \memberof L \mbox{ then } n \circ m \memberof K \}$.
In particular,
given an alphabet $A$,
the category of formal $A$-languages $\power{A^\ast}$ is a complete Heyting monoid
(the free complete Heyting monoid over the set $A$).
The extended nonnegative real numbers
{\bf R} = \mbox{$\langle [0,\infty],\geq,+,0,\diff,\wedge,\infty,\vee,0 \rangle$}
with opposite order is a complete (noncartesian) Heyting monoid,
where tensor product is numerical sum $s + r$ with unit $0$,
and tensor implication is numerical difference $s \diff r \define s - r \mbox{ if } s \geq r, = 0 \mbox{ otherwise}$.
There is a complete Heyting monoid $\power{{\bf R}}$ associated with the extended nonnegative real numbers {\bf R},
whose morphisms $\term{0}{R}{0}$ are subsets of reals $R \subseteq [0,\infty]$ with $\bot_{0,0} = \emptyset$ and $\top_{0,0} = [0,\infty]$,
whose homset order is the closed-above order $S \preceq R$ when $S \subseteq \abov{R}$,
whose composition is defined pointwise by $S \circ R \define \{ s+r \mid s \memberof S, r \memberof R \}$,
and whose identity is $\term{0}{\singleton{0}}{0}$.
The singleton operator
${\bf R} \stackrel{\singleton{}}{\longrightarrow} \power{{\bf R}}$
functorially embeds ${\bf R}$ into $\power{{\bf R}}$.
The infimum operator $\wedge$ is a functor
$\power{{\bf R}} \stackrel{\wedge}{\longrightarrow} {\bf R}$,
and (on the single homset) infimum reflects $\wedge \dashv \{\,\}$
the powerset of reals $\power{{\bf R}}$ into the reals ${\bf R}$.
The examples $\power{A^\ast}$ and $\power{{\bf R}}$
motivate and are special cases of the following important construction.
Just as every set $C$ has an associated subset Heyting algebra $\power{C}$,
so also every category {\bf C} has an associated {\em subset category\/} $\power{{\bf C}}$,
whose objects are {\bf C}-objects,
and whose arrows are subsets of homsets:
$\term{y}{R}{x}$ when $R \subseteq {\bf C}[y,x]$.
So $\power{{\bf C}}[y,x] = \power{{\bf C}[y,x]}$ with $\bot_{y,x} = \emptyset$ and $\top_{y,x} = {\bf C}[y,x]$.
The tensor product in $\power{{\bf C}}$ is defined pointwise,
$S \circ R
 \define \{ z \stackrel{s \cdot_C r}{\rightarrow} x \mid s \memberof S, r \memberof R \}$,
generalizing the concatenation of formal languages and the addition of nondeterministic reals.
The identity at $x$ is the singleton set $x \stackrel{\{x\}}{\rightarrow} x$,
which can be identified with $x$ itself.
The left tensor implication is defined by
$R \tensorimplysource T
 \define \{ x \stackrel{s}{\rightarrow} z \mid (\forall r) \mbox{ if } r \memberof R \mbox{ then } r \cdot_C s \memberof T \}$
for any two $\power{{\bf C}}$-arrows $\term{y}{R}{x}$ and $\term{y}{T}{z}$,
and the right tensor implication is defined dually.
The booleans are the ``simplest'' subset category
{\bf 2} = $\power{{\bf 1}}$.

More generally,
every biposet {\bf P} has an associated {\em closure subset category\/} $\power{{\bf P}}$,
whose arrows, tensor product, and identities are as in the unordered (identity order) case,
and whose homset order is the closed-below order $S \preceq R$ when $S \subseteq \below{R}$.
The definition of the implications follow from the continuity of the tensor product:
the left tensor implication is
$R \tensorimplysource T \define \bigcup \{ \term{x}{S}{z} \mid R \circ S \preceq T \}$,
and the right tensor implication is defined dually.
Since every category {\bf C} is a biposet with the identity order on homsets,
the subset construction $\power{{\bf C}}$ is a special case of the closure subset construction.
It is easiest and most natural to define closure subset categories.
Furthermore,
this accords exactly with the appropriate generalization
when biposets (or better, bipreorders) are replaced by bicategories.
However,
it is standard practice to use partial orders and closed subsets of terms.
Any closure subset category $\power{{\bf P}}$ has an associated {\em closed subset category\/} $\closurepower{{\bf P}}$,
whose objects are the principal ideals $\{ \below{x} \mid x \mbox{ a {\bf P}-type} \}$,
whose arrows $\term{\below{y}}{R}{\below{x}}$ are closed-below subsets of terms $R \subseteq {\bf P}[y,x]$ and $R = \below{R}$,
whose homset order is subset inclusion $S \preceq R$ when $S \subseteq R$,
and whose tensor product is the closure of the $\power{{\bf P}}$-composition
$S \circ R
 \define \below{\{ \term{z}{s \circ r}{x} \mid s \memberof S, r \memberof R \}}$.
The definition of the implications is as above
$R \tensorimplysource T \define \bigcup \{ \term{\below{x}}{S}{\below{z}} \mid R \circ S \preceq T \}$.
For any biposet {\bf P},
the closed subset category $\closurepower{{\bf P}}$ is a complete Heyting category.
For any {\bf P}-term $\term{y}{r}{x}$
the orthogonality ideal is a term $\term{x}{{\perp}(r)}{y}$ in $\closurepower{{\bf P}}$. 
In fact,
orthogonality is a contravariant lax functor,
${\perp}(x) = {\downarrow} x$ and ${\perp}(r) \circ {\perp}(s) \subseteq {\perp}(s \circ r)$, 
which is also hom-set contravariant,
if $s \preceq r$ then ${\perp}(r) \subseteq {\perp}(s)$.

\paragraph{Type Sums.}
The closure subset construction $\power{{\bf P}}$ does not capture the notion of ``relational structures'' completely.
Although it introduces nondeterminism on the arrows,
it leaves the objects alone.
Type sums introduce distributivity on objects in a constructive fashion.
We give a brief survey of type sums here.

A popular ``external'' model for predicates in logic is provided by subtypes.
These are often constructed by a factorization/inclusion adjointness on slice categories of functional terms. 
Subtypes are closely connected with the ``internal'' model for predicates called comonoids
(discussed in the appendix).
For any type $x$,
an $x$-{\em subtype\/} $\term{y}{i \dashv p}{x}$
is another name for an internal coreflective pair $i \dashv p$ between $y$ and $x$;
that is,
$y = i \circ p$ and $p \circ i \preceq x$.
The interior term $\term{x}{p \circ i}{x}$ is the comonoid associated with the subtype.
We can define the usual subtype order between any two $x$-subtypes
$\term{y}{i \dashv p}{x}$ and $\term{z}{j \dashv q}{x}$
as $\pair{y}{i} \preceq \pair{z}{j}$
when there exists a functional term $\term{y}{h \dashv h^{\rm op}}{z}$
such that $i = h \circ j$ and $q \circ h^{\rm op} = p$.
The largest $x$-subtype is the identity $\term{x}{x \dashv x}{x}$.
A term $\term{z}{s}{y}$ is an (external) source {\em subterm\/} of a term $\term{y}{r}{x}$,
when $s = i \circ r$ for some source subtype $\term{z}{i \dashv p}{y}$.
Two terms $\term{z}{s}{x}$ and $\term{y}{r}{x}$ with common target type $x$
satisfy the {\em domain(-of-definition) order\/} $s \sqsubseteq r$
when
$z$ is a subtype of $y$ mediated by the coreflective pair $\term{z}{i \dashv p}{y}$
and
$s \preceq i \circ r$.
A more complete axiomatization of subtypes and comonoids is given in \cite{Kent89}.

The {\em empty type\/} $0$ is a special type
such that for any type $x$ there are unique terms
between $x$ and $0$ in either direction.
So $0$ is an initial type,
satisfying the condition $\term{0}{r}{x}$ implies $r = \bot_{0,x}$;
and $0$ is a terminal type,
satisfying the condition $\term{x}{r}{0}$ implies $r = \bot_{x,0}$.
A type that is both initial and terminal is a null type.
The null type $0$ is the ``empty sum'',
the sum of the empty collection of types.
For any pair of types $y$ and $x$,
the bottom term
$y \stackrel{\bot_{y,x}}{\rightharpoondown} x$ 
is the composition
$\bot_{y,x} = \bot_{y,0} \circ \bot_{0,x}$.
The empty type
$\term{0}{\bot_{0,x} \dashv \bot_{x,0}}{x}$ is the smallest subtype of any type $x$,
and its associated comonoid is the smallest comonoid.
Given two types $y$ and $x$,
the {\em sum} of $y$ and $x$ is a composite type $y \oplus x$
having $y$ and $x$ as disjoint subtypes
$y \stackrel{i_y \dashv p_y}{\rightharpoondown} {y \oplus x} \stackrel{i_x \dashv p_x}{\leftharpoondown} x$ 
which cover $y \oplus x$.
So $y \oplus x$ comes equipped
with two {\em injection terms\/}
$y \stackrel{i_y}{\rightharpoondown} y \oplus x \stackrel{i_x}{\leftharpoondown} x$
and two {\em projection terms\/}
$y \stackrel{p_y}{\leftharpoondown} y \oplus x \stackrel{p_x}{\rightharpoondown} x$
which satisfy the ``comonoid covering equation''
$(p_y \circ i_y) \vee (p_x \circ i_x) =  y \oplus x$
stating that the subtype comonoids cover the sum type,
and satisfy the ``subtype disjointness equations''
$i_y \circ p_y = y$,
$i_y \circ p_x = \bot_{y,x}$,
$i_x \circ p_y = \bot_{x,y}$, and
$i_x \circ p_x = x$,
or the ``comonoid disjointness equation''
$(p_y \circ i_y) \wedge (p_x \circ i_x) =  \bot_{y \oplus x}$
stating that the subtype comonoids partition the sum type.

Equivalently,
the sum type $y \oplus x$ is both a coproduct via the injections
and a product via the projections of the types $y$ and $x$.
Given any pair of terms
$y \stackrel{t}{\rightharpoondown} z \stackrel{s}{\leftharpoondown} x$
there is a unique term
$y \oplus x \stackrel{\relcopair{t}{s}{}{}}{\rightharpoondown} z$,
called the sum {\em source pairing\/} of $t$ and $s$,
which satisfies the source pairing conditions
$i_y \circ \relcopair{t}{s}{}{} = t$ and
$i_x \circ \relcopair{t}{s}{}{} = s$.
Just define
$\relcopair{t}{s}{}{} \define (p_y \circ t) \vee (p_x \circ s)$.
These properties say that the sum $y \oplus x$ is a coproduct.
Equivalently,
any term $\term{y \oplus x}{r}{z}$ satisfies
the ``subterm covering condition'' $r_y \vee r_x = r$
and the ``subterm disjointness condition'' $r_y \wedge r_x = \bot_{y \oplus x,z}$,
where the $y$-th and $x$-th internal source subterms of $r$ are defined by
$r_y \define (p_y \circ i_y) \circ r$ and $r_x \define (p_x \circ i_x) \circ r$.
Dually, given any pair of terms
$y \stackrel{t}{\leftharpoondown} z \stackrel{s}{\rightharpoondown} x$
there is a unique term
$z \stackrel{\relpair{t}{s}{}{}}{\rightharpoondown} y \oplus x$,
called the sum {\em target pairing\/} of $t$ and $s$,
which satisfies the target pairing conditions
$\relpair{t}{s}{}{} \circ p_y = t$ and
$\relpair{t}{s}{}{} \circ p_x = s$.
Just define
$\relpair{t}{s}{}{} \define (t \circ i_y) \vee (s \circ i_x)$.
These properties say that the sum $y \oplus x$ is a product.
Equivalently,
any term $\term{z}{r}{y \oplus x}$ satisfies
the ``subterm covering condition'' $r^y \vee r^x = r$
and the ``subterm disjointness condition'' $r^y \wedge r^x = \bot_{y \oplus x,z}$,
where the $y$-th and $x$-th internal target subterms of $r$ are defined by
$r^y \define r \circ (p_y \circ i_y)$ and $r^x \define r \circ (p_x \circ i_x)$.
An object which is both a product and a coproduct of two other objects
is called a {\em biproduct\/}.
So type sums are biproducts.
A join bisemilattice {\bf P} is said to have {\em type sums\/} or {\em biproducts\/}
when type sums exist for any (finite) collection of types.

\paragraph{Domains/Totality.}
The ``action'' of a term $\term{y}{r}{x}$ is concentrated in and localized to a ``locus of activity'',
a source subtype called the domain-of-definition of $r$
(and a target subtype called the range of $r$).
This domain is a kind of ``effect'' or ``read-out'' of a term $r$,
and defines predicate transformation \cite{Kent89} so that $r$ becomes a predicate transformer.
There are two approaches for formulating this.

One approach regards the notion of total term as fundamental,
and domain-of-definition as derived.
In this approach a term $\term{y}{r}{x}$ is defined to be {\em total\/}
when $s \circ r = \bot_{z,x}$ implies $s = \bot_{z,y}$ for any term $\term{z}{s}{y}$.
We then axiomatize the notion of domain-of-definition
by assuming that inclusion of total terms has a right adjoint right inverse $\total{(\,)}$ 
called the {\em totalization\/} or {\em total subterm operator\/} at $x$,
forming a coreflective pair ${\rm Inc} \dashv \total{(\,)}$
with ${\rm Inc} \cdot \total{(\,)} = {\rm Id}$.
This means that $t \sqsubseteq r$ iff $t \sqsubseteq \total{r}$
for any total term $\term{z}{t}{x}$ and any term $\term{y}{r}{x}$;
moreover,
$\total{t} = t$ for any total term $t$.
Equivalently,
$\total{r}$ is the largest total term under $r$ in the domain order:
(1) $\total{r} \sqsubseteq r$ and
(2) $t \sqsubseteq r$ implies $t \sqsubseteq \total{r}$ for total $t$.
So,
there is a $y$-subtype $d \stackrel{i \dashv p}{\rightharpoondown} y$
called the {\em domain subtype\/} of $r$,
such that $\total{r} \preceq i \circ r$.
Since total terms are closed above we must have equality $\total{r} = i \circ r$.
The associated $r$-subterm $\total{r}$ is called the {\em totalization\/} of $r$.
The domain subtype $d \stackrel{i \dashv p}{\rightharpoondown} y$
is the $y$-subtype where the term $\term{y}{r}{x}$ has non-nil action. 
It is the largest $y$-subtype
whose associated $r$-subterm is total,
in the sense that any other such subtype factors through the domain subtype.
We need additional axioms to ensure
that any term $r$ is recoverable from its totalization by the identity $r = p \circ \total{r}$.

Another, perhaps better,
approach regards the notion of domain-of-definition as fundamental,
and defines totalness as a derived notion.
The {\em domain subtype\/} of any term $\term{y}{r}{x}$
is the source subtype $\domain{(r)} = \term{d_r}{i_r \dashv p_r}{y}$
which satisfies the axioms:
(1) ``minimality''
    $z \succeq \domain{(r)}$ iff $p \circ i \circ r = r$
    for any source subtype $\term{z}{i \dashv p}{y}$;
(2) ``composition''
    $\domain{(s \circ r)} = \domain{(s \circ p_r)}$
    for any composable term $\term{z}{s}{y}$; and
(3) ``monotonicity''
    $r \preceq r'$ implies $\domain{(r)} \preceq \domain{(r')}$
    for any parallel term $\term{y}{r'}{x}$.
Define the {\em totalization\/} of $r$ to be the $r$-subterm $\total{r} \define i_r \circ r$.
A term $\term{y}{r}{x}$ is {\em total\/}
when its domain is the largest source subtype,
the entire source type $\domain{(r)} = y$.
Some identities for the domain operator $\domain{}$ are:
types are their own domain
$\domain{(x)} = x$;
the totalization is total, since
$\domain{(\total{r})}
 = \domain{(i_r \circ r)}
 = \domain{(i_r \circ p_r)}
 = \domain{(d_r)}
 = d_r$;
functional terms $\term{y}{f \dashv f^{\rm op}}{x}$ are total,
since the counit inequality $y \preceq f \circ f^{\rm op}$ implies
$y = \domain{(y)}
 \preceq \domain{(f \circ f^{\rm op})}
       = \domain{(f \circ p_{f^{\rm op}})}
 \preceq \domain{(f \circ x)}
       = \domain{(f)}
 \preceq y$;
in particular,
subtypes are total
$\domain{(\term{y}{i \dashv p}{x})} = y$;
domain subtypes are their own domain, since
$\domain{(p_r)}
 = \domain{(p_r \circ d_r)}
 = \domain{(p_r \circ \total{r})}
 = \domain{(r)}
 = d_r$;
only zero has empty domain
$\domain{(r)} = \term{0}{\bot_{0,y} \dashv \bot_{y,0}}{y}$
iff $r = 0_{y,x}$ for any term $\term{y}{r}{x}$; and
given any two total terms $\term{z}{s}{y}$ and $\term{y}{r}{x}$,
the composite term $\term{z}{s \circ r}{x}$ is also total,
since 
$\domain{(s \circ r)}
 = \domain{(s \circ p_r)}
 = \domain{(s \circ y)}
 = \domain{(s)}
 = z$.

Total terms are close above w.r.t. term entailment order.
Since functional terms (in particular, identity terms) are total,
and the composite of total terms are also total,
total terms form a biposet $\total{{\bf P}}$,
a subbiposet of ${\bf P}$,
${\bf P}^\dashv \subseteq \total{{\bf P}} \subseteq {\bf P}$,
which is the homset order closure of ${\bf P}^\dashv$.
So $\total{{\bf P}}$ is a subbiposet ${\bf P}$,
which preserves homset joins but usually does not have a bottom.
Total terms in Heyting categories have been suggested \cite{Hoare87}
(although not by that name)
as good models for programs
(brief discussion in the section on Heyting categories).


\paragraph{Matrices and Distributors.}
There is a cHc with type sums $\mat{{\bf R}}$ associated with the complete Heyting monoid of nonnegative reals 
{\bf R} = \mbox{$\langle [0,\infty],\geq,+,0,\diff,\wedge,\infty,\vee,0 \rangle$};
whose objects are sets $X,Y,Z, \cdots$,
whose morphisms $\term{Y}{\phi}{X}$ are $\product{Y}{X}$-indexed collections of reals
$\phi = \{ \phi_{yx} \mid y \memberof Y, x \memberof X \}$
(that is, real-valued characteristic functions $\product{Y}{X} \stackrel{\phi}{\rightarrow} [0,\infty]$),
whose composition $\term{Z}{\psi \circ \phi}{X}$ for morphisms $\term{Z}{\psi}{Y}$ and $\term{Y}{\phi}{X}$ is
$(\psi \circ \phi)_{zx} \define \bigwedge_{y \in Y} [\psi_{zy} + \phi_{yx}]$,
and whose identity $\term{X}{X}{X}$ at $X$ is defined by
$X_{x'x} = 0 \mbox{ if } x'=x,
         = \infty \mbox{ otherwise}$.
Terms $\term{Y}{\phi}{X}$ can be viewed as {\em fuzzy relations\/},
where $\phi_{yx}$ measures the degree of membership in $\phi$,
with $\phi_{yx} = 0$ asserting full (crisp) membership $(y,x) \memberof \phi$
and  $\phi_{yx} = \infty$ asserting full nonmembership $(y,x) \not\memberof \phi$.
More generally,
every cHc {\bf H} has an associated {\em matrix category\/} $\mat{{\bf H}}$,
whose objects are {\em {\bf H}-vectors\/} ${\cal X} = \pair{X}{|\;|_{\cal X}}$
where $X$ is an indexing (node) set and $X \stackrel{|\;|_{\cal X}}{\rightarrow} {\rm Obj}({\bf H})$ is a (typing) function,
whose arrows $\term{{\cal Y}}{R}{{\cal X}}$ are {\em {\bf H}-matrices\/}
where $R$ is a $\product{Y}{X}$-indexed collection of {\bf H}-terms
$R = \left( \term{|y|_{\cal Y}}{r_{yx}}{|x|_{\cal X}} \mid y \memberof Y, x \memberof X \right)$
(in other words, a generalized ${\rm Ar}({\bf H})$-valued characteristic functions
 $\product{Y}{X} \stackrel{r}{\rightarrow} {\rm Ar}({\bf H})$
 compatible with source and target),
whose homset order is pointwise order
$(s_{yx}) \preceq (r_{yx})$ when $s_{yx} \preceq r_{yx}$
for all $y \memberof Y$ and $x \memberof X$,
whose composition is matrix tensor product
$(S \circ R)_{zx} = S_{zY} \circ R_{Yx}
                  = \bigvee_{y \in Y} (s_{zy} \circ r_{yx})$
                  ``{\em matrix tensor product}''
for composable matrices $\term{{\cal Z}}{S}{{\cal Y}}$ and $\term{{\cal Y}}{R}{{\cal X}}$,
whose identity at ${\cal X}$ is the diagonal matrix $\term{{\cal X}}{{\cal X}}{{\cal X}}$
defined as identity {\bf H}-terms
${\cal X}_{xx} = \term{|x|_{\cal X}}{|x|_{\cal X}}{|x|_{\cal X}}$
on the diagonal
and zero (bottom) {\bf H}-terms
${\cal X}_{x'x} = \term{|x'|_{\cal X}}{\bot}{|x|_{\cal X}}$
off the diagonal,
and whose matrix tensor implications are
$(S \tensorimplytarget R)_{zy} = S_{zX} \tensorimplytarget R_{yX}
                               = \bigwedge_{x \in X} (s_{zx} \tensorimplytarget r_{yx})$
                               ``{\em right matrix tensor implication}'' and
$(R \tensorimplysource T)_{xz} = R_{Yx} \tensorimplysource T_{Yz}
                               = \bigwedge_{y \in Y} (r_{yx} \tensorimplysource t_{yz})$
                               ``{\em left matrix tensor implication}''.
Matrices $\term{Y}{R}{X}$ can be viewed as {\em fuzzy {\bf H}-relations\/}.
For any cHc {\bf H},
the matrix category $\mat{{\bf H}}$ is a complete Heyting category
for which biproducts (type sums) exist.
For the complete cartesian Heyting monoid of boolean values
{\bf 2} = \mbox{$\langle \{0,1\},\leq,\wedge,1,\Rightarrow,\vee,0 \rangle$}
        = $\power{{\bf 1}}$
the associated cHc with biproducts is
$\mat{{\bf 2}} = \mat{\power{{\bf 1}}} = {\bf Rel}$
the category of ordinary relations.

Every category {\bf C} has an associated {\em distributor category\/} $\distrib{{\bf C}}$
defined by $\distrib{{\bf C}} \define \mat{\power{{\bf C}}}$.
In more detail,
$\distrib{{\bf C}}$ is the category,
whose objects are {\em distributed {\bf C}-objects\/} or {\bf C}-{\em vectors\/} ${\cal X} = \pair{X}{|\;|_{\cal X}}$ as above,
whose arrows $\term{{\cal Y}}{R}{{\cal X}}$ are {\em distributed {\bf C}-arrows\/} or {\bf C}-{\em distributors}
where $R \subseteq \triproduct{Y}{{\rm Ar}({\bf C})}{X}$ is a digraph between the underlying node sets
consisting of compatible triples:
if $(y,r,x) \memberof R$ then $|y|_{\cal Y} \stackrel{r}{\rightarrow} |x|_{\cal X}$ is a {\bf C}-arrow,
whose tensor product is defined pointwise as
$(S \circ R)_{z,x} \define \bigcup_{y \in Y} [S_{zy} \circ R_{yx}]$,
and whose identity at ${\cal X}$ is the {\bf C}-distributor
${\cal X} \define \{ (x,|x|_{\cal X},x) \mid x \memberof X \}
          \subseteq \triproduct{X}{{\rm Ar}({\bf C})}{X}$
consisting (on the diagonal) of all the {\bf C}-identities indexed by ${\cal X}$.
The $(y,x)$-th fiber of a $\distrib{{\bf C}}$-term $\term{{\cal Y}}{R}{{\cal X}}$,
defined by $R_{yx} \define \{ \term{y}{r}{x} \mid r \memberof R \}$,
is a $\power{{\bf C}}$-term $\term{y}{R_{yx}}{x}$,
and $R$ is the disjoint union
$R = \coprod_{y \in {\cal Y}, x \in {\cal X}} R_{yx}$
of its $\power{{\bf C}}$-term fibers.
For any category {\bf C},
the distributor category $\distrib{{\bf C}}$ is a complete Heyting category
for which biproducts (type sums) exist.
The category of relations is the ``simplest'' distributor category
${\bf Rel} = \distrib{{\bf 1}}$.
Since any category {\bf C} has a unique functor
${\bf C} \stackrel{!}{\rightarrow} {\bf 1}$
to the one-arrow category,
every distributor category has a functor
(morphism of distributor categories)
$\distrib{{\bf C}} \stackrel{\distrib{!}}{\rightarrow} {\bf Rel}=\distrib{{\bf 1}}$.

In distributor categories $\distrib{{\bf C}}$ a comonoid $W$ of type $X$
is essentially a subobject (subset) $W \subseteq X$,
and so $\comonoid{}{X}{} \cong \power{X}$.
More generally,
every biposet {\bf P} has an associated {\em closure distributor category\/} $\distrib{{\bf P}} \define \mat{\power{{\bf P}}}$,
whose objects, arrows, tensor product and identities are as above,
and whose homset order is the pointwise closed-below order.
Given any set of attributes or sorts $A$,
a signature
$\Sigma = \{ \Sigma_{y,a} \mid y \memberof \mbox{multiset}(A), a \memberof A \}$ 
over $A$
determines a term category ${\bf T}_\Sigma$,
the initial algebraic theory over $\Sigma$,
whose objects are multisubsets of $A$ 
(arities, tuplings, etc.)
and whose arrows are tuples of $\Sigma$-terms.
A parallel pair of arrows
$\term{{\cal Y}}{S,R}{{\cal X}}$
in the distributor category $\distrib{{\bf T}_\Sigma^{\rm op}}$
is a Horn clause logic program,
whose predicate names are ${\cal X}$-nodes,
whose clause names are ${\cal Y}$-nodes, 
whose clause-head atoms are (w.l.o.g.) collected together as $S$,
whose clause-body atoms are collected together as $R$,
and whose associated fixpoint operator (see appendix) is the inverse/direct flow composite
$((\,) \tensorimplytarget R) \circ S$
defined on Herbrand interpretations with database scheme ${\cal X}$.
In much of the logic of dialectical processes
(in particular, for Girard's completeness theorem)
closure subset categories suffice.
However,
for the constraint dialectic,
the full nondeterminism and parallelism of distributor categories is essential.

\section{Semantics}

Flow is at the heart of computational and dynamic systems.
From the calculi and semantics of processes comes the notion of process communication and process flow.
From logic programming and Petri net theory comes the idea that flow is dialectical,
in the sense of moving in both a direct and an inverse direction.
Flow is the behavior of dialectical processes.
Direct flow is modelled by a nonsymmetric tensor product,
whereas inverse flow is modelled by
both a left (reverse-time, source, quo-object) tensor implication
and a right (forward-time, target, subobject) tensor implication
(or tensor exponentiations).
This bidirectional notion of flow is called the {\em flow\/} (or {\em motion\/}) {\em dialectic\/}.

Both dialectical logic and linear logic deal principally with the dynamical notions of {\em state\/} and {\em transitions\/} 
(involving ``dialectically contradictory'' activities \cite{Kent87},
 such as the creation/destruction or production/consumption of values, often representing resources),
whereas standard logic,
both classical and intuitionistic,
deals with the relatively static notion of monotonically increasing truth values
(once true, true forever).
Dialectical and linear logic are proper extensions of standard logic,
relegating the cartesian-ness of the standard fragment \cite{Kent88}
(weakening, contraction, etc.)
to local contexts:
that is,
they have locally cartesian-closed semantical structures.
Presently linear logic requires the commutativity or symmetry of tensor product,
in order to define a simpler semantics.
However, the semantics of dialectical processes,
which includes traditional process semantics,
is not commutative.
This argues strongly that commutativity should be excluded initially,
and only included later when desired via a symmetrization construction on the nonsymmetric case.
The semantics and logic of dialectical processes in this paper
agrees with linear logic in subject studied and philosophy.
They disagree in approach taken (I use a previously developed theory of dialectical systems)
and in emphasis:
linear logic emphasizes the importance of the linearity properties of implication and negation;
whereas dialectical logic emphasizes the importance of the central dialectical contradiction
(adjointness) between tensor product and tensor implication,
thus giving logic a process interpretation.
The logic of dialectical processes is more general than linear logic for two reasons:
1. dialectical logic is nonsymmetric (has a nonsymmetric tensor product operation)
   with linear logic a symmetric subcase;
2. linear logic is a typeless subcase of dialectical logic
   (all types are merged into one type).

\paragraph{Heyting Categories.}
The full intuitionistic semantics of dialectical logic is defined in terms of Heyting categories.
Concisely speaking,
a {\em Heyting category\/} is a closed bilattice;
that is,
an bilattice {\bf H} whose tensor product has right adjoints on both left and right.
The underlying bilattice represents the structural aspect of a Heyting category,
whereas the closedness property represents the dialectical or flow aspect.

In more detail,
the flow aspect consists of the following data and axioms.
For any two {\bf H}-terms $\term{y}{r}{x}$ and $\term{z}{s}{x}$ with common target type
there is a composite term $\term{z}{s \tensorimplytarget r}{y}$ between their source types,
defined by the dialectical axiom
$t \circ r \preceq_{z,x} s \mbox{ iff } t \preceq_{z,y} s \tensorimplytarget r$,
stating that the binary operation $\tensorimplytarget$ called {\em right tensor implication\/},
is right adjoint to tensor product on the right.
Right tensor implication $\tensorimplytarget$,
like all exponentiation or division operators including numerical ones,
is covariantly monotonic on the left and contravariantly monotonic on the right.
This dialectical axiom,
generalizing the deduction theorem of standard logic,
defines the formal semantics of tensor implication $\tensorimplytarget$ in terms of tensor product $\circ$.
From the dialectical axiom easily follows
the inference rule of right modus ponens
$(s \tensorimplytarget r) \circ r \preceq s$
and the inference rule
$t \preceq (t \circ r) \tensorimplytarget r$.
Also immediate from the axioms are
the transitive, reflexive, mixed associative and unital laws:
$(t \tensorimplytarget s) \circ (s \tensorimplytarget r)
 \preceq (t \tensorimplytarget r)$,
$y \preceq (r \tensorimplytarget r)$,
$t \tensorimplytarget (s \circ r)
 = (t \tensorimplytarget r) \tensorimplytarget s$,
$(r \tensorimplytarget x) = r$.
Right tensor implication allows each term $\term{y}{r}{x}$ to specify a {\em right inverse flow\/}
${\bf H}[z,y] \stackrel{\tensorimplytarget r}{\leftarrow} {\bf H}[z,x]$
for each type $z$.
The above mixed associative and unital laws say that right inverse flow $\tensorimplytarget r$
is (contravariantly) functorial in $r$ with respect to the category {\bf H}.
Thus,
each term $r$,
using right tensor product and right tensor implication,
specifies a ``right dialectical base'' for each type $z$.
Dually,
for any two {\bf H}-terms $\term{y}{r}{x}$ and $\term{y}{t}{z}$ with common source type
there is a composite term $\term{x}{r \tensorimplysource t}{z}$ between their target types,
defined by the dialectical axiom
$r \circ s \preceq_{y,z} t \mbox{ iff } s \preceq_{x,z} r \tensorimplysource t$,
stating that the binary operation $\tensorimplysource$ called {\em left tensor implication\/},
is right adjoint to tensor product on the left.
Left tensor implication allows each term $\term{y}{r}{x}$ to specify a {\em left inverse flow\/}
${\bf H}[x,z] \stackrel{r \tensorimplysource}{\leftarrow} {\bf H}[y,z]$
for each type $z$.
The mixed associative and unital laws say that left inverse flow $r \tensorimplysource$
is (covariantly) functorial in $r$ with respect to the category {\bf H},
thus defining a ``left dialectical base''.
Together the left and right implications satisfy the mixed associative law
$s \tensorimplysource (t \tensorimplytarget r)
 = (s \tensorimplysource t) \tensorimplytarget r$.
From both the left and right modus ponens,
we get the derived rules
$(r \tensorimplytarget r) \tensorimplysource r
 = r
 = r \tensorimplytarget (r \tensorimplysource r)$.
Since tensor product is left adjoint on both left and right to tensor implication,
it preserves arbitrary joins
$s \circ (r \vee r')
 = (s \circ r) \vee (s \circ r')$,
$s \circ  \bot_{y,x} 
 = \bot_{z,x}$, 
$(s \vee s') \circ r
 = (s \circ r) \vee (s' \circ r)$ and
$\bot_{z,y} \circ r
 = \bot_{z,x}$.
Since tensor implications are right adjoint to tensor product,
they preserve arbitrary meets
$r \tensorimplysource (t \wedge t')
 = (r \tensorimplysource t) \wedge (r \tensorimplysource t')$,
$r \tensorimplysource \top_{y,z} 
 = \top_{x,z}$, 
$(s \wedge s') \tensorimplytarget r
 = (s \tensorimplytarget r) \wedge (s' \tensorimplytarget r)$ and
$\top_{z,x} \tensorimplytarget r
 = \top_{z,y}$.
The two dialectical axioms assert that the bilattice {\bf H} is closed.

For any functional Heyting term $\term{y}{f \dashv f^{\rm op}}{x}$,
tensor implication relates the adjoints by
$f = f^{\rm op} \tensorimplysource x$
and
$f^{\rm op} = x \tensorimplytarget f$.
More generally,
left $f$-product is equal to left $f^{\rm op}$-implication
$f \circ (\,) = f^{\rm op} \tensorimplysource (\,)$
and right $f^{\rm op}$-product is equal to right $f$-implication
$(\,) \circ f^{\rm op} = (\,) \tensorimplytarget f$,
and we have the adjoint triples
\begin{center}
   $\begin{array}{c@{\;\dashv\;}r@{\;=\;}l@{\;\dashv\;}c}
       f^{\rm op} \circ (\,) \mbox{  }
          & \mbox{  } f \circ (\,) & f^{\rm op} \tensorimplysource (\,) \mbox{  }
          & \mbox{  } f \tensorimplysource (\,) \\
       (\,) \circ f \mbox{  }
          & \mbox{  } (\,) \circ f^{\rm op} & (\,) \tensorimplytarget f \mbox{  }
          & \mbox{  } (\,) \tensorimplytarget f^{\rm op}.
    \end{array}$
\end{center}    
Such adjoint triples appear naturally in the dialectical view of dynamic logic called the standard aspect \cite{Kent89},
which discusses the equivalent notions of hyperdoctrines of comonoids and spannable dialectical flow categories.
A Heyting category with one object (universal type) is called a {\em Heyting monoid\/}
{\bf M} = \mbox{$\langle M,\preceq,\circ,e,\tensorimplysource,\tensorimplytarget,\vee,\bot,\wedge,\top \rangle$.}
A preliminary version of Heyting monoid without homset lattice notions,
was investigated early on \cite{Lambek},
and called {\em residuated preorder\/}.
See also \cite{Birkhoff,Henkin}.
The opposite biposet ${\bf H}^{\rm op}$ is a Heyting category with implications switched.
Since complete Heyting categories are Heyting categories,
Heyting categories are ubiquitous;
in particular,
subset categories $\power{{\bf C}}$ and distributor categories $\distrib{{\bf C}}$ are Heyting categories.

Concurrent with the development of this paper,
an algebraic theory for the ``laws of progamming'' has been advocated \cite{Hoare87},
whose axioms are essentially those for Heyting categories;
or more precisely,
Heyting categories (in particular, cHc) with affirmation/consideration modalities and domain subtypes.
The affirmation modality is defined in the appendix.
The consideration modality is its order-theoretic dual.
The topological notions of affirmation and consideration are discussed further
in both the standard aspect and the object aspect of dialectical logic \cite{Kent88,Kent89}.
In the program interpretation,
arbitrary Heyting terms represent progam specifications,
total Heyting terms represent programs,
and either subtypes or comonoids (see appendix) represent conditions.
Types represent local contexts for local states of the system.
Term entailment order is interpreted as a measure of ``nondeterminism''
with $r \preceq s$ asserting that $r$ is more deterministic than $s$.
The top term $\term{y}{\top_{y,x}}{x}$ represents the worst (most nondeterministic) program,
and functional terms represent fully deterministic (minimally nondeterministic) programs.
The bottom term $\term{y}{\bot_{y,x}}{x}$, although deterministic, is not a program since its domain-of-definition is empty.
The totalization $\term{d}{\total{r}}{x}$ of a term $\term{y}{r}{x}$
is the least deterministic program (on the domain-of-definition) of that specification.
In summary,
the ``Laws of Programming'' can be interpreted in Heyting categories as follows.
\begin{center}
   \begin{tabular}{|lc||lc|}
      \hline
      \multicolumn{2}{|c||}{``Laws of Programming''} & \multicolumn{2}{c|}{{\bf Heyting Categories}} \\
      \hline
      \hline
      program specifications          & $S$               & terms                     & $\term{y}{r}{x}$              \\
      programs                        & $P$               & total terms               & $\term{y}{t}{x}$              \\
      conditions                      & $b$               & comonoids                 & $u \memberof \comonoid{}{x}{}$ \\
                                      &                   & subtypes                  & $\term{y}{i \dashv p}{x}$     \\
      \hline
      nondeterminism order            & $P \subseteq Q$   & term entailment order     & $r \preceq s$ \\
      sequential composition          & $P {\bf ;} Q$   & tensor product            & $s \circ r$ \\
      nondeterministic choice         & $P \bigcup Q$   & boolean sum               & $s \vee r$  \\
      \verb:SKIP:, the nop            & {\rm II}          & identity (types-as-terms) & $\term{x}{x}{x}$   \\ 
      \verb:ABORT:, the worst program & $\bot$            & top term                  & $\term{y}{\top_{y,x}}{x}$ \\
      weakest prespecification        & $S/T$             & tensor implication        & $t \tensorimplytarget s$ \\
      \hline
      conditional or branch           & $P {\triangleleft} b {\triangleright} Q$ & derived expression & $(v \circ r) \vee (\boolneg{v} \circ s)$ \\
                                      & {\bf if} $b$ {\bf then} $P$ {\bf else} $Q$    & \multicolumn{2}{c|}{where $\boolneg{v} \define (v \imply \bot_y) = \interior{(v \tensorimplysource \bot_y)}$} \\
                                      &                                               & \multicolumn{2}{c|}{and $\interior{(\,)}$ is the affirmation modality} \\
      iteration or {\bf while}-loop   & $b \ast P$        & derived expression        & $\closure{(u \circ r)} \circ \boolneg{u}$ \\
                                      & {\bf while} $b$ {\bf do} $P$                  & \multicolumn{2}{c|}{where $\closure{(\,)}$ is the consideration modality} \\
      \hline
   \end{tabular}
\end{center}
More recently \cite{Kent89} these laws (concerning structure and flow in Heyting categories)
have been connected with the older program semantics which uses Hoare triples.

\paragraph{Tensor Negation.}
Glivenko's theorem,
defining the classical part of standard intuitionistic logic,
seems to rely in part upon the symmetry (commutativity) of the boolean product (lattice meet) in Heyting algebras.
Recall that a biposet {\bf P} is quasisymmetric when
$r \orthogonal s$ iff $s \orthog{x} r$ iff $r \orthog{y} s$
for all {\bf P}-types $y$ and $x$ and all opposed pairs of {\bf P}-terms
$\term{y}{r}{x}$ and $\opterm{y}{s}{x}$.
We can define quasisymmetry for {\bf P}-terms alone:
a {\bf P}-term $\term{y}{r}{x}$ is {\em quasisymmetric\/} or {\em orthogonally balanced\/} when
$s \orthog{x} r$ iff $r \orthog{y} s$
for all {\bf P}-terms $\term{x}{s}{y}$ opposed to $r$.
I cannot overemphasize the importance of the notion of quasisymmetry,
especially in the object aspect of classical dialectical logic \cite{Kent88}.
Dually,
a {\bf P}-term $\term{y}{r}{x}$ is {\em coquasisymmetric\/} when
it is quasisymmetric in the codual ${\bf P}^{\rm co}$,
which is {\bf P} with the opposite homset order;
that is,
when
$r \circ s \succeq_{y,y} y$ iff $s \circ r \succeq_{x,x} x$
for all {\bf P}-terms $\term{x}{s}{y}$ opposed to $r$.
Identities are quasisymmetric,
and quasisymmetric {\bf P}-terms are closed under composition. 
The {\em center\/} of {\bf P},
denoted by $\zenter{{\bf P}}$,
is the sub-biposet consisting of all {\bf P}-types and all quasisymmetric {\bf P}-terms.
All {\bf P}-isomorphisms are quasisymmetric.
Quasisymmetric {\bf P}-terms are closed under arbitrary joins w.r.t. $\preceq$ (when they exist).
When arbitrary joins of quasisymmetric terms exist,
the center $\zenter{{\bf P}}$ is a kind of generalized topology
with finite tensor products functioning as ``finite intersections''
and arbitrary boolean sums (joins) functioning as ``arbitrary unions'' \cite{Kent88}.
For this reason quasisymmetric terms are also called $\zenter{{\bf P}}$-open terms.

Now let the biposet {\bf P} be a Heyting category {\bf H}.
For any {\bf H}-term $\term{y}{r}{x}$,
the {\em left $x$-dual\/} of $r$ is $x \tensorimplytarget r$,
the largest term with source $x$ and target $y$ which is semi-orthogonal to $r$ at $x$:
$(x \tensorimplytarget r) \orthog{x} r$,
and if $s \orthog{x} r$ for $\term{x}{s}{y}$ then $s \preceq_{x,y} x \tensorimplytarget r$.
Dually,
the {\em right $y$-dual\/} of $r$ is $r \tensorimplysource y$,
the largest term with source $x$ and target $y$ which is semi-orthogonal to $r$ at $y$.
We have
$r \orthogonal s$
 iff ($s \orthog{x} r$ and $r \orthog{y} s$)
 iff ($s \preceq_{x,y} x \tensorimplytarget r$ and $s \preceq_{x,y} r \tensorimplysource y$)
 iff  $s \preceq_{x,y} (r \tensorimplysource y) \wedge (x \tensorimplytarget r)$.
Define the {\em tensor negation\/} of the Heyting term $\term{y}{r}{x}$ to be the term
$\tneg{}{r}{}
 = \tneg{y}{r}{x}
 \define (r \tensorimplysource y) \wedge (x \tensorimplytarget r)$.
So for any Heyting term $\term{y}{r}{x}$, 
the orthogonality ideal ${\perp}(r)$ is the principal ideal
${\perp}(r)
 = \principalideal{\tneg{}{r}{}}
 = \principalideal{(r \tensorimplysource y) \wedge (x \tensorimplytarget r)}$,
and tensor negation $\term{x}{\tneg{}{r}{}}{y}$
is the largest (oppositely directed) term orthogonal to $r$:
$\tneg{}{r}{} = \top_{{\perp}(r)}$;
or, phrased as an equivalence, 
$r \orthogonal s \mbox{ iff } s \preceq_{x,y} \tneg{}{r}{}$.
The definition of Boolean categories below uses this equivalence
to axiomatize tensor negation without the need for tensor implications.
The sense of this equivalence is that tensor negation is the ``tensor complement'' of $r$.
So tensor negation in dialectical logic is entirely analogous to (and generalizes) boolean negation in standard logic,
where the boolean negation of a Heyting element $a$ is the largest element disjoint from $a$,
$a \wedge b = 0 \mbox{ iff } b \leq \tneg{}{a}{}$.
Since tensor negation
${\bf H}[y,x] \stackrel{\tneg{y}{}{x}}{\rightarrow} {\bf H}[x,y]^{\rm op}$
is contravariantly monotonic,
$s \preceq_{y,x} r$ implies $\tneg{}{r}{} \preceq_{x,y} \tneg{}{s}{}$,
it is a dialectical negation.
In more detail,
since orthogonality is a symmetrical notion,
$s \preceq_{x,y} \tneg{y}{r}{x}$
 iff $r \orthogonal s$
 iff $r \preceq_{y,x} \tneg{x}{s}{y}$,
tensor negation is a self-adjoint monotonic function
$\tneg{y}{}{x} \dashv \tneg{x}{}{y}^{\rm coop}$.
Since tensor negation $\tneg{}{}{}$ is self-adjoint,
it maps arbitrary joins to meets
$\tneg{}{(\vee_i r_i)}{} = \wedge_i (\tneg{}{r_i}{})$,
which in the binary case gives the DeMorgan's law:
$\tneg{}{(s \vee r)}{} = \tneg{}{s}{} \wedge \tneg{}{r}{}$
and in the nullary case gives the law:
$\tneg{}{\bot_{y,x}}{} =  \top_{x,y}$.
We also have the derived rule
$\tneg{z}{(s \circ r)}{x}
 = (r \tensorimplysource (z \tensorimplysource s)) \wedge ((x \tensorimplytarget r) \tensorimplytarget s)$.
As remarked before,
the generalized inverseness notion of an adjoint pair of terms
$\term{y}{r \dashv s}{x}$ forms a kind of polar-tension structure,
since there is only one possible right adjoint $r \dashv s$ iff $s = r^{\rm op}$.
However,
the generalized inverseness notion of an orthogonal pair of terms $\term{y}{r \orthogonal s}{x}$
does not form a polar-tension structure.
But we can make orthogonality that by assuming the existence of tensor negations:
$\term{y}{r \orthogonal \tneg{}{r}{}}{x}$ forms a kind of polar-tension structure,
since there is only one possible tensor negation $r \orthogonal s$ iff $s \preceq \tneg{}{r}{}$.
A subtype $\term{y}{i \dashv p}{x}$ has only one kind of complement
$\tneg{}{i}{} = p = i^{\rm op}$ and $\tneg{}{p}{} = \tneg{}{(i^{\rm op})}{} = i$,
whereas a functional {\bf H}-term $\term{y}{f}{x}$ has two kinds of complements:
its tensor negation $\term{x}{\tneg{}{f}{}}{y}$ and its right adjoint $\term{x}{f^{\rm op}}{y}$.
In general,
these two complements are related by
$\tneg{}{f}{} \preceq f^{\rm op} = x \tensorimplytarget f$
and
$\tneg{}{(f^{\rm op})}{} \preceq f = f^{\rm op} \tensorimplysource x$.
The two complements are identical $\tneg{}{f}{} = f^{\rm op}$
iff $\term{y}{f \dashv f^{\rm op}}{x}$ is a subtype.

A Heyting term $\term{y}{r}{x}$ is quasisymmetric precisely when
the left and right orthogonal duals coincide and equal the tensor negation
$\tneg{}{r}{} = x \tensorimplytarget r = r \tensorimplysource y$,
since 
 $s \circ r \preceq x$
 iff
 $s \preceq x \tensorimplytarget r$
 iff
 $s \preceq r \tensorimplysource y$
 iff
 $r \circ s \preceq y$.
For a quasisymmetric functional term $\term{y}{f}{x}$,
the two kinds of complements, tensor negation and right adjoint, are one:
$\tneg{}{f}{} = f^{\rm op}$ and $f = \tneg{}{f^{\rm op}}{}$;
so that,
$\term{y}{f \dashv f^{\rm op}}{x}$ is a subtype.
This is an indication that quasisymmetry is a very strong and restrictive concept.
This should be compared with the result in the object aspect of dialectical logic,
that ``quasisymmetry is equivalent to topological dialecticality''.
Tensor negation is contravariant lax functorial
$\tneg{}{r}{} \circ \tneg{}{s}{} \preceq_{x,z} \tneg{}{(s \circ r)}{}$,
so that tensor negation and tensor product are related by the inequalities
$s \circ r \preceq \tneg{}{\tneg{}{s}{}}{} \circ \tneg{}{\tneg{}{r}{}}{}
           \preceq \tneg{}{(\tneg{}{r}{} \circ \tneg{}{s}{})}{}$
and
$s \circ r \preceq \tneg{}{\tneg{}{(s \circ r)}{}}{}
           \preceq \tneg{}{(\tneg{}{r}{} \circ \tneg{}{s}{})}{}$.
A Heyting term $\term{y}{r}{x}$ is {\em coquasisymmetric\/} when
it is the tensor negation $r = \tneg{}{s}{}$ of a quasisymmetric term $\term{x}{s}{y}$.
This notion of Heyting coquasisymmetry is close to,
but not identical with,
the notion of biposet coquasisymmetry above.
However,
they agree on closed Heyting terms (see below).
By definition tensor negation maps quasisymmetric terms into coquasisymmetric terms.
A term $\term{y}{r}{x}$ is an {\bf H}-isomorphism iff its tensor negation is a categorical inverse:
$\tneg{}{r}{} \circ r = x$ and $r \circ \tneg{}{r}{} = y$.
Isomorphisms are both quasisymmetric and coquasisymmetric.
For isomorphisms the tensor implications are expressible as
$r \tensorimplysource t = \tneg{}{r}{} \circ t$ and $s \tensorimplytarget r = s \circ \tneg{}{r}{}$.

\paragraph{Double Negation.}
Let {\bf H} be a Heyting category.
Let $\tnegtneg{}{}{}$ symbolize double tensor negation,
defined by
$\tnegtneg{y}{r}{x} \define \tneg{x}{(\tneg{y}{r}{x})}{y}$
for any pair of types $y$ and $x$, and any term $\term{y}{r}{x}$.
Double negation $\tnegtneg{}{}{}$is a local closure operator:
``monotonic''  $r \preceq_{y,x} s$ implies $\tnegtneg{}{r}{} \preceq_{y,x} \tnegtneg{}{s}{}$,
``increasing'' $r \preceq_{y,x} \tnegtneg{}{r}{}$, and
``idempotent'' $\tnegtneg{}{(\tnegtneg{}{r}{})}{} = \tnegtneg{}{r}{}$.
A term $\term{y}{r}{x}$ is {\em double-negation closed\/}
when $r = \tnegtneg{}{r}{}$;
or equivalently,
when $r = \tneg{}{s}{}$ for some term $\term{x}{s}{y}$.
Denote the collection of closed terms in ${\bf H}[y,x]$ by $\tnegtneg{}{{\bf H}}{}[y,x]$.
Then $\tnegtneg{}{{\bf H}}{}[y,x]$ is a lattice,
which is a meet-subsemilattice of the lattice ${\bf H}[y,x]$
with meets in $\tnegtneg{}{{\bf H}}{}[y,x]$,
called classical boolean products,
identical $\bprod_i r_i = \wedge_i r_i$
to meets in ${\bf H}[y,x]$,
and joins in $\tnegtneg{}{{\bf H}}{}[y,x]$,
called classical boolean sums,
defined (following Glivenko) as the double negation $\bsum_i r_i = \tnegtneg{}{(\vee_i r_i)}{}$
of joins in ${\bf H}[y,x]$.
Double negation 
${\bf H}[y,x] \stackrel{\tnegtneg{}{}{}}{\rightarrow} \tnegtneg{}{{\bf H}}{}[y,x]$
reflects $\tnegtneg{}{}{} \dashv {\rm Inc}$ arbitrary Heyting terms into closed terms.
Identity terms (types) are closed,
since $x = \tneg{}{x}{}$.
The smallest and largest closed terms from $y$ to $x$ are
$0_{y,x} \define \tnegtneg{}{\bot_{y,x}}{}
              = \tneg{}{\top_{x,y}}{}$ and
$1_{y,x} \define \tnegtneg{}{\top_{y,x}}{}
              = \top_{y,x}
              = \tneg{}{\bot_{x,y}}{}
              = \tneg{}{0_{x,y}}{}$,
respectively.
If {\bf H} is a quasisymmetric category,
then all functional terms are subtypes,
all subtypes are double-negation closed,
its functional part ${\bf H}^\dashv$ is a ``preorderlike'' category
consisting only of subtype terms $\term{y}{i \dashv p}{x}$,
and the dialectical base ${\bf H}^\dashv \stackrel{\scriptbf{H}}{\longrightarrow} {\bf adj}$
is an ``extension/restriction'' base
with direct image ${\bf H}[y,y] \stackrel{p \circ (\,) \circ i}{\longrightarrow} {\bf H}[x,x]$ being ``extension to $x$''
and inverse image ${\bf H}[y,y] \stackrel{i \circ (\,) \circ p}{\longleftarrow} {\bf H}[x,x]$ being ``restriction to $y$''.
So,
if we are interested in a general notion of ``functionality'' in Heyting categories
(such as ordinary functions in {\bf Rel} or functors in {\bf Cat}),
then we should not assume quasisymmetry.

If $\term{y}{r}{x}$ is a quasisymmetric term,
then $\tnegtneg{}{r}{} = [y \tensorimplytarget (r \tensorimplysource y)] \wedge [(x \tensorimplytarget r) \tensorimplysource x]$
(in a quasisymmetric category
 $\tnegtneg{}{r}{} = y \tensorimplytarget (r \tensorimplysource y) = (x \tensorimplytarget r) \tensorimplysource x$).
If $\term{y}{r}{x}$ is quasisymmetric,
then $\tnegtneg{}{r}{}$ is also quasisymmetric,
  since $p \circ \tnegtneg{}{r}{} \preceq x$
implies $p \circ r \preceq x$
    iff $p \preceq \tneg{}{r}{} = \tneg{}{\tnegtneg{}{r}{}}{}$
implies $\tnegtneg{}{r}{} \circ p \preceq y$.
\begin{Lemma}[Functoriality]
   Double negation is lax functorial on quasisymmetric terms:
   $\tnegtneg{}{s}{} \circ \tnegtneg{}{r}{} \preceq_{z,x} \tnegtneg{}{(s \circ r)}{}$
   for all composable pairs of quasisymmetric terms $\term{z}{s}{y}$ and $\term{y}{r}{x}$.
\end{Lemma}
\begin{proof}
   We prove something equivalent:
   for all composable pairs of quasisymmetric terms $\term{z}{s}{y}$ and $\term{y}{r}{x}$,
   $s \circ \tnegtneg{}{r}{} \preceq_{z,x} \tnegtneg{}{(s \circ r)}{}$ when $s$ is double negation closed.
   By modus ponens on left and right
   $((x \tensorimplytarget r) \tensorimplytarget s) \circ s \circ ((x \tensorimplytarget r) \tensorimplysource x) \preceq x$.
   So (1)
   $s \circ ((x \tensorimplytarget r) \tensorimplysource x) 
    \preceq_{z,x} ((x \tensorimplytarget r) \tensorimplytarget s) \tensorimplysource x
    = (x \tensorimplytarget (s \circ r)) \tensorimplysource x$. 
   On the other hand
   $(y \tensorimplytarget (r \tensorimplysource y)) \circ (r \tensorimplysource \tneg{}{s}{}) \circ (\tneg{}{s}{} \tensorimplysource y)
    \preceq y$ by transitivity (used twice).
   But $s = \tnegtneg{}{s}{} \preceq \tneg{}{s}{} \tensorimplysource y$
   since $s$ is closed and quasisymmetric.
   So
   $(y \tensorimplytarget (r \tensorimplysource y)) \circ (r \tensorimplysource \tneg{}{s}{}) \circ s \preceq y$.
   Again since $s$ is quasisymmetric
   $s \circ (y \tensorimplytarget (r \tensorimplysource y)) \circ (r \tensorimplysource \tneg{}{s}{}) \preceq z$.
   Hence, (2)
   $s \circ (y \tensorimplytarget (r \tensorimplysource y))
    \preceq z \tensorimplytarget (r \tensorimplysource \tneg{}{s}{})
          = z \tensorimplytarget (r \tensorimplysource (s \tensorimplysource z))
          = z \tensorimplytarget ((s \circ r) \tensorimplysource z)$.
   Putting both facts together
   $s \circ \tnegtneg{}{r}{}
          = s \circ [y \tensorimplytarget (r \tensorimplysource y)] \wedge [(x \tensorimplytarget r) \tensorimplysource x]
    \preceq [s \circ (y \tensorimplytarget (r \tensorimplysource y))] \wedge [s \circ ((x \tensorimplytarget r) \tensorimplysource x)]
    \preceq [z \tensorimplytarget ((s \circ r) \tensorimplysource z)] \wedge [(x \tensorimplytarget (s \circ r)) \tensorimplysource x]
          = \tnegtneg{}{(s \circ r)}{}$.
   Finally,
   $\tnegtneg{}{s}{} \circ \tnegtneg{}{r}{}
    \preceq_{z,x} \tnegtneg{}{(\tnegtneg{}{s}{} \circ r)}{}
    \preceq_{z,x} \tnegtneg{}{(\tnegtneg{}{(s \circ r)}{})}{}
    = \tnegtneg{}{(s \circ r)}{}$
   by monotonicity and idempotency of $\tnegtneg{}{}{}$.  
\end{proof}
By rights this functoriality lemma should be called the ``bottleneck lemma''
since we need it \cite{Girard} to prove associativity of the classical tensors defined below.
The concept of quasisymmetry,
although quite natural by itself,
was motivated by this lemma.

Following Glivenko,
in analogy with the definition of the classical boolean connectives,
the tensor connectives for classical dialectical logic,
classical tensor product $\tprod$ and classical tensor sum $\tsum$,
are definable in terms of the Heyting tensor product $\circ$ and tensor negation $\tneg{}{}{}$.
For any two $\circ$-composable terms $\term{z}{s}{y}$ and $\term{y}{r}{x}$
the tensor product term $\term{z}{s \tprod r}{x}$ and the tensor sum term $\term{z}{s \tsum r}{x}$
are $\tnegtneg{}{}{}$-closed terms define by
    $s \tprod r \define \tnegtneg{}{(s \circ r)}{}$
and $s \tsum r \define \tneg{}{(\tneg{}{r}{} \tprod \tneg{}{s}{})}{}
                    = \tneg{}{(\tneg{}{r}{} \circ \tneg{}{s}{})}{}$.
For all terms we immediately have the DeMorgans laws
$\tneg{}{(s \tsum r)}{} = \tneg{}{r}{} \tprod \tneg{}{s}{}$ and
$\tneg{}{(s \bsum r)}{} = \tneg{}{s}{} \bprod \tneg{}{r}{}$,
for $\zenter{{\bf H}}$-open terms we have the DeMorgans inequalities
$\tneg{}{(s \tprod r)}{} \preceq \tneg{}{r}{} \tsum \tneg{}{s}{}$ and
$\tneg{}{(s \bprod r)}{} \preceq \tneg{}{s}{} \bsum \tneg{}{r}{}$,
and for $\tnegtneg{}{}{}$-closed terms we have the DeMorgans laws
$\tneg{}{(s \tprod r)}{} = \tneg{}{r}{} \tsum \tneg{}{s}{}$ and
$\tneg{}{(s \bprod r)}{} = \tneg{}{s}{} \bsum \tneg{}{r}{}$.

A Heyting term is {\em polar\/} when it is $\tnegtneg{}{}{}$-closed and $\zenter{{\bf H}}$-open;
that is,
when the term is in $\tnegtneg{}{\zenter{{\bf H}}}{}$.
The {\em pole\/} of any Heyting term
is the double negation of its $\zenter{{\bf H}}$-interior (if it exists).
The lax functoriality of double negation $\tnegtneg{}{}{}$ implies that
the classical tensor product is associative
$t \tprod (s \tprod r) = (t \tprod s) \tprod r$
on polar terms.
Also,
types are identities $y \tprod r = r = r \tprod x$ on polar terms.
The {\em Boolean pole\/} of $\zenter{{\bf H}}$,
denoted by $\zenter{{\bf H}}_{\tprod}^{\bsum}$,
is the join bisemilattice 
$\zenter{{\bf H}}_{\tprod}^{\bsum} = \triple{\quadruple{\tnegtneg{}{\zenter{{\bf H}}}{}}{\preceq}{\tprod}{\rm Id}}{\bsum}{0}$
consisting of all types and all polar terms
(join bisemilattice since finite homset joins exist, but not necessarily finite homset meets),
with the classical tensor product and boolean sum.
$\zenter{{\bf H}}_{\tprod}^{\bsum}$ is a lax (Heyting) subcategory of $\zenter{{\bf H}}$.
Dually,
a Heyting term is {\em antipolar\/} when it is $\tnegtneg{}{}{}$-closed and $\zenter{{\bf H}}$-closed;
that is,
when it is the tensor negation of a polar term.
The image $\tneg{}{\zenter{{\bf H}}}{}$ of tensor negation on the pole
is the collection of all antipolar terms.
The tensor DeMorgans laws
(and the associativity of the tensor product $\tprod$)
imply that the classical tensor sum $\tsum$ is associative
$t \tsum (s \tsum r) = (t \tsum s) \tsum r$
on antipolar terms.
Also,
types are identities $y \tsum r = r = r \tsum x$ on antipolar terms.
The {\em Boolean antipole\/} of $\zenter{{\bf H}}$,
denoted by $\zenter{{\bf H}}_{\tsum}^{\bprod}$,
is the meet bisemilattice
$\zenter{{\bf H}}_{\tsum}^{\bprod} = \triple{\quadruple{\tneg{}{\zenter{{\bf H}}}{}}{\preceq}{\tsum}{\rm Id}}{\bprod}{1}$
consisting of all types and all antipolar terms,
and the classical tensor sum and boolean product.
Moreover,
tensor negation is a 2-involution,
a morphism of join bisemilattices
$\zenter{{\bf H}}_{\tprod}^{\bsum} \stackrel{\tneg{}{}{}}{\rightarrow} {\zenter{{\bf H}}_{\tsum}^{\bprod}}^{\rm coop}$
and a morphism of meet bisemilattices
${\zenter{{\bf H}}_{\tprod}^{\bsum}}^{\rm coop} \stackrel{\tneg{}{}{}}{\leftarrow} \zenter{{\bf H}}_{\tsum}^{\bprod}$:
$\tneg{}{}{}$ is self-inverse $\tneg{}{\tneg{}{r}{}}{} = r$,
$\tneg{}{x}{} = x$,
$\tneg{}{}{}$ switches source and target $\tneg{}{(\term{y}{r}{x})}{} = \term{x}{\tneg{}{r}{}}{y}$,
and $\tneg{}{}{}$ is (contravariant) monotonic on homsets $r \preceq_{y,x} s$ implies $\tneg{}{s}{} \preceq_{x,y} \tneg{}{r}{}$.
This complex,
consisting of a join and meet bisemilattice and the negation involution between them,
is called the {\em Boolean\/} of $\zenter{{\bf H}}$ or the {\em Boolean center\/} of {\bf H},
and is denoted by $\Boolean{\zenter{{\bf H}}}$.

The special property $s \orthog{\tprod} r$ iff $s \preceq \tneg{}{r}{}$
called the {\em orthogonality-entailment axiom\/},
which relates term-orthogonality with term-order,
holds for all polar terms.
Equivalently,
the special property $s \coorthog{\tsum} r$ iff $\tneg{}{s}{} \preceq r$,
which relates term-coorthogonality with term-order,
holds for all antipolar terms.
The Boolean center $\Boolean{\zenter{{\bf H}}}$ is quasisymmetric:
the Boolean pole $\zenter{{\bf H}}_{\tprod}^{\bsum}$ is a quasisymmetric category
since a Heyting term $\term{y}{r}{x}$ is $\circ$-quasisymmetric iff it is $\tprod$-quasisymmetric,
and the Boolean antipole $\zenter{{\bf H}}_{\tsum}^{\bprod}$ is a coquasisymmetric category
since a Heyting term $\term{y}{r}{x}$ being $\circ$-coquasisymmetric implies that it is $\tsum$-coquasisymmetric.
For any pair of terms in either the pole or the antipole of the Boolean center,
the Heyting tensor product and the classical tensor connectives
are arranged as $s \circ r \preceq s \tprod r \preceq s \tsum r$.
When {\bf H} is quasisymmetric the Boolean center $\Boolean{{\bf H}}$ consists of all $\tnegtneg{}{}{}$-closed terms.

A {\em polarized bisemilattice\/} {\bf P} consists of two bisemilattices,
a join bisemilattice
${\bf P}_{\tprod}^{\bsum} = \triple{\quadruple{{\bf P}_{\tprod}^{\bsum}}{\preceq_{\tprod}}{\tprod}{\rm Id}}{\bsum}{0}$
and a meet bisemilattice
${\bf P}_{\tsum}^{\bprod} = \triple{\quadruple{{\bf P}_{\tsum}^{\bprod}}{\preceq_{\tsum}}{\tsum}{\rm Id}}{\bprod}{1}$,
called the {\em pole\/} and {\em antipole\/} of {\bf P} respectively,
and two morphisms of bisemilattices,
a morphism of join bisemilattices
${\bf P}_{\tprod}^{\bsum} \stackrel{\tneg{}{}{}}{\rightarrow} {{\bf P}_{\tsum}^{\bprod}}^{\rm coop}$
and a morphism of meet bisemilattices
${{\bf P}_{\tprod}^{\bsum}}^{\rm coop} \stackrel{\tneg{}{}{}}{\leftarrow} {\bf P}_{\tsum}^{\bprod}$
which are inverse $\tneg{}{}{} \cdot \tneg{}{}{}^{\rm coop} = {\rm Id}$ to each other.
Just as for Heyting categories,
objects and arrows in either the pole ${\bf P}_{\tprod}^{\bsum}$ or the antipole ${\bf P}_{\tsum}^{\bprod}$
are called {\em types\/} and {\em terms\/},
respectively.
The Boolean center $\Boolean{\zenter{{\bf H}}}$ of any Heyting category {\bf H} is a polarized bisemilattice.
Morphisms of polarized bisemilattices
can be defined in either a polar or an antipolar sense.
A {\em morphism of polarized bisemilattices\/} ${\bf P} \stackrel{H}{\rightarrow} {\bf Q}$
consists of a morphism of join bisemilattices
${\bf P}_{\tprod}^{\bsum} \stackrel{H_{\tprod}^{\bsum}}{\rightarrow} {\bf Q}_{\tprod}^{\bsum}$ 
called the {\em pole\/} of $H$,
and a morphism of meet bisemilattices
${\bf P}_{\tsum}^{\bprod} \stackrel{H_{\tsum}^{\bprod}}{\rightarrow} {\bf Q}_{\tsum}^{\bprod}$
called the {\em antipole\/} of $H$,
which are interdefinable with
$H_{\tsum}^{\bprod} \define \tneg{P}{}{} \cdot (H_{\tprod}^{\bsum})^{\rm coop} \cdot (\tneg{Q}{}{})^{\rm coop}$
and
$H_{\tprod}^{\bsum} \define \tneg{P}{}{} \cdot (H_{\tsum}^{\bprod})^{\rm coop} \cdot (\tneg{Q}{}{})^{\rm coop}$.

\paragraph{Boolean Categories.}
\begin{figure}
   \begin{center}
   \begin{tabular}{|r|c|c|} \hline
                                               & \multicolumn{1}{c|}{{\bf Intuitionistic}} & \multicolumn{1}{c|}{{\bf Classical}} \\ \hline\hline
      {\bf Standard Logic}                     & Heyting algebras                          & Boolean algebras                     \\
                                               & (in particular, subset algebras)          &                                      \\ \hline
      {\bf Linear Logic}                       & commutative Heyting monoids               & commutative Boolean monoids          \\
                                               & (in particular, ``phase spaces'')         &                                      \\ \hline
      {\bf Dialectical Logic}                  & Heyting categories                        & (quasisymmetric) Boolean categories  \\
      \multicolumn{1}{|c|}{(this paper)}       & (in particular, subset categories)        &                                      \\ \hline
      {\bf Dialectical Logic}                  & Heyting categories                        & (quasisymmetric) Boolean categories  \\
      \multicolumn{1}{|c|}{(extended version)} & with type sums                            & with type sums                       \\ 
                                               & (in particular, distributor categories)   &                                      \\ \hline
   \end{tabular}
   \end{center}
   \caption{Semantic Domains for various Logics \label{semdom}}
\end{figure}
Ignoring idempotency and commutativity,
a Boolean algebra $B = \langle B,\leq,\wedge,\vee,1,0,\neg \rangle$
can be viewed as two monoidal semilattices,
a monoidal join semilattice
$B_\wedge^\vee = \triple{\quadruple{B}{\leq}{\wedge}{1}}{\vee}{0}$
and a monoidal meet semilattice
$B_\vee^\wedge   = \triple{\quadruple{B}{\leq}{\vee}{0}}{\wedge}{1}$
on an underlying poset $\langle B,\leq \rangle$
with negation $\neg$ being an internal involution:
a monoidal join semilattice morphism
$B_\wedge^\vee \stackrel{\neg}{\rightarrow} {B_\vee^\wedge}^{\rm coop}$,
$b \leq b'$ implies $\neg b' \leq \neg b$,
$\neg(c \wedge b) = (\neg c) \vee (\neg b)$,
$\neg 1 = 0$, 
$\neg(b \vee b') = (\neg b) \wedge (\neg b')$ and
$\neg 0 = 1$,
and a monoidal meet semilattice morphism
${B_\wedge^\vee}^{\rm coop} \stackrel{\neg}{\leftarrow} B_\vee^\wedge$,
which is self-inverse $\neg(\neg b) = b$ or $\neg \cdot \neg^{\rm coop} = {\rm Id}$.
More generally,
a {\em Boolean category\/} {\bf B} is a polarized bisemilattice
for which the term-sets, type-sets and homset-order of the pole and the antipole coincide
${\rm Ar}({\bf B}) = {\rm Ar}({\bf B}_{\tprod}^{\bsum}) = {\rm Ar}({\bf B}_{\tsum}^{\bprod})$,
${\rm Obj}({\bf B}) = {\rm Obj}({\bf B}_{\tprod}^{\bsum}) = {\rm Obj}({\bf B}_{\tsum}^{\bprod})$ and
$\preceq_{\tprod} = \preceq_{\tsum} = \preceq$
(and are not just isomorphic as in polarized bisemilattices,
where the term-sets and type-sets are not identical,
but only in bijective correspondence via negation),
and which satisfies the {\em orthogonality-entailment axiom\/}
\[ s \orthog{\tprod} r \mbox{ iff } s \preceq \tneg{}{r}{} \]
for all opposed terms $\term{y}{r}{x}$ versus $\opterm{y}{s}{x}$,
which relates term-orthogonality with term-order
(because of the precise duality expressed through tensor negation,
 $s \orthog{\tprod} r$ iff $\tneg{}{r}{} \coorthog{\tsum} \tneg{}{s}{}$,
 polar orthogonality can be expressed as, and is equivalent to,
 antipolar coorthogonality).

In more detail,
a Boolean category {\bf B} consists of
a set of types (objects) ${\rm Type}({\bf B})$,
a set of terms (arrows) ${\rm Term}({\bf B})$ ordered type-wise by
a partial order $\preceq$ which has homset lattice join $\bsum$ and homset lattice meet $\bprod$
and two category compositions $\tprod$ and $\tsum$,
where the pole
${\bf B}_{\tprod}^{\bsum} = \triple{\quadruple{{\bf B}}{\preceq}{\tprod}{\rm Id}}{\bsum}{0}$
and the antipole
${\bf B}_{\tsum}^{\bprod} = \triple{\quadruple{{\bf B}}{\preceq}{\tsum}{\rm Id}}{\bprod}{1}$
are join and meet bisemilattices,
respectively,
with an internal 2-involution
${\bf B}_{\tprod}^{\bsum} \stackrel{\tneg{}{}{}}{\rightarrow} {{\bf B}_{\tsum}^{\bprod}}^{\rm coop}$.
A Boolean category is finitely distributive in two senses:
from the left
$s \tprod (\bsum_i r_i) = \bsum_i (s \tprod r_i)$ in ${\bf B}_{\tprod}^{\bsum}$ and
$s \tsum (\bprod_i r_i) = \bprod_i (s \tsum r_i)$ in ${\bf B}_{\tsum}^{\bprod}$,
and also from the right in both poles.
The tensor negation is 
(1) a doubly-contravariant (everything ``flips'') morphism of join bisemilattices
${\bf B}_{\tprod}^{\bsum} \stackrel{\tneg{}{}{}}{\rightarrow} {{\bf B}_{\tsum}^{\bprod}}^{\rm coop}$
identity on types,
$\tneg{}{(y \stackrel{r}{\rightarrow} x)}{} = x \stackrel{\tneg{}{r}{}}{\rightarrow} y$,
$\tneg{}{(s \tprod r)}{} = (\tneg{}{r}{}) \tsum (\tneg{}{s}{})$,
$\tneg{}{x}{} = x$,
$r \preceq_{y,x} r'$ implies $\tneg{}{r'}{} \preceq_{x,y} \tneg{}{r}{}$ and
$\tneg{}{(r \bsum r')}{} = (\tneg{}{r}{}) \bprod (\tneg{}{r'}{})$;
(2) a doubly-contravariant morphism of meet bisemilattices
${{\bf B}_{\tprod}^{\bsum}}^{\rm coop} \stackrel{\tneg{}{}{}}{\leftarrow} {\bf B}_{\tsum}^{\bprod}$
in the reverse direction and opposite sense,
$\tneg{}{(s \tsum r)}{} = (\tneg{}{r}{}) \tprod (\tneg{}{s}{})$ and
$\tneg{}{(r \bprod r')}{} = (\tneg{}{r}{}) \bsum (\tneg{}{r'}{})$;
(3) which is self-inverse $\tneg{}{(\tneg{}{r}{})}{} = r$.
In a Boolean category orthogonality preserves composition,
in the sense that:
$q \orthogonal s$ and $p \orthogonal r$ implies $(p \tprod q) \orthogonal (s \tprod r)$.
Also,
a Boolean category satisfies the product-sum comparison (or ``mix'') axiom:
$s \tprod r \preceq_{z,x} s \tsum r$
for all terms $\term{z}{s}{y}$ and $\term{y}{r}{x}$.
A one object Boolean category is called a {\em Boolean monoid\/}.
The homsets ${\bf B}[x,x]$ are Booleans monoids for each type $x$.
A Boolean category is {\em complete\/}
when the poles are both complete Heyting categories;
that is,
the homsets are complete lattices,
tensor product is completely distributive (continuous) w.r.t. boolean sum, and
tensor sum is completely distributive (continuous) w.r.t. boolean product.
Morphisms of Boolean categories are just morphisms of polarized bisemilattices.

A term $\term{y}{r}{x}$ in a Boolean category is {\em invertible\/} when its tensor negation is a categorical inverse:
$\tneg{}{r}{} \tprod r = x$ and $r \tprod \tneg{}{r}{} = y$.
So invertible terms are the same as {\bf B}-isomorphisms.
For isomorphisms the direct and inverse image operators are isomorphisms of Boolean monoids.
Clearly, all identities are isomorphisms.
Isomorphisms are closed under tensor product, tensor sum and tensor negation.
In fact,
the tensor sum collapses to the tensor product
$s \tsum r = s \tprod r$ for composable isomorphisms. 
When all terms in a Boolean category are isomorphisms,
the Boolean category is known as a {\em lattice-ordered groupoid\/}.
In general,
the collection of all isomorphisms in a Boolean category {\bf B}
is a Boolean subcategory of {\bf B} which is a lattice-ordered groupoid.
A summary of the appropriate semantic domains for various logics is given in the Figure~\ref{semdom}.

Recall that a term $\term{y}{r}{x}$ is ${\bf B}_{\tprod}^{\bsum}$-quasisymmetric when
$p \tprod r \preceq x$ iff $r \tprod p \preceq y$,
and is ${\bf B}_{\tsum}^{\bprod}$-coquasisymmetric when
$p \tsum r \succeq x$ iff $r \tsum p \succeq y$.
So $r$ is ${\bf B}_{\tprod}^{\bsum}$-quasisymmetric iff $\tneg{}{r}{}$ is ${\bf B}_{\tsum}^{\bprod}$-coquasisymmetric.
This means that the tensor negation 2-involution restricts and corestricts precisely
to the center of ${\bf B}_{\tprod}^{\bsum}$ and the cocenter of ${\bf B}_{\tsum}^{\bprod}$:
$\zenter{{\bf B}_{\tprod}^{\bsum}} \stackrel{\tneg{}{}{}}{\rightarrow} {\zenter{{\bf B}_{\tsum}^{\bprod}}}^{\rm coop}$.
Call this the {\em center\/} of {\bf B},
and denote it by $\zenter{{\bf B}}$.
A Boolean category {\bf B} is {\em quasisymmetric\/} when $\zenter{{\bf B}} = {\bf B}$.
Quasisymmetric Boolean categories
(and the Boolean center of their associated closed subset categories) 
are fundamental semantic structures for complete classical dialectical logic.

Let $\term{y}{r}{x}$ be any fixed ${\bf B}^{\tprod}_{\bsum}$-term.
For any ${\bf B}^{\tsum}_{\bprod}$-term $\term{y}{t}{z}$ with source type in common with $r$,
define the {\em left tensor implication\/} ${\bf B}^{\tsum}_{\bprod}$-term
$\term{x}{r \tensorimplysource t}{z}$ by $r \tensorimplysource t \define \tneg{}{r}{} \tsum t$.
Similarly,
for any ${\bf B}^{\tsum}_{\bprod}$-term $\term{z}{s}{x}$ with target type in common with $r$
define the {\em right tensor implication\/} ${\bf B}^{\tsum}_{\bprod}$-term
$\term{z}{s \tensorimplytarget r}{y}$ by $s \tensorimplytarget r \define s \tsum \tneg{}{r}{}$.
The dialectical axioms
$t \tprod r \preceq_{z,x} s \mbox{ iff } t \preceq_{z,y} s \tensorimplytarget r$
and
$r \tprod s \preceq_{y,z} t \mbox{ iff } s \preceq_{x,z} r \tensorimplysource t$
hold on quasisymmetric terms.
Adjoining these implication operators to the center pole ${\cal Z}({\bf B}^{\tprod}_{\bsum})$
makes this into a quasisymmetric Heyting category ${\cal H}(\zenter{{\bf B}})$
called the {\em Heyting center\/} of {\bf B},
whose tensor negation is the same as in {\bf B}.
So all terms in ${\cal H}(\zenter{{\bf B}})$ are double negation closed.
\begin{Theorem}[Center Reflection]
   If ${\bf H}$ is a quasisymmetric Heyting category,
   then the Boolean center $\Boolean{{\bf H}}$ is a quasisymmetric Boolean category.
   Any quasisymmetric Boolean category ${\bf B}$ is a quasisymmetric Heyting category $\Heyting{{\bf B}}$.
   For any quasisymmetric Boolean category {\bf B},
   the Boolean center of {\bf B} as a Heyting category is just {\bf B} itself
   $\Boolean{\Heyting{{\bf B}}} = {\bf B}$.
   For any quasisymmetric Heyting category ${\bf H}$,
   the Boolean center as a Heyting category,
   is just the center pole
   $\Heyting{\Boolean{{\bf H}}} = {\bf H}_{\tprod}^{\bsum}$,
   the lax subHeyting category of ${\bf H}$ consisting of double negation closed terms.
\end{Theorem}

\section{Classical Axiomatics}

We follow both the semantics of dialectical processes and the axiomatics given by Girard for linear logic.
However,
when linear logic deviates from dialectical process semantics,
we follow the latter.
A hallmark of both dialectical and linear logic is the fact that
the standard connectives and truth-values split into tensors and booleans,
as in Table~\ref{split}.

\begin{table}
   \begin{center}
      \begin{tabular}{|rp{1in}|rp{1.8in}|p{1.6in}|} \hline
         \multicolumn{2}{|c|}{\bf Standard Logic} & \multicolumn{2}{c|}{\bf Dialectical Logic} & \multicolumn{1}{c|}{\bf Uses} \\ \hline\hline
         $\wedge$ & boolean product & $\tprod_{z,y,x}$ & tensor (horizontal) product & direct flow                   \\ \cline{3-5}
                  &                 & $\bprod_{y,x}$   & boolean (vertical) product  & parallelism \& inverse flow   \\ \hline
         $\top$   & true            & $\pair{m}{x}$    & monoids (comonoids)                      & tensor validity               \\ \cline{3-5}
                  &                 & $1_{y,x}$        & top process                 & boolean validity              \\ \hline
         $\vee$   & boolean sum     & $\tsum_{z,y,x}$  & tensor (horizontal) sum     & inverse flow                  \\ \cline{3-5}
                  &                 & $\bsum_{y,x}$    & boolean (vertical) sum      & parallelism \& direct flow    \\ \hline
         $\bot$   & false           & $\pair{m}{x}$    & monoids (comonoids)                    & orthogonality                 \\ \cline{3-5}
                  &                 & $0_{y,x}$        & bottom process              & disjointness                  \\ \hline
      \end{tabular}
   \end{center}
   \caption{Splitting of Connectives and Truth values \label{split}}
\end{table}

\paragraph{Language.}
There is a collection of {\em type symbols\/} $x,y,z, \cdots$,
and a collection of {\em atoms\/} or {\em atomic term symbols\/} $a,b,c, \cdots$.
Each atom $a$ is a term formula,
and has a unique source type $y$ and a unique target type $x$,
denoted by $\term{y}{a}{x}$.
Each atom $\term{y}{a}{x}$ has a {\em dual\/} or {\em complement\/} $\term{x}{\dot{a}}{y}$. 
Atoms and their duals are called {\em literals\/}.
So type symbols are the nodes of a graph {\bf Lang},
and literals (and other composite term formulas) form the edges.
For each pair of types $y$ and $x$,
there are two distinguished term symbols $\term{y}{0}{x}$ and $\term{y}{1}{x}$.
Each type $x$ is represented as a term formula $\term{x}{x}{x}$,
which is a self-loop at node $x$ in the graph {\bf Lang}.
Composite term formulas are built up recursively from literals
by horizontally applying the tensor operation symbols $\tprod$ and $\tsum$,
and vertically applying the boolean operation symbols $\bsum$ and $\bprod$,
in an obvious type-consistent fashion. 
Term formulas are also called terms.
This will be legitimized below when it is shown that the (equivalence classes of) term formulas form a Boolean category. 
Following Girard's approach,
there is an external involution ${\bf Lang} \stackrel{\tneg{}{}{}}{\rightarrow} {\bf Lang}^{\rm op}$ called {\em tensor negation\/},
which is defined recursively on terms as follows:
{\bf base}      $\tneg{}{a}{} \define \dot{a}$
             and $\tneg{}{(\dot{a})}{} \define a$;
{\bf recursion} $\tneg{}{x}{} \define x$,
                 $\tneg{}{(\beta \tprod \alpha)}{} \define (\tneg{}{\alpha}{}) \tsum (\tneg{}{\beta}{})$
             and $\tneg{}{(\beta \tsum \alpha)}{} \define (\tneg{}{\alpha}{}) \tprod (\tneg{}{\beta}{})$,
                 $\tneg{}{(\alpha \bsum \alpha')}{} \define (\tneg{}{\alpha}{}) \bprod (\tneg{}{\alpha'}{})$
             and $\tneg{}{(\alpha \bprod \alpha')}{} \define (\tneg{}{\alpha}{}) \bsum (\tneg{}{\alpha'}{})$, and
                 $\tneg{}{(\term{y}{0}{x})}{} \define \term{x}{1}{y}$
             and $\tneg{}{(\term{y}{1}{x})}{} \define \term{x}{0}{y}$.
\begin{Fact}
   $\tneg{}{(\tneg{}{\alpha}{})}{} = \alpha$ for every term $\alpha$.
\end{Fact}
In addition to the previous symbols which specify types and terms,
there are two special symbols $\vdash$ and $\perp$ which specify
the binary relation of {\em entailment\/} between parallel terms
and the binary relation of {\em orthogonality\/} between opposed terms,
respectively.
The entailment and orthogonality relations on terms give two equivalent ways in which to specify dialectical logic.

\paragraph{Inference Rules.}
The formal semantics of classical dialectical logic will be defined
via axioms and inference rules.
The novelty of this approach lies in the use of orthogonality assertions,
rather than just term entailment assertions alone. 
An orthogonality assertion is a statement of the form $\beta \orthog{} \alpha$
for two opposed terms $\term{y}{\alpha}{x}$ versus $\opterm{y}{\beta}{x}$,
and when $\beta \orthog{} \alpha$ holds,
we say that $\alpha$ is {\em orthogonal\/} to $\beta$.
An orthogonality assertion is interpreted as the orthogonality of the terms specified by the opposed term formulas.
The orthogonality relation $\perp$ has a negation-dual relation $\coorthog{}$,
called {\em coorthogonality\/},
and defined by $\beta \coorthog{} \alpha$ when $\tneg{}{\alpha}{} \orthog{} \tneg{}{\beta}{}$.
An entailment assertion is a statement of the form $\alpha \vdash \beta$ 
for two parallel terms $\term{y}{\alpha,\beta}{x}$,
and when $\alpha \vdash \beta$ holds,
we say that $\alpha$ {\em entails\/} $\beta$.
The entailment relation $\vdash$ has an obvious dual relation $\vdash^{\rm op}$
defined by $\beta \vdash^{\rm op} \alpha$ when $\alpha \vdash \beta$;
so that,
$\vdash^{\rm op} = \dashv$.
We use the equivalence notation
$\alpha \logequiv \beta$
when both
$\alpha \vdash \beta$ and $\beta \vdash \alpha$ hold,
and we say that $\alpha$ is {\em entailment equivalent\/} to $\beta$.
When ``$\alpha$ entails identity'',
that is when $\alpha \vdash x$ holds,
we say that the term $\alpha$ itself is {\em provable\/}.
So an endoterm $\term{x}{\alpha}{x}$ is provable
iff $\alpha \memberof \below{x}$ the principal ideal of the identity term.

We give two versions of inference rules for the term calculus:
an {\em entailment version\/} which is closely related to the semantics of dialectical logic,
and an {\em orthogonality version\/} which extends Girard's version \cite{Girard} of the linear logic.
In each version we group the rules according to their semantics:
the vertical aspect in Table~\ref{vertterm}
\begin{table}
   \begin{center}
      \begin{tabular}{|c|c|c|} 
         \cline{1-1} \cline{3-3}
         {\sc entailment version}
         &&
         {\sc orthogonality version} \\
         \cline{1-1} \cline{3-3}
         \multicolumn{3}{c}{} \\
         \cline{1-1} \cline{3-3}
         && \\
         \multicolumn{1}{|c}{} & \multicolumn{1}{c}{\makebox[0in]{\bf Homset Order}} & \\
         && \\
         \fbox{\begin{tabular}{c}
                  $\rulezero{\alpha \vdash \alpha}
                            {{\bf reflexivity}}$ \\
                  for terms $\term{y}{\alpha}{x}$
               \end{tabular}}
         &&
         \fbox{\begin{tabular}{c}
                  $\rulezero{\alpha \orthogonal \tneg{}{\alpha}{}}
                            {{\bf logical axiom}}$ \\
                  for terms $\term{y}{\alpha}{x}$
               \end{tabular}} \\
         && \\
         \fbox{\begin{tabular}{c}
                  $\ruletwo{\alpha \vdash \beta}
                           {\beta \vdash \gamma}
                           {\alpha \vdash \gamma}
                           {{\bf transitivity}}$ \\
                  for terms $\term{y}{\alpha,\beta}{x}$ versus $\opterm{y}{\gamma}{x}$
               \end{tabular}}
         &&
         \fbox{\begin{tabular}{c}
                  $\ruletwo{\alpha \orthogonal \tneg{}{\beta}{}}
                           {\beta \orthogonal \gamma}
                           {\alpha \orthogonal \gamma}
                           {{\bf cut}}$ \\
                  for terms $\term{y}{\alpha,\beta}{x}$ versus $\opterm{y}{\gamma}{x}$
               \end{tabular}} \\
         && \\
         \fbox{\begin{tabular}{c}
                  $\ruleone{\alpha \vdash \beta}
                           {\tneg{}{\beta}{} \vdash \tneg{}{\alpha}{}}
                           {{\bf contravariance}}$ \\
                  for terms $\term{y}{\alpha}{x}$ versus $\opterm{y}{\beta}{x}$
               \end{tabular}}
         &&
         \fbox{\begin{tabular}{c}
                  $\ruleone{\beta \orthogonal \alpha}
                           {\alpha \orthogonal \beta}
                           {{\bf symmetry}}$ \\
                  for terms $\term{y}{\alpha}{x}$ versus $\opterm{y}{\beta}{x}$
               \end{tabular}} \\
         && \\
         \multicolumn{1}{|c}{} & \multicolumn{1}{c}{\makebox[0in]{\bf Booleans}} & \\
         && \\
         \fbox{\begin{tabular}{c}
                  $\rulezero{0_{yx} \vdash \alpha}
                            {{\bf bottom}}$ \\
                  for terms $\term{y}{\alpha}{x}$
               \end{tabular}}
         &&
         \fbox{\begin{tabular}{c}
                  $\rulezero{0_{yx} \orthogonal \alpha}
                            {{\bf zero}}$ \\
                  for terms $\term{y}{\alpha}{x}$
               \end{tabular}} \\
         && \\
         \fbox{\begin{tabular}{c}
                  $\rulezero{\alpha \vdash (\alpha \bsum \alpha')}
                           {{\bf 1st u.b.}}$ \\
                  for terms $\term{y}{\alpha,\alpha'}{x}$
               \end{tabular}}
         &&
         \fbox{\begin{tabular}{c}
                  $\ruleone{\alpha \orthogonal \beta}
                           {(\alpha \bprod \alpha') \orthogonal \beta}
                           {{\bf 1st} $\bprod$}$ \\
                  for terms $\term{y}{\alpha,\alpha'}{x}$ versus $\opterm{y}{\beta}{x}$
               \end{tabular}} \\
         && \\
         \fbox{\begin{tabular}{c}
                  $\rulezero{\alpha' \vdash (\alpha \bsum \alpha')}
                           {{\bf 2nd u.b.}}$ \\
                  for terms $\term{y}{\alpha,\alpha'}{x}$
               \end{tabular}}
         &&
         \fbox{\begin{tabular}{c}
                  $\ruleone{\alpha' \orthogonal \beta}
                           {(\alpha \bprod \alpha') \orthogonal \beta}
                           {{\bf 2nd} $\bprod$}$ \\
                  for terms $\term{y}{\alpha,\alpha'}{x}$ versus $\opterm{y}{\beta}{x}$
               \end{tabular}} \\
         && \\
         \fbox{\begin{tabular}{c}
                  $\ruletwo{\alpha \vdash \beta}
                           {\alpha' \vdash \beta}
                           {(\alpha \bsum \alpha') \vdash \beta}
                           {{\bf l.u.b.}}$ \\
                  for terms $\term{y}{\alpha,\alpha',\beta}{x}$
               \end{tabular}}
         &&
         \fbox{\begin{tabular}{c}
                  $\ruletwo{\alpha \orthogonal \beta}
                           {\alpha' \orthogonal \beta}
                           {(\alpha \bsum \alpha') \orthogonal \beta}
                           {$\bsum$}$ \\
                  for terms $\term{y}{\alpha,\alpha'}{x}$ versus $\opterm{y}{\beta}{x}$
               \end{tabular}} \\
         && \\
         \cline{1-1} \cline{3-3}
      \end{tabular}
   \end{center}
   \caption{{\bf Vertical Aspect of Term Rules} \label{vertterm}}
\end{table}
and the horizontal aspect in Table~\ref{horizterm}.
\begin{table}
   \begin{center}
      \begin{tabular}{|c|c|c|} 
         \cline{1-1} \cline{3-3}
         {\sc entailment version}
         &&
         {\sc orthogonality version} \\
         \cline{1-1} \cline{3-3}
         \multicolumn{3}{c}{} \\
         \cline{1-1} \cline{3-3}
         && \\
         \multicolumn{1}{|c}{} & \multicolumn{1}{c}{\makebox[0in]{\bf Tensors}} & \\
         && \\
         \fbox{\begin{tabular}{c}
                  $\rulezero{(y \tprod \alpha) \logequiv \alpha \logequiv (\alpha \tprod x)}
                           {{\bf identity}}$ \\
                  for terms $\term{y}{\alpha}{x}$
               \end{tabular}}
         &&
         \fbox{\begin{tabular}{c}
                  $\ruledoublezero{(\alpha \tprod x) \orthog{} \tneg{}{\alpha}{} \orthog{} (y \tsum \alpha)}
                                  {(y \tprod \alpha) \orthog{} \tneg{}{\alpha}{} \orthog{} (\alpha \tsum x)}
                                  {{\bf identity}}$ \\
                  for terms $\term{y}{\alpha}{x}$
               \end{tabular}} \\
         && \\
         \fbox{\begin{tabular}{c}
                  $\ruletwo{\beta \vdash \delta}
                           {\alpha \vdash \gamma}
                           {(\beta \tprod \alpha) \vdash (\delta \tprod \gamma)}
                           {{\bf monotonicity}}$ \\
                  for terms $\term{z}{\beta,\delta}{y}$ and $\term{y}{\alpha,\gamma}{x}$
               \end{tabular}}
         &&
         \fbox{\begin{tabular}{c}
                  $\ruletwo{\beta \orthogonal \delta}
                           {\alpha \orthogonal \gamma}
                           {(\beta \tprod \alpha) \orthogonal (\gamma \tsum \delta)}
                           {$\tprod \! \tsum$}$ \\
                  for terms $\term{z}{\beta}{y}$ versus $\opterm{z}{\delta}{y}$ \\
                  and $\term{y}{\alpha}{x}$ versus $\opterm{y}{\gamma}{x}$
               \end{tabular}} \\
         && \\
         \multicolumn{3}{|c|}{\fbox{\begin{tabular}{c}
                                     $\beta \orthogonal \alpha$ iff $\beta \vdash \tneg{}{\alpha}{}$ \mbox{   } ({\bf orthog}-{\bf entail}) \\
                                     for terms $\term{y}{\alpha}{x}$ versus $\opterm{y}{\beta}{x}$
                                  \end{tabular}}} \\
         && \\
         \multicolumn{3}{|c|}{\fbox{\begin{tabular}{c}
                                     $\beta \orthogonal \alpha$ iff $\beta \tprod \alpha \vdash x$ and $\alpha \tprod \beta \vdash y$ \mbox{   } ({\bf orthogonality definition}) \\
                                     for terms $\term{y}{\alpha}{x}$ versus $\opterm{y}{\beta}{x}$
                                  \end{tabular}}} \\
         && \\
         \cline{1-1} \cline{3-3}
      \end{tabular}
   \end{center}
   \caption{{\bf Horizontal Aspect of Term Rules} \label{horizterm}}
\end{table}
The homset-order axioms in the two versions are immediately equivalent;
in fact,
the logical axioms are equivalent to reflexivity of entailment,
the cut rule is equivalent to transitivity of entailment,
and symmetry is equivalent to contravariance of tensor negation.
So entailment is a homset preorder on terms,
and ${\bf Lang}$ is a preordered graph.
Similarly,
the tensor axioms,
the $\tprod \tsum$-rule and monotonicity of tensor product $\tprod$,
are equivalent.
By applying tensor negation,
the monotonicity of tensor product $\tprod$ and the monotonicity of tensor sum $\tsum$
are equivalent facts.
The cut rule implies that orthogonality is monotonic:
if $\beta \orthogonal \alpha$ and $\alpha' \vdash \alpha$ then $\beta \orthogonal \alpha'$.
The boolean rules assert that $\bsum$ is a least upper bound
and that $\bprod$ is a greatest lower bound in the entailment order.
The zero rule provides the axiomatics for both bottom $0$ and top $1$.
Thus,
the (internal) vertical aspect of term formulas has the structure of a lattice;
with the (external) tensor negation,
ignoring types,
it has the structure of a Boolean algebra.
The entailment axioms,
minus contravariance,
are essentially the axioms for a join bisemilattice.
The vertical aspect of the basic calculus corresponds to standard (propositional) logic.
The horizontal aspect of the basic calculus,
minus the orthogonality definition axiom,
is a dialectical logic analog or typed version of the ``multiplicative fragment'' adjoined by linear logic.
The definition of orthogonality,
which axiomatizes ``Boolean orthogonality'' or the definition of orthogonality in Boolean categories,
separates dialectical logic from typed linear logic.
We want to show that the horizontal aspect of term formulas has categorical structure for both tensor product and tensor sum.
We can do this quite simply by extending entailment to sequences of term formulas.

\paragraph{Sequents.}
A {\em sequent\/} $\alpha$ is a path of term formulas ({\bf Lang}-edges)
$\term{y}{\alpha}{x} = \term{y}{\alpha_n}{x_{n-1}} \rightharpoondown \cdots \rightharpoondown \term{x_1}{\alpha_1}{x}$.
Such a path is a typed version of a sequence of term formulas.
The concatenation of two sequents $\term{z}{\beta}{y}$ and $\term{y}{\alpha}{x}$
is denoted by $\term{z}{\beta \circ \alpha}{x}$.
The empty sequent at type symbol $x$ is denoted by $\term{x}{\varepsilon_x}{x}$.
So sequents are arrows in a free (path) category ${\bf Lang}^\ast$
having concatenation $\circ$ as composition and empty paths $\varepsilon_x$ as identities.
The category of sequents ${\bf Lang}^\ast$ inherits from the graph of terms {\bf Lang}
a weak {\em vector entailment\/} homset order $\vec{\vdash}$,
defined by
$\alpha \vec{\vdash} \beta$ when $|\alpha| = |\beta|$ and $\alpha_i \vdash \beta_i$ for all $1 \leq i \leq n$,
where $\alpha = \alpha_n \circ \cdots \circ \alpha_1$.
Clearly,
sequent concatenation is monotonic w.r.t. vector entailment:
if $\beta \vec{\vdash} \delta$ and $\alpha \vec{\vdash} \gamma$
then $(\beta \circ \alpha) \vec{\vdash} (\delta \circ \gamma)$
for any two composable parallel pairs of sequents
$\term{z}{\beta,\delta}{y}$ and $\term{y}{\alpha,\gamma}{x}$.
So $\vec{{\bf Lang}}^\ast \define \pair{{\bf Lang}^\ast}{\vec{\vdash}}$ is a bipreorder (preordered category).
Extend tensor negation to sequents by defining the sequent ``vector'' {\em tensor negation\/}
$\vecneg{\alpha} \define \tneg{}{\alpha_1}{} \circ \cdots \circ \tneg{}{\alpha_n}{}$
for any sequent $\term{y}{\alpha}{x}$ which is the path of terms
$\alpha = \alpha_n \circ \cdots \circ \alpha_1$;
in particular,
$\vecneg{\varepsilon_x} \define \varepsilon_x$.
Vector tensor negation is contravariant:
if $\alpha \vec{\vdash} \beta$ then $\vecneg{\beta} \vec{\vdash} \vecneg{\alpha}$.
So vector tensor negation is a categorical involution $\vecneg{\vecneg{\alpha}} = \alpha$;
that is,
a contravariant functor
$\vec{{\bf Lang}}^\ast \stackrel{\vec{\neg}}{\rightarrow} (\vec{{\bf Lang}}^\ast)^{\rm coop}$,
which is self-inverse $\vecneg{} \cdot (\vecneg{})^{\rm coop} = {\rm Id}$.
The category of sequents, vector entailment, and vector tensor negation form a polarized bipreorder $\vec{{\bf Lang}}^\ast$.

Sequents will be interpreted in Boolean categories.
A sequent can be interpreted in a Boolean category in either a polar sense (using $\tprod$) or an antipolar sense (using $\tsum$).
The two senses are inter-translatable via tensor negation.
In Girard's version of linear logic,
sequents are interpreted in the antipolar sense.
The interpretation of a sequent $\term{y}{\alpha}{x}$ in the polar sense
is done via the {\em tensor product term\/} $\term{y}{\tprod(\alpha)}{x}$,
a sequent of length one,
which is defined by 
$\tprod(\alpha) \define \alpha_n \tprod \cdots \tprod \alpha_1$.
More precisely,
{\bf base}
   $\tprod(\varepsilon_x) \define x$
   for any type $x$,
and
{\bf induction} 
   $\tprod(\beta \circ \alpha) \define \beta \,\tprod\, \tprod(\alpha)$
   for any term $\term{z}{\beta}{y}$ and any sequent $\term{y}{\alpha}{x}$.
In particular,
$\tprod(\alpha) = \alpha \tprod x$ for any term $\term{y}{\alpha}{x}$.
So the tensor product operator is a type-preserving graph morphism
${\bf Lang}^\ast \stackrel{\tprod}{\longrightarrow} {\bf Lang}$
from the category of sequents ${\bf Lang}^\ast$ to the graph of terms {\bf Lang}.
Dually,
the interpretation of a sequent $\term{y}{\alpha}{x}$ in the antipolar sense
is done via the {\em tensor sum term\/} $\term{y}{\tsum(\alpha)}{x}$,
a sequent of length one,
which is defined by 
$\tsum(\alpha) \define \alpha_n \tsum \cdots \tsum \alpha_1$.
More precisely,
{\bf base}
   $\tsum(\varepsilon_x) \define x$
   for any type $x$,
and
{\bf induction} 
   $\tsum(\beta \circ \alpha) \define \tsum(\beta) \,\tsum\, \alpha$
   for any sequent $\term{z}{\beta}{y}$ and any term $\term{y}{\alpha}{x}$.
In particular,
$\tsum(\alpha) = x \tsum \alpha$ for any term $\alpha$.
So the tensor sum operator is also a type-preserving graph morphism
${\bf Lang}^\ast \stackrel{\tsum}{\longrightarrow} {\bf Lang}$.
By induction we can show that
the tensor product and tensor sum operations
are related by the DeMorgan's laws
$\tneg{}{(\tprod \alpha)}{} = \tsum(\vecneg{\alpha})$
and
$\tneg{}{(\tsum \alpha)}{} = \tprod(\vecneg{\alpha})$.

In the polar sense of interpretation,
we require that each sequent $\alpha$ be logically equivalent to its tensor product term $\tprod(\alpha)$.
So define a {\em polar entailment\/} homset order $\vdash_{\tprod}$ by
$\alpha \vdash_{\tprod} \beta$ when $\tprod(\alpha) \vdash \tprod(\beta)$.
Polar entailment partially orders ${\bf Lang}^\ast$-homsets,
if we quotient out by logical equivalence $\logequiv_{\tprod}$ defined by:
$\alpha \logequiv_{\tprod} \beta$
when both $\alpha \vdash_{\tprod} \beta$ and $\beta \vdash_{\tprod} \alpha$ hold.
Then any sequent $\term{y}{\alpha}{x}$ is entailment equivalent to its associated tensor product term
$\alpha \logequiv_{\tprod} \tprod(\alpha)$,
as is required by the polar interpretation,
since $\tprod(\tprod(\alpha)) = \tprod(\alpha) \tprod x \logequiv \tprod(\alpha)$.
The tensor product of terms is associative,
up to polar entailment equivalence (for sequents),
since $\gamma \tprod (\beta \tprod \alpha) \logequiv_{\tprod} \gamma \circ (\beta \circ \alpha)
= (\gamma \circ \beta) \circ \alpha \logequiv_{\tprod} (\gamma \tprod \beta) \tprod \alpha$.
Polar entailment equivalence $\logequiv_{\tprod}$ extends term entailment equivalence $\logequiv$;
that is,
polar entailment equivalence coincides with entailment equivalence on terms,
$\beta \logequiv_{\tprod} \alpha$ iff $\beta \logequiv \alpha$
for all terms $\term{y}{\alpha,\beta}{x}$.
So,
the tensor product of terms is associative,
up to term entailment equivalence:
$\gamma \tprod (\beta \tprod \alpha) \logequiv (\gamma \tprod \beta) \tprod \alpha$.
By induction tensor product preserves composition,
up to term equivalence
$\tprod(\beta \circ \alpha) \logequiv \tprod(\beta) \tprod \tprod(\alpha)$ .
Sequent concatenation is monotonic w.r.t. polar entailment:
if $\beta \vdash_{\tprod} \delta$ and $\alpha \vdash_{\tprod} \gamma$
then $(\beta \circ \alpha) \vdash_{\tprod} (\delta \circ \gamma)$
for any two composable parallel pairs of sequents
$\term{z}{\beta,\delta}{y}$ and $\term{y}{\alpha,\gamma}{x}$,
since tensor product is monotonic.
So,
the category of sequents ${\bf Lang}^\ast$ forms a bipreorder
${\bf Lang}^\ast_{\tprod} \define \pair{{\bf Lang}^\ast}{\vdash_{\tprod}}$
with polar entailment $\vdash_{\tprod}$.
By induction using the monotonicity rule,
the tensor product operator is monotonic w.r.t. vector entailment:
if $\alpha \vec{\vdash} \beta$ then $\tprod(\alpha) \vdash \tprod(\beta)$.
So vector entailment is weaker than polar entailment:
if $\alpha \vec{\vdash} \beta$ then $\alpha \vdash_{\tprod} \beta$.

Dually,
in the antipolar sense of interpretation,
we require that each sequent $\alpha$ be logically equivalent to its tensor sum term $\tsum(\alpha)$.
So define an {\em antipolar entailment\/} homset order $\vdash_{\tsum}$ by
$\alpha \vdash_{\tsum} \beta$ when $\tsum(\alpha) \vdash \tsum(\beta)$.
The category of sequents ${\bf Lang}^\ast$ forms a bipreorder
${\bf Lang}^\ast_{\tsum} \define \pair{{\bf Lang}^\ast}{\vdash_{\tsum}}$
with antipolar entailment $\vdash_{\tsum}$.
Again,
vector entailment is weaker than antipolar entailment:
if $\alpha \vec{\vdash} \beta$ then $\alpha \vdash_{\tsum} \beta$.
The polar and antipolar orders are two alternate interpretations for the entailment relation $\vdash$ on sequents.
They are polar duals,
and are interdefinable via the equivalence:
$\alpha \vdash_{\tprod} \beta$ iff $\vecneg{\beta} \vdash_{\tsum} \vecneg{\alpha}$.
More concisely,
vector tensor negation is an involution
${\bf Lang}^\ast_{\tprod} \stackrel{\vec{\neg}}{\rightarrow} ({\bf Lang}^\ast_{\tsum})^{\rm coop}$.
So the category of sequents, the two polarities of entailment, and vector tensor negation form a polarized bipreorder ${\bf Lang}^\ast$.

\paragraph{The Term Category.}
Entailment partially orders ${\bf Lang}$-homsets,
if we quotient out by logical equivalence $\logequiv$.
Entailment equivalence quotienting is done automatically when we use the closed subset construction.
For any term $\term{y}{\alpha}{x}$,
let $\term{[y]}{[\alpha]}{[x]}$ denote the quotient term (entailment equivalence class) of $\alpha$.
Let ${\bf Term}$ denote the quotient graph of {\bf Lang};
that is,
{\bf Term} is the graph of types and quotient terms.
Define the boolean and tensor operations on quotient terms
via representatives.
For example,
define the tensor product and tensor sum of quotient terms by
$[\beta] \tprod [\alpha] \define [\beta \tprod \alpha]$ and
$[\beta] \tsum [\alpha] \define [\beta \tsum \alpha]$.
Define the quotient entailment order by
$[\alpha] \vdash [\beta]$ when $\alpha \vdash \beta$,
and define the quotient orthogonality relation by
$[\beta] \orthog{} [\alpha]$ when $\beta \orthog{} \alpha$ is provable.
Finally,
define the quotient tensor negation by
$\tneg{}{[\alpha]}{} \define [\tneg{}{\alpha}{}]$.
These operations and relations are well-defined,
and the tensors are associative.
Since term tensor product and sum are monotonic w.r.t. entailment order,
the tensor product and sum of quotient terms are also monotonic w.r.t. entailment order.
So there is a join bisemilattice
${\bf Term}_{\tprod}^{\bsum} = \triple{\quadruple{{\bf Term}}{\vdash}{\tprod}{{\rm Id}}}{\bsum}{0}$
called the {\em quotient term pole\/},
whose objects are (quotients of) types,
whose arrows are quotient terms,
whose composition is the tensor product of quotients,
and whose homset order is quotient entailment.
Similarly,
there is a meet bisemilattice
${\bf Term}_{\tsum}^{\bprod} = \triple{\quadruple{{\bf Term}}{\vdash}{\tsum}{{\rm Id}}}{\bprod}{1}$
called the {\em quotient term antipole\/}.
Tensor negation is an involution of join bisemilattices
${\bf Term}_{\tprod}^{\bsum} \stackrel{\tneg{}{}{}}{\rightarrow} ({\bf Term}_{\tsum}^{\bprod})^{\rm coop}$,
and also an involution of meet bisemilattices
$({\bf Term}_{\tprod}^{\bsum})^{\rm coop} \stackrel{\tneg{}{}{}}{\leftarrow} {\bf Term}_{\tsum}^{\bprod}$.
So the two quotient term poles and quotient tensor negation
form a polarized bisemilattice,
also denoted by {\bf Term},
for which the orthogonality-entailment axiom and the orthogonality definition axiom hold.
\begin{Theorem}
   The category {\bf Term} of quotient terms is a Boolean category.
\end{Theorem}
The DeMorgan's law $\tneg{}{(\tprod \alpha)}{} = \tsum(\vecneg{\alpha})$
states that the pair of {\em tensor term\/} operations is a morphism of polarized bipreorders
${\bf Lang}^\ast \stackrel{\langle \tprod , \tsum \rangle}{\longrightarrow} {\bf Term}$.
It is a quotient functor
(a full functor which is a bijection on objects),
which constructs ${\bf Term}$ as the entailment-quotient category of ${\bf Lang}^\ast$.

\paragraph{Soundness and Completeness.}
A {\em classical structure\/} $\pair{\Im}{{\bf B}}$ for the basic calculus,
the internal language of classical dialectical logic,
consists of a Boolean category ${\bf B}$
and an interpretion map (graph morphism) ${\bf Lang} \stackrel{\Im}{\longrightarrow} {\bf B}$
which preserves negation, identities, entailment order, zeroes, ones, boolean products and sums, and tensor products and sums.
The interpretation map $\Im$ assigns to each type symbol $x$ a {\bf B}-type $\Im(x)$
and assigns to each atom $\term{y}{a}{x}$ a {\bf B}-term $\term{\Im(y)}{\Im(a)}{\Im(x)}$.
Following the polar sense of interpretation,
we extend the interpretation $\Im$ to sequents by defining
$\Im_{\tprod}(\alpha) \define \Im(\tprod \alpha)$ for any sequent $\term{y}{\alpha}{x}$.
So $\Im$ is a morphism of polarized bipreorders
${\bf Lang}^\ast \stackrel{\Im}{\longrightarrow} {\bf B}$,
with the polar interpretation embodied in the polar part
${\bf Lang}^\ast \stackrel{\Im_{\tprod}}{\longrightarrow} {\bf B}_{\tprod}^{\bsum}$
of $\Im$ (a morphism of bipreorders),
and the antipolar interpretation embodied in the antipolar part
${\bf Lang}^\ast \stackrel{\Im_{\tsum}}{\longrightarrow} {\bf B}_{\tsum}^{\bprod}$
of $\Im$ (which is defined by
$\Im_{\tsum} \define \tneg{}{}{} \cdot (\Im_{\tprod})^{\rm coop} \cdot (\tneg{B}{}{})^{\rm coop}$).
$\Im_{\tprod}$ preserves order,
since if $\beta \vdash \alpha$ for any two parallel sequents $\term{y}{\beta,\alpha}{x}$
then $\Im_{\tprod}(\beta) = \Im(\tprod \beta) \preceq \Im(\tprod \alpha) = \Im_{\tprod}(\alpha)$.
Since
$\beta \logequiv \alpha$ implies $\Im_{\tprod}(\beta) = \Im_{\tprod}(\alpha)$
for any two parallel sequents $\term{y}{\beta,\alpha}{x}$, 
there is a functor
${\bf Term}_{\tprod}^{\bsum} \stackrel{\Im_{\tprod}^{\bsum}}{\longrightarrow} {\bf B}_{\tprod}^{\bsum}$
uniquely satisfying the functorial equation $\Im_{\tprod} = \tprod(\,) \cdot \Im_{\tprod}^{\bsum}$.
The extended interpretation $\Im_{\tprod}^{\bsum}$ is the polar part of a morphism of Boolean categories
${\bf Term} \stackrel{\Im}{\longrightarrow} {\bf B}$.
The antipolar part,
using the antipolar interpretation and tensor sum terms,
is defined by
$\Im_{\tsum}^{\bprod} \define \tneg{}{}{} \cdot (\Im_{\tprod}^{\bsum})^{\rm coop} \cdot (\tneg{B}{}{})^{\rm coop}$.
The entailment quotient and the term category
define the fundamental classical structure $\pair{[\,]}{{\bf Term}}$,
whose extended interpretation is the identity functor ${[\,]}_{\tprod}^{\bsum} = {\rm Id}_{\rm Term}$.
\begin{Theorem}
   The Boolean category {\bf Term} is free (w.r.t the connectives) over the language (type-atom graph) {\bf Lang}.
\end{Theorem}
An orthogonality assertion $\beta \orthog{} \alpha$,
for two opposed sequents $\term{y}{\alpha}{x}$ versus $\opterm{y}{\beta}{x}$,
is ({\em tensorially\/}) {\em valid\/} in a structure $\Im$
when the orthogonality $\Im(\beta) \orthog{} \Im(\alpha)$ holds in the Boolean category ${\bf B}$. 
As a special case,
a endosequent $\term{x}{\alpha}{x}$ is valid in $\Im$ when $\Im(\alpha) \preceq \Im(x)$.
A {\em tautology\/} is an orthogonality assertion $\beta \orthog{} \alpha$ which is valid in any classical structure.
\begin{Theorem}[Soundness]
   The basic calculus for dialectical logic is sound w.r.t. validity in classical structures.
\end{Theorem}
\begin{Theorem}[Completeness]
   The basic calculus for dialectical logic is complete w.r.t. validity in classical structures.
\end{Theorem}
\begin{proof}
   Suppose $\beta \orthogonal \alpha$ is a tautology at $x$.
   Then,
   since $\beta \orthogonal \alpha$ is valid in every classical structure,
   it is valid in the free classical structure $\pair{[\,]}{{\bf Term}}$,
   and so the orthogonality  $[\beta] \orthogonal [\alpha]$ holds in {\bf Term}.
   But by definition,
   $[\beta] \orthogonal [\alpha]$ iff $\beta \orthogonal \alpha$ is provable.
\end{proof}

\paragraph{Summary.}
In this paper we have discussed the internal process aspect of dialectical logic,
which is the logic of the flow dialectic.
In the promised extension \cite{Kent88} of this paper
we will also discuss the external object aspect of dialectical logic,
which is the logic of the flow constraint dialectic.
This external aspect involves the semantic notions of monoids (preorder objects), processes, topologies and topomonoidal structures,
and the axiomatic notions of exponentials (Girard's affirmation and consideration modalities) and quantifiers.

\appendix
\section{Subtypes}

\paragraph{Comonoids.}
For any type $x$ in a bisemilattice {\bf P} a {\em comonoid\/} $u$ at $x$,
denoted by $u \type x$,
is an endoterm $\term{x}{u}{x}$
which satisfies the ``part'' axiom (coreflexivity) $u \preceq_{x,x} x$,
stating that $u$ is a part of the type (identity term) $x$,
and the ``idempotency'' axiom (cotransitivity) $u \preceq_{x,x} u \circ u$.
A comonoid is also called an {\em interior term\/}.
Since $u \circ u \preceq x \circ u = u$,
we can replace the inequality in the idempotency axiom with the equality $u \circ u = u$.
For a functional term (adjoint pair) $\term{y}{f \dashv f^{\rm op}}{x}$ 
the composite interior endoterm $\term{x}{f^{\rm op} \circ f}{x}$ is called the {\em comonoid\/} of the functional term $f$.
This comonoid is the top comonoid $f^{\rm op} \circ f = x$
iff $f$ is an epimorphism iff $f \dashv f^{\rm op}$ is a reflective pair.
The comonoids $\term{y}{p \circ i}{y}$ of subtypes $\term{y}{i \dashv p}{x}$ 
are special $x$-comonoids which split (through $y$).
In this sense comonoids are generalized subtypes.
Comonoids of type $x$ are ordered by entailment $\preceq_x \define \preceq_{x,x}$.
The bottom endoterm $\bot_x$ is the smallest comonoid of type $x$.
The join $v \vee u$ of any two comonoids $v,u$ of type $x$ is also a comonoid of type $x$.
Denote the join semilattice of comonoids of type $x$ by $\comonoid{}{x}{}$.
We can interpret the semilattice $\comonoid{}{x}{}$ as a ``state-set'' indexed by the type $x$,
with a comonoid $u \memberof \comonoid{}{x}{}$ being a ``state'' of a system.
The state $u \memberof \comonoid{}{x}{}$ has internal structure
and is a composite object sharing an ordering of nondeterminism $\preceq_x$ with other states.

For any two comonoids $u, v \in \comonoid{}{x}{}$
the tensor product is a lower bound
$u \circ v \preceq u$ and $u \circ v \preceq v$
which is an upper bound for comonoids below $u$ and $v$:
if $w \preceq u$ and $w \preceq v$ then $w \preceq u \circ v$.
If $u$ and $v$ commute $u \circ v = v \circ u$
then the tensor product $u \circ v$ is a comonoid;
in which case it is the meet $u \circ v = u \wedge v$ in $\comonoid{}{x}{}$.
{\bf [Standardization property:]}
the bisemilattice {\bf P} is said to be {\em locally standard\/} when
$\comonoid{}{x}{}$ is closed under tensor product for each type $x$;
that is,
when the tensor product $u \circ v$ is a comonoid
for any two comonoids $u, v \in \comonoid{}{x}{}$.
Then $\comonoid{}{x}{}$ is a lattice,
with the tensor product $v \circ u$ of two comonoids $v,u \memberof \comonoid{}{x}{}$ being the lattice meet in $\comonoid{}{x}{}$,
and the tensor product identity (or type) endoterm $x$ being the largest comonoid of type $x$.
Furthermore,
the meet distributes over the join.
We assume that any join bisemilattice {\bf P} is locally standard.
This standardization property means that
the local contexts (monoidal semilattices) of comonoids $\{ \comonoid{}{x}{} \mid x \mbox{ a type} \}$ are standard contexts (distributive lattices).

In a complete Heyting category {\bf H} an endoterm $\term{x}{p}{x}$
contains a largest comonoid of the same type $x$,
called the {\em interior\/} of $p$ and denoted by $\interior{p}$.
The interior is defined as the join
$\interior{p} \define \bigvee \{ w \memberof \comonoid{}{x}{} \mid w \preceq_x p \}$,
and satisfies the condition
$w \preceq_x p$ iff $w \preceq_x \interior{p}$
for all comonoids $w \memberof \comonoid{}{x}{}$.
In an arbitrary join bisemilattice {\bf P},
we use this condition to define (and to assert the existence of) the interior of endoterms.
The interior $\interior{p}$, 
when it exists,
is the largest generalized {\bf P}-subtype inside $p$.
The interior of endoterms models the ``affirmation modality'' of linear logic \cite{Girard}.
Any comonoid $w \memberof \comonoid{}{x}{}$ is its own interior $\interior{w} = w$.
Without the local standardization assumption,
meets would still exist in $\comonoid{}{x}{}$:
the interior of the tensor product is the meet
$\interior{(u \circ v)} = u \wedge v = \interior{(v \circ u)}$.

We are especially interested in join bisemilattices {\bf P}
for which any {\bf P}-endoterm has such an interior.
Such bisemilattices can be called interior (or affirmation) bisemilattices.
A join bisemilattice {\bf P} is an {\em interior bisemilattice\/}
when at each type $x$ the inclusion-of-comonoids monotonic function
$\comonoid{}{x}{} \stackrel{{\rm Inc}_x}{\longrightarrow} {\bf P}[x,x]$
has a right adjoint
${\bf P}[x,x] \stackrel{\interior{(\,)}}{\rightarrow} \comonoid{}{x}{}$
called {\em interior\/},
which with inclusion forms a coreflective pair of monotonic functions ${\rm Inc}_x \dashv \interior{(\,)}$.
Composition $\interior{(\,)} \cdot {\rm Inc}_x$ is an general interior operator on endoterms.
Any meets that exist in ${\bf P}[x,x]$
are preserved by interior
$\interior{(p \wedge q)} = \interior{p} \circ \interior{q}$
for endoterms $p,q \memberof {\bf P}[x,x]$,
since interior is a right adjoint.
In an interior Heyting category {\bf H},
the distributive lattice of comonoids $\comonoid{}{x}{}$ at each type $x$ 
is actually a complete cartesian Heyting monoid,
which is another name for a complete Heyting algebra.
Since interiors exist,
for any two comonoids $u,v \memberof \comonoid{}{x}{}$
we can make the definition
$u \imply v \define \interior{(u \tensorimplysource v)}$. 
Then
$u \imply v
 = \interior{(u \tensorimplysource v)}
 = \interior{(v \tensorimplytarget u)}$
is a locally standard implication,
since
$w \preceq u \imply v$ iff
$w \preceq \interior{(u \tensorimplysource v)}$ iff
$w \preceq (u \tensorimplysource v)$ iff 
$u \circ w \preceq v$ iff
$w \circ u \preceq v$ iff
$w \preceq (v \tensorimplytarget u)$ iff 
$w \preceq \interior{(v \tensorimplytarget u)}$.
Comonoids in bisemilattices,
and even more strongly in interior Heyting categories,
play the role of ``localized truth values''.
Any complete Heyting category is an interior Heyting category.

In a bisemilattice {\bf P},
for each {\bf P}-adjunction (functional term) $\term{y}{f \dashv f^{\rm op}}{x}$
and each {\bf P}-comonoid $v \memberof \comonoid{}{y}{}$ at $y$,
the endoterm $\term{x}{f^{\rm op} \circ v \circ f}{x}$ is a {\bf P}-comonoid
$(f^{\rm op} \circ v \circ f) \memberof \comonoid{}{x}{}$ at $x$.
So the direct image monotonic function ${\bf P}^f$ restricts to {\bf P}-comonoids.
Denote this restriction by
$\comonoid{}{y}{} \stackrel{\Omega^{f}}{\longrightarrow} \comonoid{}{x}{}$
and call it the {\em direct image\/} also.
When {\bf P} is an interior bisemilattice,
the direct image function has a right adjoint 
$\comonoid{}{y}{} \stackrel{\Omega_{f}}{\longleftarrow} \comonoid{}{x}{}$ 
called the {\em inverse image\/} monotonic function,
and defined by
$\comonoid{f}{u}{} \define \interior{(f \circ u \circ f^{\rm op})}$
for each {\bf P}-comonoid $u \memberof \comonoid{}{x}{}$.
If we denote this adjointness by
$\comonoid{}{f}{} \define (\Omega^{f} \dashv \Omega_{f})$,
then the comonoid construction $\Omega$ is an indexed adjointness
(dialectical base)
${\bf P}^\dashv \stackrel{\Omega}{\longrightarrow} {\bf adj}$,
mapping functional {\bf P}-terms into the subcategory of {\bf adj} consisting of distributive lattices
and adjoint pairs of monotonic functions.

In subset categories $\power{{\bf C}}$ a comonoid of type $x$ 
is either the empty endoterm $\term{x}{\emptyset}{x}$ or the identity singleton $\term{x}{\{x\}}{x}$,
and these can be interpreted as the truth-values {\bf false} and {\bf true},
so that $\comonoid{}{x}{}$ is the complete Heyint algebra $\comonoid{}{x}{} \cong {\bf 2}$.
In closure subset categories $\power{{\bf P}}$
a comonoid $\term{x}{W}{x}$ of type $x$ is a closed-below subset $W \subseteq {\bf P}[x,x]$
of {\bf P}-endoterms $\term{x}{w}{x}$,
which are subparts of the identity $w \preceq x$
and which factor (possibly trivially) $w \preceq v \circ u$ into two other endoterms $v,u \memberof W$.
Since $\power{{\bf P}}$ is a cHc,
the lattice of comonoids $\comonoid{\power{{\rm P}}}{x}{}$ is also a complete Heyting algebra.
Any {\bf P}-comonoid $\term{x}{w}{x}$ is embeddable as
the $\power{{\bf P}}$-comonoid $\term{x}{\downarrow{w}}{x}$.
So we can regard $\power{{\bf P}}$-comonoids as generalized {\bf P}-comonoids
called {\em closure subset {\bf P}-comonoids\/}.

For any source and target comonoids
$v \memberof \comonoid{}{y}{}$ and $u \memberof \comonoid{}{x}{}$
the term $\term{v}{r_{vu}}{u}$ defined by $r_{vu} \define v \circ r \circ u$
is called the $(v,u)$-{\em th subterm\/} of $r$.
A {\bf P}-{\em coprocess\/} $\term{v}{r}{u}$
is a {\bf P}-term $\term{y}{r}{x}$
which satisfies the external source constraint $v \circ r \succeq_{y,x} r$
saying that $r$ restricts to the source comonoid $v \type y$,
and which satisfies the external target constraint $r \circ u \succeq_{y,x} r$
saying that $r$ corestricts to the target comonoid $u \type x$.
The source/target restriction conditions can be replaced by the two equalities
$v \circ r = r$ and $r \circ u = r$;
or by the single equality
$r_{vu} = v \circ r \circ u = r$.
Thus,
the notion of coprocess allows comonoids to function as identity arrows,
or objects,
of some category.
To make this precise we define the biposet $\comonoid{}{{\bf P}}{}$,
whose objects are {\bf P}-comonoids and whose arrows are {\bf P}-coprocesses.
Although $\comonoid{}{x}{} \subseteq {\bf P}[x,x]$,
note that $\comonoid{}{x}{} \neq {\bf P}[x,x]$,
since endoarrows exist which are not comonoids.
Given any {\bf P}-term $y \stackrel{r}{\rightarrow} x$,
let $\filtersource{r} \subseteq \comonoid{}{y}{}$
denote the collection
$\filtersource{r} \define \{ v \mid v \circ r \succeq_{y,x} r \}$
of all comonoids at the source type $y$ satisfying source restriction.
Since $\filtersource{r}$ is closed above and closed under finite meets (= tensor products)
it is a filter in the lattice $\comonoid{}{y}{}$
called the {\em source filter\/} of $r$.
Similarly,
the {\em target filter\/} $\filtertarget{r}$ of $r$ is the collection
$\filtertarget{r} \define \{ u \mid r \preceq_{y,x} r \circ u \}
                  \subseteq \comonoid{}{x}{}$
of all comonoids at $x$ satisfying target corestriction. 
Given two comonoids $v \type y$ and $u \type x$,
a term $\term{y}{r}{x}$ is a coprocess $\term{v}{r}{u}$
iff $v \memberof \filtersource{r}$ and $u \memberof \filtertarget{r}$.

Unfortunately,
the category $\comonoid{}{{\bf P}}{}$ is not as useful as one might desire;
in particular,
there is no canonical functor to the underlying category {\bf P} of types and terms
since identities are not preserved.
But by suitably weakening the constraint $v \circ r = r = r \circ u$
we get a very useful and interesting category.
A {\em Hoare triple\/} or {\em Hoare assertion\/} $v \type y \stackrel{r}{\rightarrow} u \type x$,
denoted traditionally although imprecisely by $\{v\}r\{u\}$,
consists of a ``flow specifying'' {\bf P}-term $\term{y}{r}{x}$ and two {\bf P}-comonoids,
a ``precondition'' or source comonoid $v \memberof \comonoid{}{y}{}$
and a ``postcondition'' or target comonoid $u \memberof \comonoid{}{x}{}$,
which satisfy the ``precondition/postcondition constraint''
$v \circ r \preceq r \circ u$.
Clearly,
composition of Hoare triples $\{w\}s\{v\} \circ \{v\}r\{u\} = \{w\}(s \circ r)\{u\}$ is well-defined
and $\{u\}x\{u\}$ is the identity Hoare triple at the comonoid $u \type x$.
Also,
there is a zero triple $\{v\}0_{y,x}\{u\}$ for any precondition $v \memberof \comonoid{}{y}{}$ and postcondition $u \memberof \comonoid{}{x}{}$,
and if $\{v\}r\{u\}$ and $\{v\}s\{u\}$ are two triples with the same precondition and postcondition
then $\{v\}(r \bsum s)\{u\}$ is also a triple.
So typed comonoids as objects and Hoare triples as arrows form a join bisemilattice $\Hoare{{\bf P}}$
called the {\em Hoare assertional category\/} over {\bf P}.
There is an obvious underlying type/term functor
$\Hoare{{\bf P}} \stackrel{T_P}{\longrightarrow} {\bf P}$
which is a morphism of join bisemilattices.
For each type $x$ in {\bf P},
the {\em fiber\/} over $x$ is the subcategory
$T_P^{\rm -1}(x) \subseteq \Hoare{{\bf P}}$
of all comonoids and triples which map to $x$.
The objects in $T_P^{\rm -1}(x)$ are the comonoids of type $x$
and the triples in $T_P^{\rm -1}(x)$ are of the form $\{u'\}x\{u\}$,
pairs of comonoids of type $x$ satisfying $u' \preceq u$.
Hence,
the fiber over $x$ is just the join semilattice (actually, lattice) of comonoids
$T_P^{\rm -1}(x) = \comonoid{}{x}{}$.
The axiomatics, semantics and dialectics of Hoare assertional categories and associated constructions,
and their relationship to dynamic logic,
is explored in detail in \cite{Kent89}.

\paragraph{Topotypes and Topomatrices.}
The closure subset construction $\power{{\bf P}}$ does not capture the notion of ``relational structures'' completely.
Although it introduces nondeterminism on the arrows,
it leaves the objects alone.
The notions of ``topology'' and``subtype'' can be naturally combined and locally defined in any cHc {\bf H}.
Topologies of subtypes introduce distributivity on objects.
A {\em topology of {\bf H}-comonoids\/} or {\em {\bf H}-topotype\/} $W = \pair{W}{x}$, 
denoted by $W \type x$,
is a topology $W$ in the complete lattice $\comonoid{}{x}{}$ of comonoids at $x$ regarded as a one-object subcategory of {\bf H}
(the more general notion of a {\em topology\/} in a cHc {\bf H} is discussed in \cite{Kent88});
that is,
$W$ is a collection $W \subseteq \comonoid{}{x}{}$ of comonoids of $x$,
which is closed under finite tensor products and arbitrary homset joins.
A topotype is a kind of ``power type'',
which is {\em not} imposed from without,
but arises naturally out of the mathematical structure.
Since tensor products are finite homset meets for comonoids,
a topotype $W \type x$ is just a standard topology in the complete lattice $\comonoid{}{x}{}$.
An advantage of standard topologies over general tensor product topologies
is that homset order is more directly related to topological meet.
$W$ is interpreted to be an {\em object of inner truth-values\/} at type $x$,
and its topological nature can be used to define approximation or limit structures on terms whose source or target is $x$.
Any comonoid $u \type x$ can be identified with the topotype $u = \{ \bot_x, u , x \}$.

A topomatrix is a matrix indexed by topologies.
Given two topotypes $V \type y$ and $U \type x$,
an {\bf H}-{\em topomatrix\/} $\term{V \type y}{R}{U \type x}$,
denoted by $R = (r_{vu} \mid v \memberof V, u \memberof U)$,
is a $\comonoid{}{{\bf H}}{}$-matrix
$\product{V}{U} \stackrel{R}{\rightarrow} {\rm Ar}(\comonoid{}{{\bf H}}{})$
monotonically indexed by the source and target topologies.
Monotonic indexing means that
if $v \preceq v'$ and $u \preceq u'$ then $r_{vu} \preceq r_{v'u'}$.
This monotonic indexing property is similar to
the compatibility of ordinary partial functions on the overlap of their domains of definition.
Every cHc {\bf H} has an associated {\em category of topomatrices\/} $\matropo{{\bf H}}$,
whose objects are topotypes $U \type x$,
whose arrows $\term{V \type y}{R}{U \type x}$ are topomatrices,
whose homset order is pointwise order
$(s_{vu}) \preceq (r_{vu})$ when $s_{vu} \preceq r_{vu}$ for all $v \memberof V$ and $u \memberof U$,
whose tensor product is the matrix product
$(S \circ R)_{wu} \define \bigvee_{v \in V} [s_{wv} \circ r_{vu}]$,
and whose identity at $U \type x$ is the topomatrix
$(u' \circ u = u' \wedge u \mid u',u \memberof U)$.
The join operator is a {\em join\/} functor
$\matropo{{\bf H}} \stackrel{{\displaystyle \vee}}{\rightarrow} {\bf H}$,
which maps each topotype to its underlying type $\bigvee(U \type x) = x$
and maps each $\product{V}{U}$ topomatrix $R = (r_{vu})$ to its {\em join term\/}
$\bigvee R = \bigvee_{v \in V,u \in U} r_{vu}$,
the join of all the coprocess entries in $R$.
The $(V,U)$-th component of the join functor $\bigvee$ is a {\em join\/} join-continuous monotonic function
$\matropo{{\bf H}}[V \type y,U \type x] \stackrel{{\displaystyle \vee}_{V,U}}{\longrightarrow} {\bf H}[y,x]$.
The category of comonoids $\comonoid{}{{\bf H}}{}$ can be embedded
$\comonoid{}{{\bf H}}{} \stackrel{{\rm Inc}}{\longrightarrow} \matropo{{\bf H}}$
into the category of topomatrices $\matropo{{\bf H}}$ by
${\rm Inc}(u \type x) = \{ \bot,u,x \} \type x$
and
${\rm Inc}(\term{v \type y}{r}{u \type x})
 = \{ (\bot,\bot,\bot),(\bot,\bot,u),(\bot,\bot,x),(v,\bot,\bot),(y,\bot,\bot) \}
   \cup \singleton{(v,r,u)}
   \cup \{ (v,r,x),(y,r,u),(y,r,x) \}$.
The composition of comonoid embedding with join is the underlying type functor
${\rm Inc} \cdot \bigvee = U_H$.
The restriction of the comonoid-as-topology embedding to identity comonoids
defines the {\em indiscrete-topology functor\/}
${\bf H} \stackrel{\singleton{}}{\longrightarrow} \matropo{{\bf H}}$,
where $\singleton{x} = \{ \bot,x \} \type x$
and $\singleton{r} = \{ (\bot,\bot,\bot),(\bot,\bot,x),(y,\bot,\bot) \}
                     \cup \singleton{(y,r,x)}$.
This functor is clearly fully-faithful,
since for two fixed types $y$ and $x$,
there is a bijection ${\bf H}[y,x] \cong \matropo{{\bf H}}[\singleton{y},\singleton{x}]$.
Also,
$\singleton{} \cdot \bigvee = {\rm Id}_H$.
This implies that the join functor is surjective on objects.

\paragraph{A Representation Theorem.}
Let $V \type y$ and $U \type x$ be any two {\bf H}-topotypes,
and let $\term{y}{r}{x}$ be any {\bf H}-term.
The topomatrix $\term{V \type y}{\decomposition{r}{V}{U}}{U \type x}$ defined by
$\decomposition{r}{V}{U} \define \left( \term{v}{r_{vu}}{u} \mid v \memberof V, u \memberof U \right)$,
where $r_{vu} \define v \circ r \circ u$ is the $(v,u)$-th subterm of $r$,
is called the {\em decomposition matrix\/} of $r$.
Such decompositions,
especially w.r.t. topological bases of comonoids,
give an internal representation of cHc's as distributor-like categories. 
This defines a {\em decomposition\/} join-continuous monotonic function
${\bf H}[y,x] \stackrel{\#_{V,U}}{\longrightarrow} \matropo{{\bf H}}[V \type y,U \type x]$,
where $\#_{V,U}(r) \define \decomposition{r}{V}{U}$.
Moreover,
any {\bf H}-term $\term{y}{r}{x}$ is recoverable from its decomposition matrix $\decomposition{r}{V}{U}$ by applying the join functor
$\bigvee_{V,U}(\#_{V,U}(r)) = \bigvee_{V,U}(\decomposition{r}{V}{U}) = \bigvee_{v \in V, u \in U} r_{v,u} = r$.
This means that the join functor is full (surjective on arrows).
Conversely,
an {\bf H}-topomatrix $\term{V \type y}{R}{U \type x}$ is recoverable from its join term $\bigvee R$ by applying the partition function
$\#_{V,U}(\bigvee_{V,U}(R)) = R$.
This means that the join functor is faithful (injective on arrows).
So for two fixed topotypes $V \type y$ and $U \type x$, 
the decomposition and join monotonic functions are inverse to each other,
and define an isomorphism
${\bf H}[y,x] \cong \matropo{{\bf H}}[V \type y,U \type x]$. 
\begin{Lemma}
   The join functor $\matropo{{\bf H}} \stackrel{{\displaystyle \vee}}{\rightarrow} {\bf H}$ is fully-faithful,
   and a surjection on objects.
\end{Lemma}

A topomatrix $\term{V \type y}{R}{\singleton{x}}$ is called a {\em column {\bf H}-topovector\/}.
If $\term{y}{r}{x}$ is any term and $V \type y$ is a topology at $y$,
then the $V$-{\em source decomposition\/} of $r$ is the column topovector
$\term{V \type y}{\source{r}{V}}{\singleton{x}}$
defined by
$\source{r}{V} \define \left( \term{v}{r_{vx}}{x} \mid r_{vx} = v \circ r, v \memberof V \right)$. 
The $V$-{\em source cotupling\/} of a column topovector $\term{V \type y}{R}{\singleton{x}}$,
where $R$ is the $V$-indexed collection of coprocesses $\left( \term{v}{r_{vx}}{x} \mid v \memberof V \right)$,
is the {\bf H}-term $\term{y}{\cotuple{R}{V}}{x}$ defined by
$\cotuple{R}{V} \define \bigvee_{v \in V} r_{vx}$. 
The source decomposition and cotupling operations are inverse to each other,
with $\cotuple{\source{r}{V}}{V} = r$ and $\source{\cotuple{R}{V}}{V} = R$.
Dually,
a topomatrix $\term{\singleton{y}}{R}{U \type x}$ is called a {\em row {\bf H}-topovector\/}.
If $\term{y}{r}{x}$ is any term and $U \type x$ is a topology at $x$,
then the $U$-{\em target decomposition\/} of $r$ is the row topovector
$\term{\singleton{y}}{\target{r}{U}}{U \type x}$
defined by
$\target{r}{U} \define \left( \term{y}{r_{yu}}{u} \mid r_{yu} = r \circ u, u \memberof U \right)$. 
The $U$-{\em target tupling\/} of a row topovector $\term{\singleton{y}}{R}{U \type x}$,
where $R$ is the $U$-indexed collection of coprocesses $\left( \term{y}{r_{yu}}{u} \mid u \memberof U \right)$,
is the {\bf H}-term $\term{y}{\tuple{R}{U}}{x}$ defined by
$\tuple{R}{U} \define \bigvee_{u \in U} r_{yu}$.
The target decomposition and tupling operations are inverse to each other,
with $\tuple{\target{r}{U}}{U} = r$ and $\target{\tuple{R}{U}}{U} = R$.

Any topology $U \type x$ at $x$
decomposes the identity term $\term{x}{x}{x}$ in either of two ways:
as the source decomposition column topovector  
$\term{U \type x}{\iota_U}{\singleton{x}}$
defined by $\iota_U \define\mbox{ } \source{x}{U} = \left( \term{u}{u}{x} \mid u \memberof U \right)$,
or as the target decomposition row topovector  
$\term{\singleton{x}}{\pi_U}{U \type x}$
defined by $\pi_U \define\mbox{ } \target{x}{U} = \left( \term{x}{u}{u} \mid u \memberof U \right)$.
Moreover,
the identity matrix at $U \type x$ decomposes as $\iota_U \circ \pi_U$,
and the identity matrix at $\singleton{x}$ decomposes as $\pi_U \circ \iota_U$,
so that $\term{U \type x}{\iota_U}{\singleton{x}}$ and $\term{\singleton{x}}{\pi_U}{U \type x}$ are inverse topomatrices.
Since $\iota_U$ and $\pi_U$ are inverse pairs,
they are adjoint pairs in both directions
$\term{U \type x}{\iota_U \dashv \pi_U}{\singleton{x}}$
and
$\term{\singleton{x}}{\pi_U \dashv \iota_U}{U \type x}$.
So,
given any term $\term{y}{r}{x}$ and any topotypes $V \type y$ and $U \type x$,
(1) the term $r$ and its source decomposition $\source{r}{V}$ are expressible
in terms of each other via the direct and inverse left flow expressions
$\source{r}{V} = \iota_V \circ \singleton{r}
               = \pi_V \tensorimplysource \singleton{r}$
and
$\singleton{r} = \pi_V \circ \source{r}{V}
               = \iota_V \tensorimplysource \source{r}{V}$,
and
(2) the term $r$ and its target decomposition $\target{r}{U}$ are expressible
in terms of each other via the direct and inverse right flow expressions
$\target{r}{U} = \singleton{r} \circ \pi_U
               = \singleton{r} \tensorimplytarget \iota_U$
and
$\singleton{r} = \target{r}{U} \circ \iota_U
               = \target{r}{U} \tensorimplytarget \pi_U$.
Furthermore,
given any two topotypes $V \type y$ and $U \type x$,
(1) a term $\term{y}{r}{x}$ and its decomposition matrix $\#_{V,U}(r) = \decomposition{r}{V}{U}$
are expressible in terms of each other via the direct flow expressions
$r = \pi_V \circ \#_{V,U}(r) \circ \iota_U$
and
$\#_{V,U}(r) = \iota_V \circ \singleton{r} \circ \pi_U$,
and
(2) an {\bf H}-topomatrix $\term{V \type y}{R}{U \type x}$ and its join term 
$\term{y}{{\displaystyle \vee} R}{x}$
are expressible in terms of each other via the direct flow expressions 
$R = \iota_V \circ \singleton{\bigvee R} \circ \pi_U$
and
$\bigvee R = \pi_V \circ R \circ \iota_U$.

For each topotype $U \type x$
the topomatrix isomorphism $\term{\singleton{x}}{\pi_U}{U \type x}$ 
is the $(U \type x)$-th component of a ``counit'' natural isomorphism
$\pi \type \bigvee \cdot \singleton{} \Longrightarrow {\rm Id}_{\matropo{H}}$,
since $\singleton{\bigvee R} \circ \pi_U = \pi_V \circ R$.
\begin{Theorem}
   For every cHc {\bf H},
   the indiscrete-topology and join functors form a categorical equivalence
   $\singleton{} \dashv \bigvee$
   between {\bf H} and its category of topomatrices $\matropo{{\bf H}}$,
   with identity unit ${\rm Id}_H = \singleton{} \cdot \bigvee$
   and natural isomorphism counit
   $\pi \type \bigvee \cdot \singleton{} \Longrightarrow {\rm Id}_{\matropo{H}}$.
\end{Theorem}

Given three topotypes $W \type z$, $V \type y$ and $U \type x$
and two terms $\term{z}{s}{y}$ and $\term{y}{r}{x}$,
the $(w,u)$-th subterm $(s \circ r)_{wu}$ is the join
$(s \circ r)_{wu} = \bigvee_{v \in V} s_{wv} \circ r_{vu}$,
so that decomposition maps tensor products of terms to products of matrices
$\decomposition{s}{W}{V} \circ \decomposition{r}{V}{U} = \decomposition{s \circ r}{W}{U}$.
Also,
the $\product{U}{U}$ decomposition matrix of the identity term $\term{x}{x}{x}$
is the identity matrix
$\decomposition{x}{U}{U} = \iota_U \circ \pi_U$,
where ${\decomposition{x}{U}{U}}_{u'u} = u' \circ u = u' \wedge u$.
The type $x$ is a direct sum of $V$-open comonoids when $x = \bigvee X$
for some collection $X \subseteq V$ of pairwise disjoint comonoids.

Let {\bf W} be a standard topology on the lattice of all {\bf H}-comonoids $\comonoid{}{{\bf H}}{}$.
{\bf W} can be partitioned into a collection of topotypes
${\bf W} = \{ {\bf W}(x) \subseteq \comonoid{}{x}{} \mid x \memberof {\rm Obj}({\bf H}) \}$.
We call such a collection {\bf W} a {\em topotypeal structure\/}.
A topotypeal structure is a ``choice functor'',
choosing a topology at each {\bf H}-type.
Topotypeal structures are a type-indexed version of Girard's topolinear spaces in linear logic.
Any topotypeal structure {\bf W} defines,
and can be identified with,
an embedding ${\bf H} \stackrel{\#_W}{\rightarrow} \matropo{{\bf H}}$,
of {\bf H} into its category of topomatrices $\matropo{{\bf H}}$
called the {\bf W}-{\em decomposition\/} of terms.
On types $\#_W = {\bf W}(x)$ is the $x$-th topotype of {\bf W},
and on terms $\#_W(r) = \decomposition{r}{W(y)}{W(x)}$
is the $\product{{\bf W}(y)}{{\bf W}(x)}$ decomposition matrix of $r$. 
Partition followed by join is the identity functor
$\#_W \cdot \vee = {\rm Id}_H$.
The indiscrete-topology inclusion functor
${\bf H} \stackrel{\singleton{}}{\longrightarrow} \matropo{{\bf H}}$
is the decomposition functor
$\singleton{} = \#_\triangle$
for the trivial topotypeal structure
$\triangle = \{ \{ \bot,x \} \subseteq \comonoid{}{x}{} \mid x \memberof {\rm Obj}({\bf H}) \}$.
For any topotypeal structure ${\bf W}$,
the {\bf W}-{\em decomposition category\/} $\matropo{{\bf W}} \subseteq \matropo{{\bf H}}$,
is the full subcategory which is the image of the ${\bf W}$-decomposition functor $\#_W$.
There is a {\bf W}-{\em join functor\/} $\matropo{{\bf W}} \stackrel{{\displaystyle \vee}_W}{\rightarrow} {\bf H}$
which is the restriction of join $\bigvee$ to ${\bf W}$-matrices $\matropo{{\bf W}}$,
and a {\bf W}-{\em decomposition functor\/} ${\bf H} \stackrel{\#_W}{\rightarrow} \matropo{{\bf W}}$
which is the corestriction of ${\bf W}$-decomposition $\#_W$ to ${\bf W}$-matrices $\matropo{{\bf W}}$.
For a fixed topotypeal structure ${\bf W}$,
these decomposition and join functors are inverse to each other.
\begin{Theorem}
   Any cHc {\bf H} is isomorphic to each of its decomposition categories:
   ${\bf H} \cong \matropo{{\bf W}}$
   for any topotypeal structure {\bf W}.
\end{Theorem}
So each topotypeal structure {\bf W} defines a representation of the cHc {\bf H}
inside of its category of topomatrices $\matropo{{\bf H}}$;
namely,
$\matropo{{\bf W}}$.

\paragraph{Flow Decomposition.}
For any cHc {\bf H},
in the category of {\bf H}-topomatrices $\matropo{{\bf H}}$
source and target tuplings are related to direct and inverse flow by the identities
   \begin{center}
      $\begin{array}{r@{\;=\;}lp{3.3in}}
          \tuple{\left( t_{zv} \mid v \memberof V \right)}{V} \circ \cotuple{\left( r_{vx} \mid v \memberof V \right)}{V}
             & \bigvee_{v \in V} \left( t_{zv} \circ r_{vx} \mid v \memberof V \right)
             & ``{\em right tensor product along $V$-source tupling}'' \\
          t \circ \tuple{\left( r_{yu} \mid u \memberof U \right)}{U}
             & \tuple{\left( t \circ r_{yu} \mid u \memberof U \right)}{U}
             & ``{\em right tensor product along $U$-target tupling}'' \\
          \cotuple{\left( r_{vx} \mid v \memberof V \right)}{V} \circ s
             & \cotuple{\left( r_{vx} \circ s \mid v \memberof V \right)}{V}
             & ``{\em left tensor product along $V$-source tupling}'' \\
          \tuple{\left( r_{yu} \mid u \memberof U \right)}{U} \circ \cotuple{\left( s_{uz} \mid u \memberof U \right)}{U}
             & \bigvee_{u \in U} \left( r_{yu} \circ s_{uz} \mid u \memberof U \right)
             & ``{\em left tensor product along $U$-target tupling}'' \\
          s \tensorimplytarget \cotuple{\left( r_{vx} \mid v \memberof V \right)}{V}
             & \tuple{\left( s \tensorimplytarget r_{vx} \mid v \memberof V \right)}{V}
             & ``{\em right tensor implication along $V$-source tupling}'' \\
          \tuple{\left( s_{zu} \mid u \memberof U \right)}{U} \tensorimplytarget \tuple{\left( r_{yu} \mid u \memberof U \right)}{U}
             & \bigwedge_{u \in U} \left( s_{zu} \tensorimplytarget r_{yu} \mid u \memberof U \right)
             & ``{\em right tensor implication along $U$-target tupling}'' \\
          \cotuple{\left( r_{vx} \mid v \memberof V \right)}{V} \tensorimplysource \tuple{\left( t_{vz} \mid v \memberof V \right)}{V}
             & \bigwedge_{v \in V} \left( r_{vx} \tensorimplytarget t_{vz} \mid v \memberof V \right)
             & ``{\em left tensor implication along $V$-source tupling}'' \\
          \tuple{\left( r_{yu} \mid u \memberof U \right)}{U} \tensorimplysource t
             & \cotuple{\left( r_{yu} \tensorimplysource t \mid u \memberof U \right)}{U}
             & ``{\em left tensor implication along $U$-target tupling}''
       \end{array}$
   \end{center}
These identities reduce the action of direct and inverse term flow to components.

\section{Dialectical Reproduction.}

We work in a Heyting category {\bf H},
and assume the existence of a special type $1$ which is a {\em separator\/} of terms in the following sense:
for any two parallel terms $\term{y}{s,r}{x}$,
if $\psi \circ s = \psi \circ r$ for all terms $\term{1}{\psi}{y}$ then $s = r$.
A term $\term{1}{\phi}{x}$ is called an {\em object\/} of type $x$,
and denoted by $\phi \objtype x$.
In relational database theory,
where the Heyting category {\bf H} is the category
of monoids and processes \cite{Kent88}
of closed subsets of $\Sigma$-terms,
a monoid $m \type x$ ({\bf H}-type) represents a constrained database scheme consisting of database scheme $x$ and semantic constraints $m$,
and an $m$-object is a database which satisfies that scheme and those semantic constraints.
In the general theory of dialectics,
two possible meanings for ``entities in dialectical motion'' are
(1) {\em comonoids\/} $u \memberof \comonoid{}{x}{}$; and
(2) {\em objects\/}   $\term{1}{\phi}{x}$.
Here we discuss the flow of objects in more detail.
In a succeeding paper \cite{Kent89} we will discuss the flow of comonoids,
and we will also discuss the important notion of transformation between these two kinds of entities.

Let ${\bf Obj}(x)$ denote the lattice of all objects of type $x$
with object order $\preceq_x \define \preceq_{1,x}$;
that is,
${\bf Obj}(x) = {\bf H}[1,x]$.
Terms define a dialectical (bidirectional) flow of objects which is expressed in terms of tensor product and implication:
for any term $\term{y}{r}{x}$
    let ${\bf Obj}^r = (\,) \circ r$ denote right tensor product by $r$,
and let ${\bf Obj}_r = (\,) \tensorimplytarget r$ denote right tensor implication by $r$.
So ${\bf Obj}^r$ is the right direct flow and ${\bf Obj}_r$ is the right inverse flow of $r$. 
We identify this dialectical flow of objects as the {\em behavior\/} of the term $r$.
The separator rule states that terms are distinguished (and can be identified) by their direct flow behavior.
Direct flow $\scriptbf{Obj}(y) \stackrel{\scriptbf{Obj}^r}{\longrightarrow} {\bf Obj}(x)$
and inverse flow ${\bf Obj}(y) \stackrel{\scriptbf{Obj}_r}{\longleftarrow} {\bf Obj}(x)$
are monotonic functions,
and the dialectical axioms state that these form an adjoint pair ${\bf Obj}^r \dashv {\bf Obj}_r$. 
As noted before direct flow is ``functorial'',
${\bf Obj}^{s \circ r} = {\bf Obj}^s \cdot {\bf Obj}^r$ 
and
${\bf Obj}^x = {\rm Id}_{\scriptbf{Obj}(x)}$,
and inverse flow is ``contravariantly functorial'',
${\bf Obj}_{s \circ r} = {\bf Obj}_r \cdot {\bf Obj}_s$
and
${\bf Obj}_x = {\rm Id}_{\scriptbf{Obj}(x)}$.
In summary,
if we combine the adjoint pairs as
${\bf Obj}(r) = ({\bf Obj}^r \dashv {\bf Obj}_r)$,
then the above laws and rules are equivalent to the statement that
the object concept or {\em flow dialectic\/} is functorial
${\bf H} \stackrel{\scriptbf{Obj}}{\longrightarrow} {\bf adj}$,
mapping types to their object lattice and terms to their behavior.
This is the sense in which terms specify the dialectical motion of objects.

So tensor product defines the {\em direct aspect\/} of term flow,
whereas tensor implication defines the {\em inverse aspect\/}.
As is clear now (manifested by the doubling of implication) and more clear latter
(however, see Kelley's development of tensors using hom-objects),
the direct aspect of flow is the principal aspect.
This notion of principal aspect seems to occur often in applied dialectics. 
We develop here the full theory of dialectical terms.
However, an interesting and coherent {\em direct subtheory\/} of terms,
using only the direct aspect of flow,
is included.
This direct subtheory seems to include much of traditional process theory,
but is impoverished by not having the concept of inverse flow.

Since the behavior of terms is identified with (dialectical) flow,
either direct flow or inverse flow,
one means of interaction/communication between terms is by flow composition.
If we make the identification ``types $\equiv$ ports'',
then terms communicate through their source and/or target ports.
A parallel pair of terms $\term{y}{s,r}{x}$,
a graph in a Heyting category,
is known as a {\em dialectical system\/}. 
The dialectical interaction (complementary union)
of the component terms of a dialectical system
occurs through both source and target ports.
The notion of {\em reproduction\/} in a system is specified by the dialectical flow (fixpoint operator)
$\yinyang_r^s(\,) = ((\,) \tensorimplytarget r) \circ s$.
This reproduction operator can be interpreted as the ``polar-turning structure'' of the preSocratic Greek philosopher Heraclitus \cite{Hussey},
and in Greek is rendered 
$\pi \alpha \lambda \iota \nu \tau \rho o \pi o \zeta$ $\alpha \rho \mu o \nu \iota \eta$.
An object $\phi$ is {\em reproduced\/} when it satisfies the fixpoint equation $\yinyang_r^s(\phi) = \phi$.
[A philosophical note:
The notion of complementary union
(two working together in one)
is not that of ``synthesis''.
Neither of the opposites is ``transformed''.
Indeed,
with synthesis,
dialectical motion would cease! 
The notion of ``reproduction'' is one of equilibrium of motion,
not lack of motion.]
Here the yin-yang symbol $\yinyang_r^s$ is used as a reminder of ancient dialectics;
{\em yin\/} inverse flow along $r$ and {\em yang\/} direct flow along $s$.
Starting with (quotient) objects at the source type,
there is a op-dual ``reverse time'' yin-yang fixpoint operator $(s \circ (r \tensorimplysource (\,))$.
There are also yang-yin operators with direct flow first and reverse flow last.
To claim a type of uniqueness for reproduced objects $\phi$ we can use:
the {\bf least fixpoint rule} 
$\yinyang_r^s(\phi) = \phi$, and if $\yinyang_r^s(t) = t$ then $\phi \preceq t$;
or the {\bf greatest fixpoint rule}
$\yinyang_r^s(\phi) = \phi$, and if $\yinyang_r^s(t) = t$ then $t \preceq \phi$.
The system motion is graphically represented as follows:
\begin{center}
   \begin{picture}(150,70)(-75,0)
      \put(-20,10){\vector(1,0){0}}
      \put(-20,35){\oval(50,50)[l]}
      \put(20,60){\vector(-1,0){0}}
      \put(20,35){\oval(50,50)[r]}
      \put(-45,35){\makebox(0,0){\shortstack{direct\\motion}}}
      \put(45,35){\makebox(0,0){\shortstack{inverse\\motion}}}
      \put(0,35){\makebox(0,0){\shortstack{proper\\motion}}}
      \put(0,60){\makebox(0,0){${\bf Obj}(d\!d)$}}
      \put(0,5){\makebox(0,0){${\bf Obj}(x)$}}
      \put(-75,0){\vector(1,0){55}}
      \put(-51,0){\makebox(0,0){\shortstack{continual\\input}}}
      \put(-100,0){\makebox(0,0){${\bf Obj}(k\!d)$}}
      \put(20,0){\vector(1,0){55}}
      \put(50,0){\makebox(0,0){\shortstack{continual\\output}}}
      \put(100,0){\makebox(0,0){${\bf Obj}(d\!k)$}}
   \end{picture}
\end{center}
where the collection of $y$-subtypes
$k\!d$, $d\!d$, $d\!k$ and $k\!k$ consists of,
respectively,
the ``atomic subtype'', ``proper subtype'', ``negative subtype'' and ``nil subtype'' 
of the source type $y$.
These correspond to clause types in Horn clause logic.

For any term $\term{y}{r}{x}$,
dialectical flow along $r$ is decreasing:
$\yinyang_r^r(\phi) \preceq \phi$ for every object $\term{1}{\phi}{x}$.
For any functional term $\term{y}{f \dashv f^{\rm op}}{x}$,
dialectical flow along $f$ is equal to dialectical flow along the associated interior comonoid $f^{\rm op} \circ f$,
$\yinyang_f^f = \yinyang_{f^{\rm op} \circ f}^{f^{\rm op} \circ f}$,
since $(\,) \circ f^{\rm op} = (\,) \tensorimplytarget f$ implies
$\yinyang_{f^{\rm op} \circ f}^{f^{\rm op} \circ f}
 = [(\,) \tensorimplytarget f] \cdot [(\,) \tensorimplytarget f^{\rm op}] \cdot [(\,) \circ f^{\rm op}] \cdot [(\,) \circ f]
 = [(\,) \circ f^{\rm op}] \cdot [(\,) \tensorimplytarget f^{\rm op}] \cdot [(\,) \circ f^{\rm op}] \cdot [(\,) \circ f]
 \preceq (\succeq) [(\,) \circ f^{\rm op}] \cdot [(\,) \circ f]
 = [(\,) \tensorimplytarget f] \cdot [(\,) \circ f] 
 = \yinyang_f^f$. 
This fact includes subtypes as a special case.
So for dialectical flow along functional terms,
we can restrict our attention to comonoids.
Let $V \type y$ be any topotype (topology of comonoids at $y$).
The join of the dialectical flows of the topotype comonoids is unity
$\bigvee_{v \in V} \yinyang_v^v = {\rm Id}$,
since
$\psi
 = \psi \circ y
 = \psi \circ (\bigvee_{v \in V} v)
 = \bigvee_{v \in V} (\psi \circ v)
 = \bigvee_{v \in V} (\psi \tensorimplytarget y) \circ v
 = \bigvee_{v \in V} (\psi \tensorimplytarget (\bigvee_{v' \in V} v')) \circ v
 = \bigvee_{v \in V} (\bigwedge_{v' \in V} (\psi \tensorimplytarget v')) \circ v
 \preceq \bigvee_{v \in V} \bigwedge_{v' \in V} ((\psi \tensorimplytarget v') \circ v)
 \preceq \bigvee_{v \in V} ((\psi \tensorimplytarget v) \circ v)
 \preceq \bigvee_{v \in V} \psi
 = \psi$
for every $y$-object $\term{1}{\psi}{y}$. 
\begin{Fact}
   For any dialectical system $\term{y}{s,r}{x}$ and any source topotype $V \type y$,
   dialectical flow decomposes as
   \[ \yinyang_r^s = \bigvee_{v \in V} \yinyang_{r_v}^{s_v} . \]
\end{Fact}
\begin{proof}
   $\bigvee_{v \in V} \yinyang_{r_v}^{s_v}
    = \bigvee_{v \in V} [(\,) \tensorimplytarget (v \circ r)] \cdot [(\,) \circ (v \circ s)]
    = \bigvee_{v \in V} [(\,) \tensorimplytarget r] \cdot [(\,) \tensorimplytarget v] \cdot [(\,) \circ v] \cdot [(\,) \circ s]
    = \bigvee_{v \in V} [(\,) \tensorimplytarget r] \cdot \yinyang_v^v \cdot [(\,) \circ s]
    = [(\,) \tensorimplytarget r] \cdot (\bigvee_{v \in V} \yinyang_v^v) \cdot [(\,) \circ s]
    = [(\,) \tensorimplytarget r] \cdot [(\,) \circ s]
    = \yinyang_r^s$.
\end{proof}
This is an abstraction of the {\tt AND}-process decomposition of clausal logic programs.



 \end{document}